\documentclass[a4paper,UKenglish,cleveref, autoref, thm-restate]{lipics-v2021}

\pdfoutput=1 

\newcommand{\longversion}[1]{#1}\newcommand{\shortversion}[1]{}

\hideLIPIcs  
\nolinenumbers 

\usepackage{amsmath,amssymb,amstext,mathtools}
\usepackage{mdframed}
\usepackage{comment}
\usepackage[noadjust]{cite}

\newtheorem{thm}{Theorem}[section]
\newtheorem{lem}[thm]{Lemma}
\newtheorem{prop}[thm]{Proposition}
\newtheorem{obs}[thm]{Observation}

\newtheorem{clm}[thm]{Claim}
\newtheorem{defn}[thm]{Definition}

\crefname{thm}{Theorem}{Theorems}
\crefname{lem}{Lemma}{Lemmas}
\crefname{prop}{Proposition}{Propositions}
\crefname{clm}{Claim}{Claims}

\DeclarePairedDelimiter{\abs}{\lvert}{\rvert}

\newcommand{\NP}{\ensuremath{\mathsf{NP}}}
\newcommand{\W}[1]{\ensuremath{\mathsf{W[#1]}}}
\newcommand{\FPT}{\ensuremath{\mathsf{FPT}}}

\newcommand{\bigO}{O}

\newcommand{\tw}{\textup{tw}}
\newcommand{\ETH}{\textup{ETH}}
\newcommand{\clique}{\textup{\textsc{Clique}}}
\newcommand{\gt}{\textup{\textsc{Grid Tiling}}}

\newcommand{\mwc}{\textup{\textsc{Multiway Cut}}}
\newcommand{\mc}{\textup{\textsc{Multicut}}}
\newcommand{\mcut}[1]{\textup{\textsc{Multicut}(\ensuremath{#1})}}
\newcommand{\tcut}[1]{\textup{#1\textsc{-Terminal Cut}}}

\newcommand{\gcut}{\textup{\textsc{Group 3-Terminal Cut}}}
\newcommand{\csp}{\textup{CSP}}

\newcommand{\KUL}{K_{UL}}
\newcommand{\KUR}{K_{UR}}
\newcommand{\KDL}{K_{DL}}
\newcommand{\KDR}{K_{DR}}

\newcommand{\calH}{\mathcal{H}}
\newcommand{\calR}{\mathcal{R}}

\newcommand{\mcH}{\mcut{\calH}}

\newcommand{\go}{\vec{g}}

\newcommand{\aCA}{\alpha}

\title{
	Multicut Problems in Embedded Graphs:\\
	The Dependency of Complexity on the Demand Pattern}

\titlerunning{Multicut Problems in Embedded Graphs Depending on the Demand Pattern} 

\authorrunning{J.Focke, F.H\"orsch, S.Li, and D.Marx} 

\author{Jacob Focke}{CISPA Helmholtz Center for Information Security, Germany}{jacob.focke@cispa.de}{https://orcid.org/0000-0002-6895-755X}{}
\author{Florian H\"orsch}{CISPA Helmholtz Center for Information Security, Germany}{florian.hoersch@cispa.de}{https://orcid.org/0000-0002-5410-613X}{}
\author{Shaohua Li}{School of Computer Science and Engineering, Central South University, Changsha, China}{shaohua.li@csu.edu.cn}{}{}
\author{D\'aniel Marx}{CISPA Helmholtz Center for Information Security, Germany}{marx@cispa.de}{https://orcid.org/0000-0002-5686-8314}{}

\Copyright{Jacob Focke and Florian H\"orsch and Shaohua Li and D\'aniel Marx} 



\usepackage[normalem]{ulem}

\begin{document}

\maketitle

\begin{abstract}
  The \textsc{Multicut} problem asks for a minimum cut separating certain pairs of vertices: formally, given a graph $G$ and a demand graph $H$ on a set $T\subseteq V(G)$ of terminals, the task is to find a minimum-weight set $C$ of edges of $G$ such that whenever two vertices of $T$ are adjacent in $H$, they are in different components of $G\setminus C$. Colin de Verdi{\`{e}}re [\textit{Algorithmica,} 2017]  showed that \textsc{Multicut} with $t$ terminals on a graph $G$ of genus $g$ can be solved in time $f(t,g)n^{O(\sqrt{g^2+gt+t})}$. Cohen-Addad et al.~[\textit{JACM}, 2021] proved a matching lower bound showing that the exponent of $n$ is essentially best possible (for  every fixed value of $t$ and $g$), even in the special case of \textsc{Multiway Cut}, where the demand graph $H$ is a complete graph.

  However, this lower bound tells us nothing about other special cases of \textsc{Multicut} such as \textsc{Group 3-Terminal Cut} (where three groups of terminals need to be separated from each other). We show that if the demand pattern is, in some sense, close to being a complete bipartite graph, then \textsc{Multicut} can be solved faster than $f(t,g)n^{O(\sqrt{g^2+gt+t})}$, and furthermore this is the only property that allows such an improvement. Formally, for a class $\mathcal{H}$ of graphs, $\textsc{Multicut}(\mathcal{H})$ is the special case where the demand graph $H$ is in $\mathcal{H}$. For every fixed class $\mathcal{H}$ (satisfying some mild closure property), fixed $g$, and fixed $t$, our main result gives tight upper and lower bounds on the exponent of $n$ in algorithms solving $\textsc{Multicut}(\mathcal{H})$.
\end{abstract}

\longversion{
	\clearpage 
\tableofcontents
\clearpage
}

\section{Introduction}\label{sec:intro}
Computing cuts and flows in graphs is a fundamental problem in theoretical computer science, with algorithms going back to the early years of combinatorial optimization \cite{ek,ford_fulkerson_1956} and significant new developments appearing even in recent years \cite{DBLP:journals/corr/abs-2309-16629,DBLP:conf/stoc/BrandGJLLPS22}. Given two vertices $s$ and $t$ in a weighted graph $G$, an \emph{$s-t$ cut} is a set $S$ of edges such that $G\setminus S$ has no path connecting $s$ and $t$. It is well known that an $s-t$ cut of minimum total weight can be found in polynomial time. However, the problem becomes much harder when generalized to more than two terminals. In the \mwc{} problem, we are given a graph $G$ and a set $T\subseteq V(G)$ of \emph{terminals}, and the task is to find a \emph{multiway cut} $S$ of minimum weight, that is, a set $S\subseteq E(G)$ such that every component of $G\setminus S$ contains at most one terminal. Dahlhaus et al.~\cite{DBLP:journals/siamcomp/DahlhausJPSY94} showed that $\mwc$ is \NP-hard even for three terminals.

The study of algorithms on planar graphs is motivated by the fact that planar graphs can be considered a theoretical model of road networks or other physical networks where crossing of edges can be assumed to be unlikely. While most \NP-hard problems remain \NP-hard on planar graphs, there is often some computational advantage that can be gained by exploiting planarity. In many cases, this advantage takes the form of a square root appearing in the running time of an algorithm solving the problem \cite{DBLP:conf/focs/MarxPP18,ChitnisFHM20,DBLP:conf/focs/FominLMPPS16,DBLP:journals/talg/MarxP22,DBLP:conf/soda/KleinM14,DBLP:conf/icalp/KleinM12,Marx12,DBLP:conf/fsttcs/LokshtanovSW12,ECDV,DBLP:conf/stoc/Nederlof20a,cfklmpps}. \mwc{} remains \NP-hard on planar graphs, but unlike in general graphs, where the problem is hard for three terminals, the problem can be solved in polynomial time  on planar graphs for any fixed number of terminals. Dahlhaus et al.~\cite{DBLP:journals/siamcomp/DahlhausJPSY94}
 gave an algorithm with a running time of the form $f(t)n^{O(t)}$ for $t$ terminals, which was improved to $n^{O(\sqrt{t})}$ by Colin de Verdi\`ere~\cite{ECDV}
  and Klein and Marx~\cite{DBLP:conf/icalp/KleinM12}. The square root in the exponent appears to be best possible: Marx~\cite{Marx12}
   showed that, assuming the Exponential-Time Hypothesis (ETH) \cite{IPZ01}, there is no algorithm for \mwc{} on planar graphs with running time $f(t)n^{o(\sqrt{t})}$ for any computable \longversion{function} $f$.

   Colin de Verdi\`ere~\cite{ECDV} actually showed a much stronger result than just an $f(t)n^{O(\sqrt{t})}$ algorithm for \mwc{} on planar graphs with $t$ terminals. First, the algorithm works not only on planar graphs, but on graphs that can be embedded on a surface  of genus at most $g$ for some fixed constant $g$. \longversion{For the clear statement of upper and lower bounds, we need to be precise with the notion of genus. The classification theorem of closed surfaces states that any connected closed surface can be obtained from the sphere either by attaching handles or by attaching crosscaps to it. An orientable surface of \emph{genus} $g$ is obtained from the sphere by attaching $g$ handles, while a nonorientable surface of \emph{nonorientable genus} $g$ is obtained from the sphere by attaching $g$ crosscaps to it. The \emph{Euler genus} of an orientable surface is twice its genus and the Euler genus of a nonorientable surface is its nonorientable genus. Given a graph $G$, the \emph{(i) orientable genus, (ii) nonorientable genus, (iii) Euler genus} of $G$ is the smallest integer $g$ such that $G$ is embeddable in a surface $N$ that is an (i) orientable surface with genus $g$, (ii) nonorientable surface with nonorientable genus $g$, (iii) arbitrary surface with Euler genus $g$, respectively.  

}The second notable feature of the algorithm of Colin de Verdi\`ere~\cite{ECDV} is that it considers the generalization \mc{} of \mwc,  where instead of requiring that all terminals have to be disconnected from each other, the input contains a demand pattern describing which pairs of terminals have to be disconnected.
\medskip

\begin{center}

\begin{tabular}{|rl|}
  \hline
  \mc &\\
  {\bf Input:}& A weighted graph $G$ 
  \\
  &and a \emph{demand graph} $H$ with $V(H)\subseteq V(G)$ being the \emph{terminals}.\\
		{\bf Output:} &
		A minimum-weight set $S\subseteq E(G)$ such that $u$ and $v$ are in distinct components\\ & of $G\setminus S$, whenever $uv$ is an edge of $H$.\\
                \hline
              \end{tabular}
            \end{center}

          Observe that the special case of $\mc$ with a complete demand graph $H$ is exactly the same as \mwc{} on the set $V(H)$ of terminals. Colin de Verdi\`ere~\cite{ECDV}
          showed that the running time achieved for \mwc{} can be obtained also for \mc{} in its full generality with arbitrary demand pattern\longversion{ $H$}.
         
\begin{thm}[Colin de Verdi\`ere~\cite{ECDV}]\label{alt}
  \mc{} can be solved in time $f(g,t)n^{O(\sqrt{g^2+gt+t})}$ for some \longversion{computable }function $f$, where $g$ is the Euler genus of $G$ and $t$ is the number of terminals.
  \end{thm}

In the following, we adopt the convention that if an instance $(G,H)$ of $\mc$ is clear from the context, then we use $t=|V(H)|$ for the number of terminals, $g$ for the Euler genus of $G$, $\go$ for the orientable genus of $G$, and $n=|V(G)|$ for the number of vertices of $G$. \longversion{Statements about the planarity or the genus of the instance refer to the graph $G$.} Furthermore, if not otherwise specified, a function $f$ appearing in the running time or other bounds means that there exists a computable function $f$ that can serve as the runtime bound.

  The exponent $O(\sqrt{g^2+gt+t})$ in \cref{alt} might look somewhat arbitrary. However, Cohen-Addad et al.~\cite{Cohen-AddadVMM21}
  showed that this form of the exponent is optimal under ETH.

  \begin{thm}[Cohen-Addad et al.~\cite{Cohen-AddadVMM21}]\label{thm:mc-tight}
Assuming the $\ETH$, there exists a universal constant $\aCA \ge 0$ such that for any
fixed choice of integers $\go\ge 0$ and $t\ge 4$, there is no algorithm that decides all unweighted \mwc{} 
instances with orientable genus at most $\go$ and at most $t$ terminals, in time $O(n^{\aCA\sqrt{\go^2+\go t + t}/\log(\go+t)})$.
\end{thm}

Note that saying orientable genus in \cref{thm:mc-tight} makes the lower bound statement stronger: if we consider Euler genus instead, then we would consider a parameter that is always at most two times larger than orientable genus, hence we would obtain the same lower bound (with a factor-2 loss in the exponent, which can be accounted for by reducing $\aCA$ by a factor of 2). In a similar way, the statement implies an analogous lower bound for nonorientable surfaces\longversion{: nonorientable genus of a graph is known to be at most twice its orientable genus}. Subsequent lower bounds (for \mwc{} and \mc) will \longversion{also }be formulated in terms of orientable genus, as this gives the strongest statements.

\longversion{
\cref{thm:mc-tight} shows that the exponent in \cref{alt} is essentially optimal, up to a logarithmic factor. The loss of this logarithmic factor comes from using the lower bounds of Marx~\cite{DBLP:journals/toc/Marx10}
 for sparse CSPs. Results based on this technique all have such a logarithmic factor loss in the exponent and currently lower bounds of this form are the best known in the literature for many problems \cite{DBLP:journals/talg/MarxP22,DBLP:journals/jcss/JansenKMS13,DBLP:conf/stoc/CurticapeanDM17,DBLP:journals/siamdm/JonesLRSS17,DBLP:conf/focs/CurticapeanX15,DBLP:conf/esa/BonnetM16,DBLP:journals/algorithmica/GuoHNS13,DBLP:journals/toct/PilipczukW18a,DBLP:journals/dam/BonnetS17,DBLP:conf/esa/Bringmann0MN16,DBLP:journals/corr/abs-1808-02162,DBLP:conf/soda/LokshtanovR0Z20,DBLP:journals/corr/abs-1802-08189,DBLP:conf/iwpec/BonnetGL17,ChitnisFHM20,DBLP:journals/talg/ChitnisFM21,DBLP:journals/algorithmica/AgrawalPSZ21,DBLP:conf/esa/CurticapeanDH21}.
}
 
\subsection{Our Contributions}

\subparagraph*{Bounded-genus graphs.}
While \cref{thm:mc-tight} appears to give a fairly tight lower bound, showing the optimality of \cref{alt} in every parameter range, it has its shortcomings. First, as a minor issue, it works only for $t\ge 4$ (this will be resolved in the current paper). More importantly, while \cref{alt} solves the \mc{} problem in its full generality, \cref{thm:mc-tight} gives a lower bound for the special case \mwc{} (that is, when $H$ is a clique). This is a general issue when interpreting tight bounds: algorithmic results solve all instances of a problem, while lower bounds show the hardness of a specific type of instances. This leaves open the possibility that there exist nontrivial classes of instances (``islands of tractability'') that can be solved significantly faster than the known upper bound. In our case, \cref{thm:mc-tight} shows that \cref{alt} is optimal in the case when the demand graph $H$ is a clique. The main goal of this paper is to understand whether \cref{alt} is optimal for other classes of patterns as well.

\longversion{
\medskip
\begin{mdframed}[backgroundcolor=gray!20]
  Are there nontrivial classes of demand pattern graphs where the algorithm of \cref{alt} is not optimal, that is, the exponent can be better than $O(\sqrt{g^2+gt+t})$?
\end{mdframed}
\medskip}
\shortversion{
  \bigskip

  \begin{tabular}{|c|}
    \hline \normalsize Are there nontrivial classes of demand pattern graphs where the algorithm of\\ \cref{alt}
\normalsize    is not optimal, that is, the exponent can be better than $O(\sqrt{g^2+gt+t})$?\\
    \hline
  \end{tabular}
  \bigskip
  
  }

There is a simple, trivial case when \cref{alt} is not optimal: if the demand graph is a biclique (complete bipartite graph), then the \mc{} instance can be solved in polynomial time (as then the task reduces to separating two sets $T_1$ and $T_2$ of terminals from each other). But what happens if we consider, say, demand patterns that are complete tripartite graphs? This special case is also known as \gcut: given a graph $G$ with three sets of terminals $T_1$, $T_2$, $T_3$, the task is to find a set $S$ of edges of minimum total weight such that $G\setminus S$ has no path between any vertex of $T_i$ and any vertex of $T_j$ for $i\neq j$. Does the lower bound of \cref{thm:mc-tight} for \mwc{} hold for \gcut{} as well? Previous work does not give an answer to this question. One can imagine many other classes of demand patterns where we have no answer yet, for example, patterns obtained as the disjoint union of a biclique and a triangle, which corresponds to a very slight generalization of $s-t$ cut. Our goal is to exhaustively investigate \emph{every} such pattern class and provide an answer for each of them.

Formally, let $\calH$ be a class of graphs\longversion{, that is, a set of graphs closed under isomorphism.}\shortversion{.} Then we denote by \mcH{} the special case of \mc{} that contains only instances $(G,H)$ with $H\in \calH$. For example, if $\calH$ is the class of cliques, then \mcH{} is exactly \mwc. The formal question that we are investigating is if, for a given $\calH$,  \cref{alt} gives an optimal algorithm for \mcH\longversion{, or in other words, whether \cref{thm:mc-tight} holds for \mcH{} instead of \mwc}.

Let us first ask a simpler qualitative question: in an algorithm for \mcH, do parameters $g$ and/or $t$ have to appear in the exponent of $n$? The notion of \W{1}-hardness is an appropriate tool to answer this question. A problem is said to be \emph{fixed-parameter tractable} (FPT) with respect to some parameter $k$ appearing in the input if there is an algorithm with running time $f(k)n^{O(1)}$ for some computable function $f$ \cite{cfklmpps}. If a problem is \W{1}-hard, then this is interpreted as strong evidence that it is not FPT, which means that the parameter has to appear in the exponent of $n$. \longversion{More formally, a \W{1}-hard problem is not FPT, assuming $\FPT\neq \W{1}$ (which is a weaker assumption than ETH).}

If $H$ is a biclique, then \mc{} corresponds to the problem of separating two sets of terminals, which can be reduced to a minimum $s-t$ cut problem. If $H$ has only two edges, then a similar reduction is known~\cite{Hu63,DBLP:journals/jct/Seymour79a}. Isolated vertices of $H$ clearly do not play any role in the problem. We say that $H$ is a \emph{trivial pattern} if, after removing isolated vertices, it is either a biclique or has at most two edges.
\longversion{
\begin{obs}\label{obs:poly}
	If every graph in $\calH$ is a trivial pattern then $\mcH$ is polynomial-time solvable.
\end{obs}
}\shortversion{\addtocounter{thm}{1}}
We show that in case of nontrivial patterns, the genus has to appear in the exponent. \longversion{We state the lower bound for orientable genus, which, as discussed earlier, makes the result stronger.}
\begin{restatable}{thm}{Wgenus}
\label{thm:Wgenus}
  Let $\calH$ be a computable class of graphs.
  \begin{enumerate}
  \item If every graph in $\calH$ is a trivial pattern, then \mcH{} is polynomial-time solvable.
    \item Otherwise, \longversion{unweighted }\mcH{} is \W{1}-hard parameterized by $\go$. This holds even when restricting to instances for which the number of terminals is bounded by some constant depending on $\calH$. 
    \end{enumerate}
\end{restatable}

For the parameter $t$, the number of terminals, we find nontrivial cases where $t$ does \emph{not} have to appear in the exponent of the running time. An \emph{extended biclique} is a graph that consists of a complete bipartite graph together with a set of isolated vertices. 
Let $\mu$ be the minimum number of vertices that need to be deleted to make a graph $H$ an extended biclique. Then we say that $H$ has distance $\mu$ to extended bicliques. Moreover, we say that a graph class $\calH$ has \emph{bounded distance to extended bicliques} if there is a constant $\mu$ such that every $H\in\calH$ has distance at most $\mu$ to extended bicliques.

\begin{restatable}{thm}{Wt}\label{thm:Wt}
  Let $\calH$ be a computable class of graphs.
  \begin{enumerate}
  \item If $\calH$ has bounded distance to extended bicliques, then \mcH{} can be solved in time $f(t,g)n^{O(g)}$.
  \item Otherwise, \longversion{unweighted }\mcH{} is \W{1}-hard parameterized by $t$, even on planar graphs.
    \end{enumerate}
  \end{restatable}
\medskip

The lower bounds in \cref{thm:Wgenus} and \cref{thm:Wt} start with the following simple observation. Let $H$ be a graph and let $H'$ be obtained by identifying two nonadjacent vertices $u$ and $v$ in $H$ into a new vertex $w$. We observe that an instance $(G',H')$ of \mc{} can be reduced to some instance $(G,H)$ in a very simple way: let us obtain $G$ by replacing the vertex $w$ in $G'$ with $u$, and connect to $u$ a new vertex $v$ with an edge of sufficiently large weight. \longversion{It is easy to see that the resulting instance $(G,H)$ has a solution of some weight $c$ if and only if the same is true for the instance $(G',H')$.}

\begin{figure}[t]
	\centering
	\includegraphics[scale=0.3]{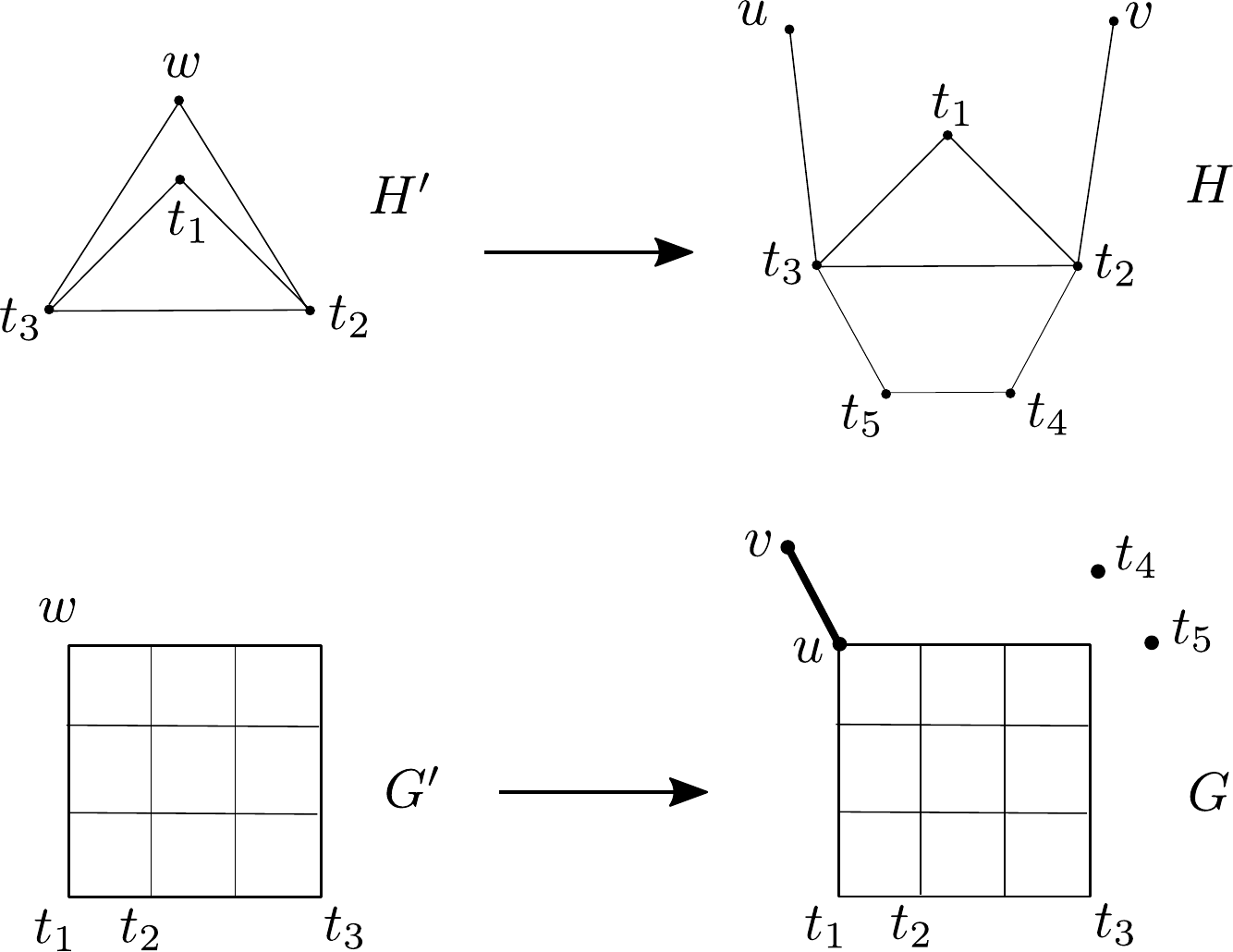}
	\caption{Reducing a \mc{} instance $(G',H')$ to $(G,H)$ where $H'$ is a projection of $H$ that is obtained by identifying the vertices $u$ and $v$ into $w$, and deleting $t_4$ and $t_5$. The corresponding modifications of $G$ are also illustrated. The edge $uv$ is bold to indicate that it has a (sufficiently) large weight.}
	\label{fig:projection}
\end{figure}
Intuitively, this means that the demand pattern $H'$ is not harder than the pattern $H$, so if $\calH$ contains $H$, then we might as well assume that $\calH$ contains $H'$. Formally, this observation is relevant when dealing with parameterized complexity. Namely, it can be shown that adding $H'$ to $\calH$ does not change whether or not \mcH{} is \W{1}-hard. We can argue similarly about the graph $H'$ obtained by removing a vertex $v$ of $H$ (in this case, the reduction simply extends $G'$ with the isolated vertex $v$). We say that $H'$ is a \emph{projection} of $H$ if $H'$ can be obtained from $H$ by deleting vertices and identifying pairs of independent vertices  (see \cref{fig:projection}). Moreover, a class of graphs $\calH$ is \emph{projection-closed} if, whenever $H\in \calH$ and $H'$ is a projection of $H$, then $H'\in\calH$. Due to the observations above, it is sufficient to show the \W{1}-hardness results of
\cref{thm:Wgenus} and \cref{thm:Wt} for projection-closed classes $\calH$. In particular, the proof of \cref{thm:Wt} uses Ramsey-theoretic arguments to show that a projection-closed class either has bounded distance to extended bicliques, or contains all cliques, or contains all complete tripartite graphs. Thus, the hardness results essentially boil down to two cases: cliques (i.e., \mwc) and complete tripartite graphs (i.e., \gcut).

Let us turn our attention now to quantitative lower bounds of the form of \cref{thm:mc-tight}. Unfortunately, the reduction from $(G',H')$ to $(G,H)$ increases the number $t$ of terminals, hence formally cannot be used to obtain lower bounds on the exact dependence on $t$. Nevertheless, one feels that this reduction should be somehow ``free.'' To obtain a more robust setting, we express this feeling by restricting our attention to projection-closed classes. For such classes $\calH$, we obtain tight bounds showing that \longversion{the complexity of }\mcH{} can have two types of behavior, depending on whether $\calH$ has bounded distance to extended bicliques.

\begin{restatable}{thm}{maingenus}\label{thm:maingenus}
  Let $\calH$ be a computable projection-closed class of graphs. Then the following holds for \mcH.
  \begin{enumerate}
  \item If every graph in $\calH$ is a trivial pattern, then there is a  polynomial time algorithm.
  \item Otherwise, if $\calH$ has bounded distance to extended bicliques, then
    \begin{enumerate}
       \item There is an $f(g,t)n^{O(g)}$ time algorithm.
      \item Assuming $\ETH$, there is a universal constant $\alpha>0$ such that for any fixed choice of $\go\ge 0$, there is no $O(n^{\alpha (\go+1)/\log(\go+2)})$ algorithm, even when restricted to unweighted instances with orientable genus at most $\go$ and $t=3$ terminals.
      \end{enumerate} 
    \item Otherwise,
    \begin{enumerate}
	    \item There is an $f(g,t)n^{O(\sqrt{g^2+gt+t})}$ time algorithm.
	    \item Assuming $\ETH$, there is a universal constant $\alpha>0$ such that for any fixed choice of $\go\ge 0$ and $t\ge 3$, there is no $O(n^{\alpha \sqrt{\go^2+\go t+t}/\log(\go+t)})$ algorithm, even when restricted to unweighted instances with orientable genus at most $\go$ and at most $t$ terminals.
    \end{enumerate}
  \end{enumerate}
\end{restatable}
\medskip

\longversion{Note that statement}\shortversion{Statement} 3\,(a) follows from \cref{alt}, while statement 2\,(a) is a new nontrivial algorithmic statement. \longversion{Furthermore, observe}\shortversion{Observe} that in the special case when $\calH$ \longversion{is the class of all cliques,}\shortversion{is all cliques,}  statement 3\,(b) recovers \cref{thm:mc-tight} on $\mwc$, with the strengthening that it works also for $t=3$, which was previously left as an open problem by Cohen-Addad et al.~\cite{Cohen-AddadVMM21}.
\begin{restatable}{cor}{MultiwaymainLB}
\label{cor:3TCutmain}
	Assuming $\ETH$, there exists a universal constant $\alpha>0$ such that for any fixed choice of integers $\go\ge 0$ and $t\ge 3$, there is no algorithm that decides all unweighted $\mwc$ instances with orientable genus at most $\go$ and at most $t$ terminals, in time $\bigO(n^{\alpha \sqrt{\go^2+\go t+t}/\log{(\go+t)}})$.
\end{restatable}

Let us remark that, in \cref{thm:maingenus}, one reason to require the condition that $\mathcal{H}$ be projection-closed is to avoid instances where most terminals are artificial and do not play any role.  For example, consider the class $\mathcal{H}$ that contains all graphs $H$ that consist of a clique of size $\log(|V(H)|)$ and a collection of isolated vertices.
This means that \mcH{} is equivalent to solving $\mwc$ on the clique part and ignoring the isolated vertices. Intuitively, this suggests that \mcH{} should not be easier than $\mwc$. However,
if there are $t=|V(H)|$ terminals, then solving the problem on the clique part means solving an instance of \mwc{} with $\log t$ terminals, which can be done in time $f(t)n^{O(\sqrt{\log t})}$ in planar graphs using \cref{alt}. This running time is much better than the lower bound in \cref{thm:maingenus} 3\,(b), showing that this lower bound is not true without the assumption of $\mathcal{H}$ being projection-closed.

\subsection{Our Techniques}\label{sec:intro:techniques}

Our contributions consist of upper bounds\longversion{ (algorithmic results)}, combinatorial results\longversion{ (showing that every graph class is either ``easy'' or otherwise contains some ``hard structure'')}, and lower bounds\longversion{ (hardness proofs exploiting these hard structures)}. \longversion{In this section, we briefly introduce these results, state them formally, indicate the proof techniques, and show how they can be used to prove \cref{thm:Wgenus,thm:Wt,thm:maingenus}.}

\paragraph*{Algorithmic Results}

For all algorithmic results, we may suppose that $G$ is simple as parallel edges may be replaced by a single edge linking the same two vertices whose weight is the total weight of the deleted parallel edges. We say that an instance $(G,H)$ of $\mc$ with a graph $G$ of Euler genus $g$ is \emph{$N$-embedded} if $G$ is given as a cellularly embedded graph in a surface $N$ of Euler genus $g$, where cellularly embedded means that each face is homeomorphic to an open disk.
Similarly to earlier work \cite{ECDV}, we may assume that we are given $N$-embedded instances: Given a \emph{connected} graph $G$, a cellular embedding in a surface having the same Euler genus a s $g$ exists \cite{Parsons1987} and there are algorithms for finding such an embedding whose running time is dominated by the time bound we want to achieve~\cite{DBLP:conf/focs/KawarabayashiMR08,DBLP:journals/siamdm/Mohar99}. Further, if the input graph consists of several components $G_1,\ldots,G_k$, it suffices to solve, for $i \in [k]$, the instance $(G_i,H[V(H)\cap V(G_i)])$ and output the union of all obtained solutions. 

The main insight behind the algorithmic results for \mc{} is a simple consequence of duality: if $S\subseteq E(G)$ is a solution in a connected graph $G$, then in the dual graph $G^*$ the edges of $S$ create a graph in which the terminals that need to be separated are in different faces (that is, any curve connecting them crosses some edge of $S$ in $G^*$).
This motivates the following definition: a \emph{multicut dual} for $(G,H)$ is a graph $C$ embedded on $N$ that is in general position with $G$ and, for every $t_1,t_2 \in V(H)$ with $t_1t_2 \in E(H)$, it holds that $t_1$ and $t_2$ are contained in different faces of $C$ (see \cref{fig:dual graph}). Here ``general position'' means that the drawings of $C$ and $G$ intersect in finitely many points, none of which is a vertex of either graph and that at every point where an edge of $C$ and edge of $G$ intersect, the two edges cross, i.e., it cannot happen that some edge touches another without actually crossing. We define the \emph{weight $w(C)$} of a multicut dual to be the total weight of the edges of $G$ crossed by edges of $C$. \longversion{One can show that the}\shortversion{The} minimum weight of a multicut for $(G,H)$ is the same as the minimum weight of a multicut dual.

\begin{figure}[t]
	\centering
	\includegraphics[scale=0.5]{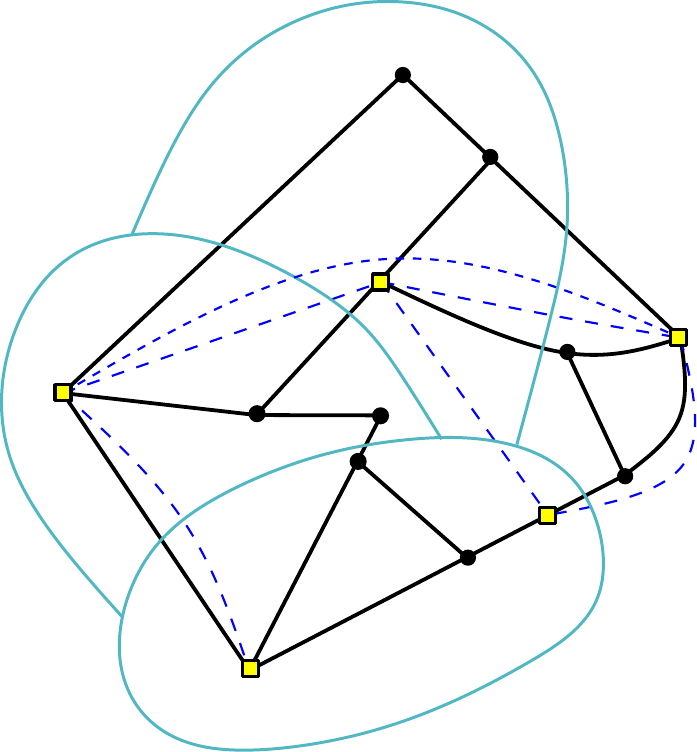}
	\caption{An example of a multicut dual. The yellow squares represent the terminals. The blue dashed curves represent the edges of $H$. The cyan graph represents a multicut dual.} 
	\label{fig:dual graph} 
\end{figure}

The algorithm of \longversion{Colin de Verdi\`ere~\cite{ECDV}
 stated in }\cref{alt} proceeds by guessing the topology of an optimal multicut dual, finding a minimum-weight multicut dual $C$ with this topology, and then arguing that the edges crossed by $C$ form indeed a minimum solution of the \mc{} instance $(G,H)$. It is argued that the multicut dual $C$ has $O(t+g)$ vertices, which implies that it has pathwidth $O(\sqrt{g^2+gt+t})$. The exponent of $n$ in the algorithm is determined by this pathwidth bound. 

One can observe that if there is some other way of bounding the pathwidth of the multicut dual $C$ (perhaps because of some extra conditions on the instance $(G,H)$), then the algorithm can be adapted to run in time with that pathwidth bound in the exponent. Furthermore, by modifying the algorithm, we can also use a bound on treewidth instead of pathwidth.

The following statement can be extracted from the proof of \cref{alt}:

\begin{restatable}{thm}{algogen}
	\label{algogen}
		Let $(G,H)$ be an $N$-embedded instance of \mc{} such that every subcubic inclusionwise minimal minimum-weight multicut dual $C$ satisfies $\tw(C)\leq  \beta$. Then a minimum weight multicut for $(G,H)$ can be computed in time $f(t,g)n^{O(\beta)}$.
              \end{restatable}
              \medskip
              Here ``inclusionwise minimal'' means that we consider only multicut duals $C$ for which removing an edge or vertex does not result in a valid multicut dual.  Notice that the word ``subcubic'' (which means that the maximum degree is at most 3) strengthens \cref{algogen}: it means that the treewidth bound needs to hold for a more restricted set of multicut duals. Observe that in \cref{algogen}, we do not need $C$ to be given with the input, only the existence of this embedded graph needs to be guaranteed.

              It might be surprising that \cref{algogen} requires that every such $C$ has low treewidth: one could expect that the existence of a single low-treewidth $C$ would already help solving the problem efficiently. However, a crucial argument of Colin de Verdi\`ere~\cite{ECDV} takes a multicut dual and modifies it to have a bounded number of intersections with some other structure. It is unclear whether this modification increases treewidth, and hence we need a bound on every relevant multicut dual. Fortunately, in our applications there are very strong structural reasons why every relevant multicut dual should have low treewidth (e.g., reasons related to the number of faces or size of dominating sets).

\longversion{\cref{algogen} does not formally follow from the proof of \cref{alt}
in \cite{ECDV}: a different dynamic programming procedure is needed to
exploit treewidth instead of pathwidth. For completeness, we present a
proof of \cref{algogen} in \cref{algtw}, reusing some of the results
from \cite{ECDV}, but making certain steps and definitions more explicit.}

We believe that formulating \cref{algogen} gives a useful tool: as we will see, it allows more efficient algorithms in certain cases. Note that \cref{algogen} can be used to recover \cref{alt}, the following way.   It can be shown that every face of an inclusionwise minimal multicut contains a terminal\longversion{ (\cref{faceterminal})} and that a graph with at most $t$ faces  drawn on a surface of Euler genus $g$ has treewidth $O(\sqrt{g^2+gt+t})$\longversion{ (\cref{twroot2})}. Thus, if there are $t$ terminals, then every minimum-weight inclusionwise minimal multicut dual has treewidth  $O(\sqrt{g^2+gt+t})$, and then \cref{algogen} gives the required running time.

Our main application of \cref{algogen} is to prove \cref{thm:maingenus}\,(2a), the algorithm for \mcH{} when $\calH$ has bounded distance to extended bicliques.

\begin{restatable}{thm}{mainalgo}
	\label{main}\label{thm:mainalgo}
		Let $\calH$ be a class of graphs whose distance to extended bicliques is at most $\mu$.
		Then $\mcut{\calH}$ can be solved in time $f(t,g)n^{O(\mu+g\mu+1)}$.
              \end{restatable}
              \medskip

              \longversion{We remark that in case of planar graphs (i.e., $g=0$), the exponent of $n$ can be improved to $O(\sqrt{\mu})$ (\cref{thm:mainalgoplanar}).              }

\begin{figure}[t]
	\centering
	\includegraphics[scale=0.4]{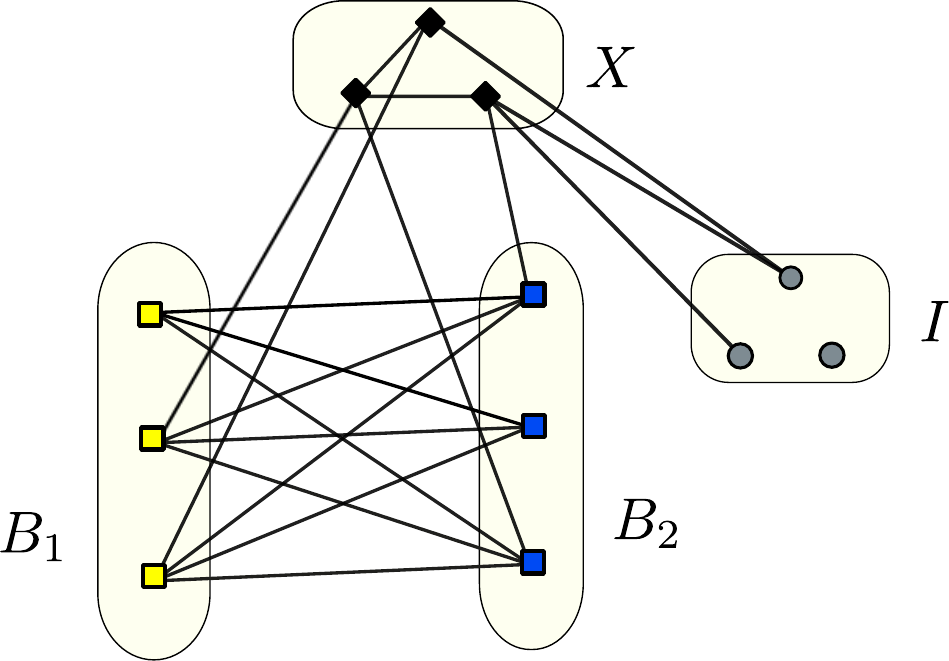}
	\caption{An example of an extended biclique partition.}  
	\label{fig:extended biclique}
\end{figure}

We now give an overview of the proof of \cref{thm:mainalgo}.
An \emph{extended biclique partition} of a graph $H$ is a partition $(B_1,B_2,I,X)$ of its vertices such that $H-X$ is an extended biclique with a complete bipartite graph with bipartition $(B_1,B_2)$ and isolated vertices $I$ (see \cref{fig:extended biclique}).

\begin{figure}[t]
	\centering
	\includegraphics[scale=0.6]{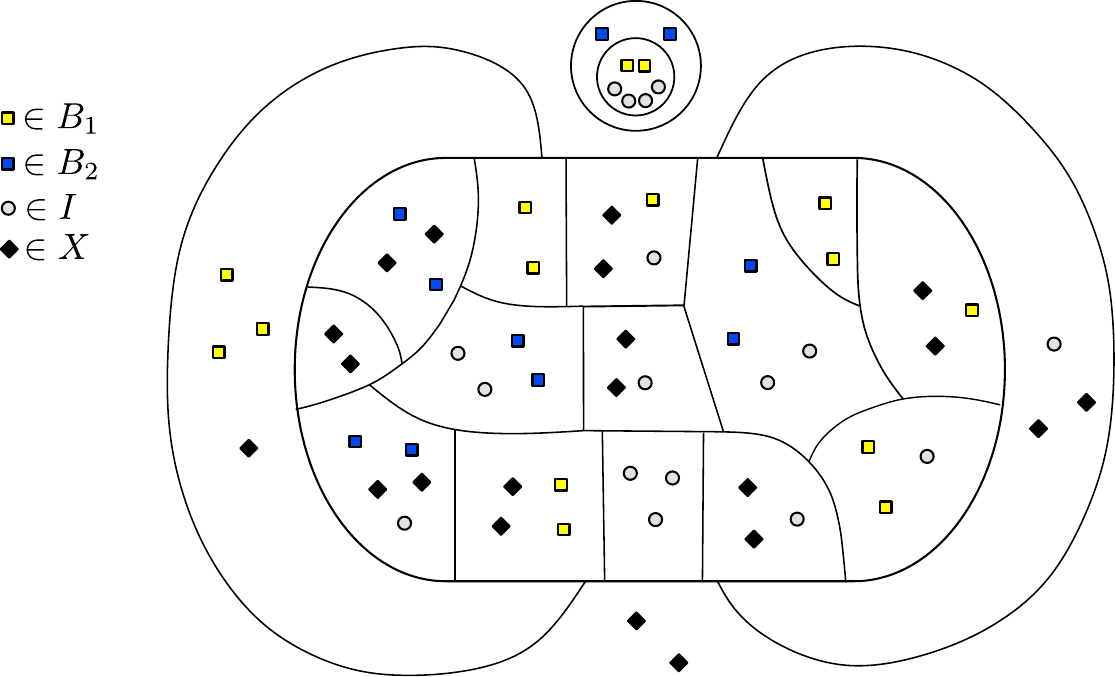}
	\caption{The black graph illustrates a minimal subcubic multicut dual $C$ for an instance $(G,H)$ where the vertices of $H$ are depicted in different shapes and colors according to their role in the extended biclique decomposition. From the minimality of $C$, it follows that, for two adjacent faces, at least one of them must contain a terminal from $X$, or otherwise one face contains a terminal from $B_1$ and the other a terminal from $B_2$. As a consequence, every degree-$3$ vertex of $C$ is incident to a face containing an element of $X$. Note that this is not necessarily true for degree-$2$ vertices, for instance consider the vertices of the inner cycle of the two components at the top of the drawing.}
	\label{fig: multicut dual with eb}
\end{figure}

Consider an instance $(G,H)$ and an extended biclique partition $(B_1,B_2,I,X)$ of $H$ with $|X|$ being at most some constant $\mu$. Let us consider a subcubic, minimum-weight inclusionwise minimal multicut dual $C$ (see \cref{fig: multicut dual with eb}). There are at most $\mu$ faces of $C$ that contain a terminal of $X$.
Let us consider a degree-3 vertex $v$ of $C$ such that none of the three faces incident to $v$ contains a terminal of $X$. For each of these faces, there exists an $i \in [2]$ such that every terminal in the face is in $B_i\cup I$. Thus, there are two of these faces and an $i\in [2]$ such that both faces contain terminals only from $B_i\cup I$. As $H$ has no edge inside $B_i\cup I$, this means that the edge of $C$ between these two faces can be removed and the resulting graph would still be a valid multicut dual. This contradicts the inclusionwise minimality of $C$. Thus, we can conclude that there is a collection of $\mu$ faces of $C$ such that every degree-3 vertex of the subcubic graph $C$ is incident to one of these faces, and hence known bounds imply that $C$ has treewidth $O(\mu+g\mu+1)$.

\begin{restatable}{lem}{mainprop}\label{mainprop}
Let $(G,H)$ be an $N$-embedded instance of \mc{}. Then every subcubic, inclusionwise minimal minimum-weight multicut dual $C$ for $(G,H)$ has $\tw(C)= O(\mu+g\mu+1)$.
\end{restatable}
\medskip

Instead of proving \cref{algogen}, we prove a technical slight generalization that can be useful in certain applications. This extension does not complicate the proof in an essential way. 
An \emph{extended} $N$-embedded instance $(G,H,F^*)$ of \mc{} is an $N$-embedded instance $(G,H)$ of \mc{} together with a set $F^*$ of faces of $G$ which is given with the input. We denote by $\kappa=|F^*|$ the number of these faces. Given a set $F^*$ of faces of $G$ and a multicut dual $C$ for $(G,H)$, let $C-F^*$ be the graph obtained from $C$ by removing every vertex that is in a face of $F^*$ and removing every edge that intersects a face of $F^*$.  We extend \cref{algogen} the following way:

\begin{restatable}{thm}{algogenext}\label{algogenext}
Let $(G,H,F^*)$ be an extended $N$-embedded instance of \mc{} such that every subcubic inclusionwise minimal minimum-weight multicut dual $C$ satisfies $\tw(C-F^*)\leq  \beta$. Then an optimum solution of $(G,H)$ can be found in time $f(t,g,\kappa)n^{O(\beta)}$.
\end{restatable}
\medskip

That is, the algorithm works efficiently under a weaker condition: it is sufficient to require that  the part of the multicut dual not intersecting $\mathcal{F}$ has low treewidth. It might be unintuitive how such a condition can be guaranteed, but the following application may give some explanation.

Pandey and van Leeuwen~\cite{doi:10.1137/1.9781611977073.81} proved that if there are $k$ faces in a planar graph such that each terminal is incident to at least one of them, then
\mwc{} can be solved in time $f(k)n^{O(\sqrt{k})}$. 
Given an $N$-embedded instance $(G,H)$ of \mc{},  the \emph{face cover number} is the size of a minimum set $F^*$ of faces such that every terminal of $H$ is incident to some face in $F^*$.
Let $F^*$ be this set of  faces and consider an inclusionwise minimal multicut dual $C$. As previously observed, every face of $C$ (and hence of $C-F^*$) contains a terminal. From the definition of $C-F^*$ it follows that boundary of a face $f_1$ of $C-F^*$ cannot go through any face $f\in F^*$. Hence, if a face $f_1$ contains a terminal on the boundary of $f$, then the face $f_1$ has to fully contain $f$. This implies that $C-F^*$ has at most $|F^*|$ faces, and hence has treewidth $O(\sqrt{g^2+gk+k})$, where $k=|F^*|$ is the face cover number. Then the following algorithmic result follows from \cref{algogenext}:
\begin{restatable}{thm}{facealg}\label{facealg}
  Let $(G,H)$ be an $N$-embedded instance of \mc{} with face cover number $k$. Then $(G,H)$ can be solved in $f(t,g)n^{O(\sqrt{g^2+kg+k})}$.
\end{restatable}
\medskip

Note that by applying the result to individual connected components, \cref{facealg} can be readily extended to instances with a disconnected graph $G$ (which is not necessarily cellularly embeddable).
Note further that, compared to the algorithm of Pandey and van Leeuwen~\cite{doi:10.1137/1.9781611977073.81}, our result considers the more general problem \mc{} and also applies to surfaces other than the plane. On the other hand, our running time has a multiplicative factor depending on $t$, the number of terminals.

\paragraph*{Combinatorial Results}
\cref{thm:Wgenus} considers two cases depending on whether every graph in $\calH$ is a trivial pattern or not. It is easy to show that the triangle is a projection of every graph that is not trivial, so part (2) of \cref{thm:Wgenus} can be proved under the assumption that the triangle is the projection of some $H\in\calH$.

\begin{restatable}{lem}{combsimple}\label{triisi}
	Let $H$ be a graph that is not a trivial pattern. Then the triangle is a projection of $H$.
\end{restatable}
\medskip

The statements of \cref{thm:Wt,thm:maingenus} contain case distinctions depending on whether $\calH$ has bounded distance to extended bicliques. We prove a combinatorial result showing that either this distance is bounded, or $\calH$ contains all cliques, or $\calH$ contains all complete tripartite graphs.

\begin{restatable}{thm}{combmainclass}
	\label{combmainclass}
\longversion{        Let $\calH$ be a projection-closed class of graphs. Then 
  one of the following holds:
        	\begin{itemize}
			\item $\calH$ contains $K_t$ for every $t\ge 1$.
			\item $\calH$ contains $K_{t,t,t}$ for every $t\ge 1$.
			\item There is a $\mu\ge 0$ such that every $H\in \calH$ has distance at most $\mu$ to extended bicliques.
      		\end{itemize}
}\shortversion{If $\calH$ is a projection-closed class of graphs, then either (i) it contains $K_t$ for every $t\ge 1$, or (ii) contains $K_{t,t,t}$ for every $t\ge 1$, or (iii) there is a $\mu\ge 0$ such that every $H\in \calH$ has distance at most $\mu$ to extended bicliques.}
\end{restatable}
\medskip

\longversion{
\cref{combmainclass} classifies graph classes. It is an immediate consequence of the following result about individual graphs and the fact that for positive integers $t_1<t_2$, we have that $K_{t_2}$ and $K_{t_2,t_2,t_2}$ contain $K_{t_1}$ and $K_{t_1,t_1,t_1}$ as projections, respectively.

\begin{restatable}{thm}{combmain}	\label{combmain}
  There is a function $f$ such that, for every graph $H$ and every positive integer $t$
 \longversion{one of the following holds:
		\begin{itemize}
			\item $K_t$ is a projection of $H$,
			\item $K_{t,t,t}$ is a projection of $H$,
			\item an extended biclique can be obtained from $H$ by deleting at most $f(t)$ vertices.
                        \end{itemize}}\shortversion{either 
(i) $K_t$ is a projection of $H$,
(ii) or $K_{t,t,t}$ is a projection of $H$, or (iii) an extended biclique can be obtained from $H$ by deleting at most $f(t)$ vertices.}
                        \end{restatable}
\medskip}\shortversion{\addtocounter{thm}{1}}

The proof of \longversion{\cref{combmain}}\shortversion{\cref{combmainclass}} considers two cases, depending on the distance of \longversion{$H$}\shortversion{a graph $H\in\calH$} to being a cograph. A graph is cograph if it does not contain the 4-vertex path $P_4$ as an induced subgraph. Suppose that ``many'' vertices are needed to be deleted to obtain a cograph from $H$. By a standard connection between hitting and packing, this means that there are many vertex disjoint induced copies of $P_4$ in $H$. This implies that $H$ has a projection that contains many vertex-disjoint triangles. We analyze the connections between these triangles using a Ramsey-theoretic argument, and conclude that either a large clique or a large complete tripartite graph can be obtained as a projection. Now consider the case when a cograph $H'$ can be obtained from $H$ by removing a few vertices. We observe two simple alternative characterizations of a cograph being bipartite: (1) it is triangle free or (2) every component is a biclique. If $H'$ contains many vertex disjoint triangles, then we argue as above in the other case. Otherwise, $H'$ (and hence $H$) can be turned into a bipartite cograph by the deletion of a few vertices. Thus $H$ is a collection of bicliques with a few extra vertices, and then it is easy to argue that one of the conclusions of \longversion{\cref{combmain}}\shortversion{\cref{combmainclass}} holds.

\paragraph*{Complexity Results}
\cref{combmainclass} shows that to obtain our main lower bound, \cref{thm:maingenus}\,(3b), it is sufficient to consider two cases: when $\calH$ is the class of all cliques \longversion{(that is, \mwc) }and when $\calH$ is the class of all tripartite graphs\longversion{ (that is, \gcut)}. For the former problem, \cref{thm:mc-tight} almost provides a lower bound, except that the case $t=3$ needs to be resolved. We fill this gap. Then we show how the proof can be adapted to the latter problem.

\longversion{For the $\W{1}$-hardness proofs of~\cref{thm:Wgenus,thm:Wt}, we need an argument that formalizes that a pattern $H$ is always at least as hard as its projection $H'$.
Given two parameterized problems, a \emph{parameterized reduction} from the first problem to the second one is an algorithm that maps every instance $(I,k)$ of the first problem to an instance $(I',k')$ of the second problem such that $k'\leq f(k)$ and the algorithm runs in $f(k)p(|I|)$ for some polynomial function $p$ and some computable function $f$.

  \begin{restatable}{lem}{lemredclosure}
\label{lem:redclosure}
 Let $\mathcal{H}$ be a computable class of graphs and let $\mathcal{H}'$ be the projection closure of $\calH$. There is a parameterized reduction from \mcut{\calH'} to \mcH, parameterized by $t$.
 Furthermore, the reduction preserves the orientable genus of the instance and the maximum edge weight.
\end{restatable}
}\shortversion{\addtocounter{thm}{1}}

\subparagraph*{The Gadget used in the Proof of \cref{thm:mc-tight}}
Let us review the proof of \cref{thm:mc-tight} by Cohen-Addad et al.~\cite{Cohen-AddadVMM21}. This lower bound uses a gadget by Marx~\cite{Marx12}
 for the \mwc{} problem on planar graphs. Let us briefly recall the properties of this gadget. For some integer $\Delta$, the gadget $G_\Delta$ has a planar embedding, where the following distinguished vertices appear on the infinite face (in clockwise order) $UL$, $u_1$, $\dots$, $u_{\Delta+1}$, $UR$, $r_1$, $\dots$, $r_{\Delta+1}$, $DR$, $d_{\Delta+1}$, $\dots$, $d_1$, $DL$, $\ell_{\Delta+1}$, $\dots$, $\ell_1$. Here the four vertices $UL$, $UR$, $DR$, $DL$, the so-called \emph{corners} of $G_\Delta$, are considered terminals. \longversion{

} Let $M\subseteq E(G_\Delta)$ be a set of edges such that the removal of these edges disconnects the four terminals from each other. We say that $M$ \emph{represents} the pair $(x,y)\in [\Delta]\times[\Delta]$ if $G_\Delta\setminus M$ has exactly four components and each component contains the terminals in  precisely one of the following sets:

\longversion{
\begin{itemize}
	\item $\{UL, u_1, \ldots, u_y, \ell_1, \ldots, \ell_x\}$,
	\item $\{UR, u_{y+1}, \ldots, u_{\Delta+1}, r_1, \ldots, r_x\}$,
	\item $\{DL, d_1, \ldots, d_y, \ell_{x+1}, \ldots, \ell_{\Delta+1}\}$, and
	\item $\{DR, d_{y+1}, \ldots, d_{\Delta+1}, r_{x+1}, \ldots, r_{\Delta+1}\}$.
        \end{itemize}}%
\shortversion{
\begin{itemize}
	\item $\{UL, u_1, \ldots, u_y, \ell_1, \ldots, \ell_x\}$,
	$\{UR, u_{y+1}, \ldots, u_{\Delta+1}, r_1, \ldots, r_x\}$,
	\item $\{DL, d_1, \ldots, d_y, \ell_{x+1}, \ldots, \ell_{\Delta+1}\}$,
	 $\{DR, d_{y+1}, \ldots, d_{\Delta+1}, r_{x+1}, \ldots, r_{\Delta+1}\}$.
        \end{itemize}}%

        That is, there a is cut between $\ell_x$ and $\ell_{x+1}$ on the left side and there is a cut at precisely the same location, between $r_x$ and $r_{x+1}$, on the right side\longversion{, and similarly between the upper and lower sides.}\shortversion{, etc.}

The gadget $G_\Delta$ has the property that every minimum set of edges separating the four terminals is a cut representing some pair $(x,y)$. Moreover, given a set $S\subseteq [\Delta]\times [\Delta]$, we can construct a gadget $G_S$ where every such cut actually represents a pair $(x,y)\in S$, and conversely, for every pair $(x,y)\in S$ there is a minimum cut representing $(x,y)$. This gadget $G_S$ can be used in the following way in a hardness proof. Suppose that the solution contains a set of edges that is a minimum cut $M$ separating the four terminals of the gadget; by the properties of the gadget $G_S$, the cut $M$ represents a pair $(x,y)\in S$.  Then, intuitively, the gadget transmits the information ``$x$'' between the left and right sides, transmits the information ``$y$'' between the upper and lower sides, and at the same time enforces that the two pieces of information are compatible, that is, $(x,y)\in S$ should hold.

Our first observation is that the construction of this gadget $G_S$, as described by Marx~\cite{Marx12}, has some stronger property that was not exploited before. Namely, the conclusion about the cut $M$ representing some pair $(x,y)\in S$ holds even if we do not require that $M$ separate all four terminals, but we allow that $DL$ and $UR$ are in the same component of $G_S\setminus M$. In other words, it turns out that every ``cheap'' cut that satisfies this weaker separation property separates $DL$ and $UR$ as well. We formalize this observation in the following way.  A \emph{good cut} of $G_S$ is a subset $M$ of its edges such that the connected components of $G_S\setminus M$ that contain one of $UL$ or $DR$ do not contain any other corners of $G_S$. \longversion{As we show in \cref{sec:gridgadgets}, the following statement can be obtained by slightly modifying the proof of correctness (but not the gadget construction itself) given by Marx~\cite{Marx12}.}\shortversion{We show that the following statement can be obtained by slightly modifying the proof of correctness.} Here $W^*$ is some specific integer that depends on $\Delta$. It is specified in \cref{sec:gridgadgets}.

\begin{restatable}{lem}{goodcutmain}\label{lem:gadget_main_new}\label{cor:gadget_stronger}
	Given a subset $S\subseteq [\Delta]^2$, the grid gadget $G_S$ can be constructed in time polynomial in $\Delta$, and it has the following properties:
	\begin{enumerate}
		\item For every $(x,y)\in S$, the gadget $G_S$ has a cut of weight $W^*$ representing $(x,y)$.
		\item If a good cut of $G_S$ has weight $W^*$, then it represents some $(x,y)\in S$.
		\item Every good cut of $G_S$ has weight at least $W^*$.
	\end{enumerate}
      \end{restatable}
      \medskip

The proof of \cref{thm:mc-tight} by Cohen-Added et al.~\cite{Cohen-AddadVMM21} considers two regimes: (1) $t=O(g)$ and (2) $t=\Omega(g)$. \longversion{Let us briefly discuss the proofs of the two regimes separately and how they can be changed to obtain our more general \cref{thm:maingenus}.

\subparagraph*{\boldmath Proof Sketch of \cref{thm:mc-tight} for $t=O(g)$}
For the regime $t=O(g)$, Cohen-Added et al.~\cite{Cohen-AddadVMM21}}\shortversion{For the regime $t=O(g)$, they} prove the lower bound already for $t=4$: clearly, this shows the same lower bound for every $t>4$, as the problem cannot get easier with more terminals. The lower bound is proved by a reduction from binary Constraint Satisfaction Problems ($\csp$)\longversion{ and use previous lower result}. An instance of a {\em binary $\csp$} is a triple $(V ,D, K)$,
where \longversion{:

\begin{itemize}
\item $V$ is a set of variables,
\item $D$ is a domain of values,
\item $K$ is a set of constraints, each of which is a triple $\langle
u,v,R\rangle$ with $u,v\in V$ and $R\subseteq D^2$.
\end{itemize}

}\shortversion{$V$ is a set of variables,
 $D$ is a domain of values,
 $K$ is a set of constraints, each of which is a triple $\langle
u,v,R\rangle$ with $u,v\in V$ and $R\subseteq D^2$.} \longversion{Informally, for each constraint $\langle u,v,R\rangle$, the tuples of $R$
indicate the allowed combinations of simultaneous values for variables $u$ and $v$. }A {\em solution} to a $\csp$ instance $(V,D,K)$ is a function $f:V\to D$ such that $(f(u),f(v))\in R$ holds for each constraint $\langle
u,v,R\rangle$. The \emph{primal graph} of the instance $(V,D,K)$ is the undirected graph with vertex set $V$ where $u$ and $v$ are adjacent if there is a constraint $\langle
u,v,R\rangle$.\longversion{

} Cohen-Added et al.~\cite{Cohen-AddadVMM21} proved a lower bound for a special case of binary $\csp$ called $4$-\textsc{Regular Graph Tiling}, where the primal graph is a 4-regular bipartite graph and the constraints are of a special form. \longversion{The problem is similar to $\gt$, which is often used as source problem in lower bounds for planar and geometric problems \cite{DBLP:conf/iwpec/Marx06,Marx12,DBLP:journals/algorithmica/FeldmannM20,ChitnisFHM20,DBLP:journals/corr/abs-2208-06015,DBLP:journals/tcs/BergKW19,DBLP:journals/jco/Blum22,DBLP:journals/algorithmica/BonnetBCTW20,DBLP:conf/alt/Chitnis22,DBLP:journals/algorithmica/ChitnisFS19}, but now the constraints can form an arbitrary nonplanar graph.}

{
	\begin{description}\setlength{\itemsep}{0pt}
		\setlength{\parskip}{0pt}
		\setlength{\parsep}{0pt}   			
		\item[\bf Name:] $4$-\textsc{Regular Graph Tiling} 
		\item[\bf Input:]   A tuple $(k,\Delta, \Gamma, \{S_v\})$, where
			\begin{itemize}
				\item $k$ and $\Delta$ are positive integers,
				\item $\Gamma$ is a $4$-regular graph (parallel edges allowed) on $k$ vertices in which the edges are labeled $U$, $D$, $L$, $R$ in a way that each vertex is adjacent to exactly one of each label,
				\item for each vertex $v$ of $\Gamma$, we have $S_v\subseteq [\Delta]\times [\Delta]$.
			\end{itemize}
		\item[\bf Output:]  For each vertex $v$ of $\Gamma$, a pair $s_v\in S_v$ such that, if $s_v=(i,j)$,
			\begin{itemize}
				\item the first coordinates of $s_{L(v)}$ and $s_{R(v)}$ are both $i$, and
				\item the second coordinates of $s_{U(v)}$ and $s_{D(v)}$ are both $j$,
			\end{itemize}
			where $U(v)$, $D(v)$, $L(v)$, and $R(v)$ denote the vertex of $\Gamma$ that is connected to $v$ via the edge labeled by $U$, $D$, $L$, and $R$, respectively.
	\end{description}
      }

Lower bounds on general binary $\csp$ \longversion{parameterized by the number of constraints }\cite{DBLP:journals/toc/Marx10,DBLP:journals/corr/abs-2311-05913} can be translated to this problem:

\begin{thm}[Cohen-Addad et al.~\cite{Cohen-AddadVMM21}]\label{thm:graphtiling}
\longversion{\begin{enumerate}
\item The \textup{\textsc{4-Regular Graph Tiling}} problem restricted to instances whose underlying graph is bipartite is W[1]-hard parameterized by the integer~$k$.
  \item Assuming $\ETH$, there exists a universal constant $\alpha$ such that for any fixed integer $k\geq 2$, there is no algorithm that decides all the \textup{\textsc{4-Regular Graph Tiling}} instances whose underlying graph is bipartite and has at most $k$ vertices, in time $O(\Delta^{\alpha \cdot k / \log k})$.
  \end{enumerate}}%
\shortversion{
Assuming $\ETH$, there exists a universal constant $\alpha$ such that for any fixed integer $k\geq 2$, there is no algorithm that decides all the \textsc{4-Regular Graph Tiling} instances whose underlying graph is bipartite and has at most $k$ vertices, in time $O(\Delta^{\alpha \cdot k / \log k})$.
  }
\end{thm}

The gadgets $G_S$ defined above offer a very convenient reduction from \textsc{4-Regular Graph Tiling}. We represent each variable $v$ by a gadget $G_{S_v}$ and, for each edge of $\Gamma$, we identify the distinguished vertices on two appropriate sides of two gadgets. By our intuitive interpretation of the gadgets, if we consider a minimum cut, then each gadget $G_{S_v}$ represents two pieces of information $(x,y)\in S_v$. Furthermore, if two sides are identified, then the fact that the cut happens on these two sides at the same place implies that the corresponding pieces of information in the two gadgets are equal---precisely as required by the definition of \textsc{4-Regular Graph Tiling}. Thus we obtain a reduction from \textsc{4-Regular Graph Tiling} to \mwc{} with $4|V(\Gamma)|=2|E(\Gamma)|$ terminals. Observe that each gadget is planar and identification of two sides increases genus by at most one: all the identifications can be done in parallel on a single handle. The final step in the proof of \cref{thm:mc-tight} is to identify every copy of $UL$ in every gadget (and similarly for the copies of $UR$, $DR$, and $DL$). One can observe that these identifications do not change the validity of the reduction and increase genus by $O(|E(\Gamma)|)$. Thus we get a reduction from \textsc{4-Regular Graph Tiling} to \mwc{} with 4 terminals ($\tcut{4}$).\longversion{

} The stronger property of \cref{lem:gadget_main_new} allows us to further identify $DL$ and $UR$ without changing the validity of the reduction. Thus we can improve the lower bound from 4 terminals to 3 terminals.

\begin{restatable}{thm}{threeterm}\label{lem:threeterm}\label{prop:3TCut}\label{thm:3TCut}
\begin{enumerate}
    \item $\tcut{3}$ on instances with orientable genus $\go$ is $\W{1}$-hard parameterized by $\go$.
    \item Assuming the $\ETH$, there exists a universal constant $\alpha$ such that for any fixed integer $\go \geq 0$, there is no algorithm that decides all unweighted $\tcut{3}$ instances with orientable genus $\go$ in time $O(n^{\alpha \cdot (\go+1)/\log (\go+2)})$.
\end{enumerate}
\end{restatable}
       \medskip       

By \cref{triisi}, the triangle is a projection of every nontrivial pattern. Thus, \cref{lem:threeterm} proves  \cref{thm:maingenus}\,(2b) and the $t\le \go$ case of \cref{thm:maingenus}\,(3b). \longversion{Furthermore, using the reduction provided by \cref{lem:redclosure}, we obtain \cref{thm:Wgenus}\,(2).}

\subparagraph*{\boldmath Proof Sketch of \cref{thm:mc-tight} for $t=\Omega(\go)$}
For the regime $t=\Omega(\go)$, Cohen-Addad et al.~\cite{Cohen-AddadVMM21} present a reduction from $\csp$ with a 4-regular constraint graph with orientable genus $\go$. The first part of their reduction can be expressed as a lower bound for $\csp$ the following way:

\begin{restatable}{thm}{thmCSPgenuslower}\label{thm:CSP-genus-lower}
There is a universal constant $\alpha$ such that for every choice of nonnegative integers $s$ and $\go$ with $s \geq \go+1$, there exists a 4-regular multigraph $P$ with $|V(P)|\leq s$ of orientable genus at most $\go$ and the property that, unless $\ETH$ fails, there is no algorithm deciding the $\csp$ instances $(V,D,K)$ of $\csp$ whose primal graph is $P$ and that runs in $O(|D|^{\alpha \sqrt{\go s+s}/\log{(\go+s)}})$. 
\end{restatable}
\medskip

The gadgets of Marx~\cite{Marx12} \longversion{can be used to reduce such an instance of}\shortversion{reduce} $\csp$ to \mwc{} by introducing one gadget for each variable and a constant number of gadgets for each constraint. This results in an instance of \mwc{} with the same genus $\go$ as the $\csp$ instance and $t=O(s)$ terminals. The $t>\go$ case of \cref{thm:mc-tight} now follows from \cref{thm:CSP-genus-lower} and this reduction.

\subparagraph*{Group Terminal Cut}
By carefully adjusting the proof of \cref{thm:mc-tight} for \mwc{}, we can also obtain a reduction to \gcut{}. Where previously we identified all corners with the same label in every gadget, we can now include corners with identical label in the same group of terminals. 
The point is that if we have a solution that separates the 3 groups, then this results in a good cut (in the sense of \cref{lem:gadget_main_new}) in each gadget. 

\begin{restatable}{thm}{threegroupCut}
  \label{thm:3groupCut}\label{negmain}
  Assuming $\ETH$, there exists a universal constant $\alpha>0$ such that for any $\go\ge 0$ and $t\ge 3$, there is no  $\bigO(n^{\alpha\sqrt{\go^2+\go t+t}/\log{(\go+t)}})$ algorithm for unweighted $\gcut$, even when restricted to instances with orientable genus at most $\go$ and at most $t$ terminals.
     \end{restatable}
     \medskip
In a very similar way, the following qualitative complexity result can be obtained.

\begin{restatable}{thm}{threegroupCutw}
  \label{thm:3groupCutw1}
   Unweighted $\gcut$ is $W[1]$-hard on planar graphs when parameterized by $t$.
     \end{restatable}
\medskip

By \cref{combmainclass}, if $\calH$ is projection-closed and has unbounded distance to extended bicliques, then $\calH$ contains either all cliques or all complete tripartite graphs. Thus in the $t>g$ regime, \cref{thm:maingenus}\,(3b) follows from the lower bound for \mwc{} (\cref{lem:threeterm}) and for \gcut{} (\cref{thm:3groupCut}) in the two cases, respectively.

\longversion{
	\subparagraph*{Improvement for Planar Graphs}
Moreover, looking at the two main cases of the proof (cliques and complete tripartite graphs), we can verify that in our proofs and earlier results, for the special case of planar graphs, the logarithmic factor does not appear in the runtime.

\begin{restatable}{thm}{planarstronger}\label{thm:planarstronger}
   Let $\calH$ be a projection-closed class of graphs. Assuming $\ETH$, there is a universal constant $\alpha>0$ such that for any fixed choice of $t\ge 3$, there is no $O(n^{\alpha \sqrt{t}})$ algorithm for unweighted \mcH, even when restricted to planar instances with at most $t$ terminals.
 \end{restatable}
\medskip

}


\longversion{
\subsection{Proof of Main Results}\label{sec:proofofmains}
We now prove our main results (\cref{thm:Wgenus} to \cref{thm:mainalgo}) using statements from \cref{sec:intro:techniques}.  The auxiliary results from \cref{sec:intro:techniques} are then shown in subsequent sections.
\medskip

\Wgenus*
\begin{proof}
	The algorithmic statement is \cref{obs:poly}. For the hardness result, we use that $\calH$ contains a graph that is not a trivial pattern. By \cref{triisi}, the triangle is in the projection-closure $\calH'$ of $\calH$. According to \cref{lem:redclosure}, there is a parameterized reduction from $\mcut{\calH'}$ to $\mcH$ with parameter $t$ that preserves orientable genus. As in a parameterized reduction the number of terminals of the created instance is bounded by a function $f$ of the number of terminals of the old instance, the terminals of relevant instances of $\mcH$ can be bounded by $t_\calH=f(3)$.
	The result then follows from the $\W{1}$-hardness result for $\tcut{3}$ given in \cref{lem:threeterm}.
\end{proof}

\Wt*
\begin{proof}
	The algorithmic statement follows from \cref{thm:mainalgo} for $\mu= \bigO(1)$.
	For the hardness result, consider the projection-closure $\calH'$ of $\calH$. According to \cref{lem:redclosure}, there is a planarity-preserving parameterized reduction from $\mcut{\calH'}$ to $\mcH$ with parameter $t$. Since $\calH'$ is both projection-closed and has unbounded distance to extended bicliques, by \cref{combmainclass}, it contains all cliques or all complete tripartite graphs.

	In the first case, $\W{1}$-hardness for $\mcut{\calH'}$ on panar graphs follows from the known fact~\cite{Marx12} that $\mwc$ is $\W{1}$-hard on planar graphs parameterized by the number of terminals $t$.
	In the latter case, $\mcut{\calH'}$ is a generalization of $\gcut$, which is $\W{1}$-hard parameterized by $t$, even on planar graphs according to \cref{thm:3groupCutw1}.
\end{proof}

\maingenus*
\begin{proof} We go through the different cases.
	\begin{itemize}
		\item Item 1 is from \cref{obs:poly}.
		\item Item 2 a) follows from \cref{thm:mainalgo} for $\mu= \bigO(1)$.
		\item For Item 2 b), note that $\calH$ contains a pattern that is not trivial, and since $\calH$ is projection-closed, according to \cref{triisi}, it also contains the triangle. 
		Let $\alpha_{\textup{3C}}$ be the universal constant from \cref{lem:threeterm}. Then the respective lower bound follows for $\alpha=\alpha_{\textup{3C}}$.
		\item Item 3 a) is a known result, see \cref{alt}.
		\item Item 3 b) Since $\calH$ is both projection-closed and has unbounded distance to extended bicliques, by \cref{combmainclass}, it contains all cliques or all complete tripartite graphs.
		
		In the former case, Let $\alpha>0$ and let us assume that for some $g\ge 0$ and $t\ge 3$, there is an algorithm $\mathbb{A}$ that solves $\mwc$ instances with orientable genus $\vec{g}$ and at most $t$ terminals in time $\bigO(n^{\alpha\sqrt{\vec gt+\vec g^2+t}/\log{(\vec g+t)}})$. If $t\ge 4$ then we obtain a contradiction to \cref{thm:mc-tight} if $\alpha\le \alpha_{\textup{MC}}$, where $\alpha_{\textup{MC}}$ is the universal constant from \cref{thm:mc-tight}. If $t=3$ then $\sqrt{\vec gt+\vec g^2+t}\le 10(\vec g+1)$. Consequently, $\mathbb{A}$ solves $\tcut{3}$ in time $\bigO(n^{\alpha\cdot 10(\vec g+1)/\log{(\vec g+2)}})$, which is a contradiction to \cref{lem:threeterm} if $\alpha<\alpha_{\text{3C}}/10$, where $\alpha_{\text{3C}}$ is the universal constant from \cref{lem:threeterm}. 
		
		In the latter case, $\calH$ contains all complete tripartite graphs, and hence $\mcH$ generalizes $\gcut$. Therefore, the result follows if $\alpha<\alpha_{\textup{G3T}}$, where $\alpha_{\textup{G3T}}$ is the universal constant from \cref{thm:3groupCut}.
		Summarizing, the statement holds for a sufficiently small choice of $\alpha>0$.
	\end{itemize}
\end{proof}

\cref{cor:3TCutmain} is a special case and direct consequence of \cref{thm:maingenus}.

\section{Algorithmic Results}\label{sec:algres}
In this section, we prove the main algorithmic result of this article. In \cref{sec:prelims}, we collect some preliminaries and, in \cref{sec:mainalgo}, we give the main proof of \cref{thm:mainalgo}.
\subsection{Prerequisites}\label{sec:prelims}

In this section, we give a collection of definitions and preliminary results that will be useful throughout Section \ref{sec:algres}.
 A graph $G$ has vertex set $V(G)$ and edge set $E(G)$. An oriented graph $\vec{G}$ has vertex set $V(\vec{G})$ and arc set $A(\vec{G})$. 
 Given disjoint sets $X,Y \subseteq V(G)$, an \emph{$(X,Y)$-cut} is a set $S$ of edges such that no component of $G\setminus S$ contains vertices from both $X$ and $Y$. Further, we use $\delta_G(X,Y)$ for the set of edges with one endvertex in $X$ and one endvertex in $Y$ and we abbreviate $|\delta_G(X,Y)|$ to $d_G(X,Y)$, $\delta_G(X,V(G)\setminus X)$ to $\delta_G(X)$, and $|\delta_G(X)|$ to $d_G(X)$. Given some $E' \subseteq E(G)$, we use $V(E')$ for the set of endvertices of $E'$. Given a path $P$ and $u,v \in V(P)$, we use $uPv$ for the unique $uv$-path which is a subgraph of $P$. 
 
\subsubsection{Genus and Treewidth}
We give some definitions and preliminaries of embbeded graphs in general. For an introduction to the subject, we refer the reader to \cite{mt}. Given an embedded graph $G$, a \emph{segment} of an edge is a subtrail of the polygon trail associated to the edge. A segment is called \emph{trivial} if it consists of a single point, \emph{nontrivial} otherwise. We use $F(G)$ for the set of faces of $G$. If a segment is incident to two distinct faces $f$ and $f'$ of $G$, we say that the segment \emph{separates} $f$ and $f'$.

We say that two graphs $G_1,G_2$ embedded on a surface $N$ are in general position 
if, for $i\in [2]$ and every $v \in V(G_i)$, we have that $v$ is not placed on any vertex or edge of $G_{3-i}$ and further, for every $e_1 \in E(G_1)$ and $e_2 \in E(G_2)$, we have that $e_1$ and $e_2$ only intersect in a finite number of points and that $e_1$ and $e_2$ cross at all points at which they intersect.

We need the following result which was proven in a similar form by Colin de Verdière \cite[Lemma 7.1]{ECDV}. It is also a direct consequence of a result by Eppstein \cite[Theorem 5]{DBLP:conf/soda/Eppstein03}.
\begin{prop}\label{twroot}
If a graph $C$ has Euler genus at most $g$, then $\tw(C)= \bigO(\sqrt{(g+1)|V(C)|})$.
\end{prop}
We can now conclude a similar treewidth bound taking into account the faces rather than the vertices of the embedded graph.
\begin{prop}\label{twroot2}
Let $C$ be a graph embedded in a surface of Euler genus $g$. Then $\tw(C)= O(\sqrt{g^2+g|F(C)|+|F(C)|})$. 
\end{prop}
\begin{proof}
We may suppose that $C$ is of maximum treewidth among all graphs embedded in $N$ having at most $|F(C)|$ faces and among all those, we have that $|V(C)|+|E(C)|$ is minimum. Clearly, we may suppose that $\tw(C)\geq 3$, so in particular $|V(C)|\geq 2$. If $C$ contains a loop $e$, we have $\tw(C\setminus\{e\})=\tw(C)$ and $|F(C\setminus\{e\})|\le |F(C)|$, a contradiction to the minimality of $C$. Hence $C$ does not contain any loop. If there is some $v \in V(C)$ with degree $d_C(v)\leq 1$, as $|V(C)|\geq 2$, we obtain $\tw(C-v)=\tw(C)$ and $|F(C\setminus\{e\})|= |F(C)|$, a contradiction to the minimality of $C$. If there is some $v \in V(C)$ with $d_C(v)=2$, then let $C'$ be obtained from $C-v$ by adding the edge $u_1u_2$ where $vu_1$ and $vu_2$ are the edges of $E(C)$ incident to $v$. Observe that $C'$ is well-defined as $C$ does not contain loops. Observe that, if $u_1=u_2$, then $u_1$ is a cut vertex of $C$. Hence, as $\tw(C)\geq 3$, we obtain $\tw(C')=\tw(C)$ and $|F(C')|=|F(C)|$, which again is a contradiction to the minimality.

Summarizing, we obtain that $d_C(v)\geq 3$ holds for all $v \in V(C)$. This yields $|V(C)|\leq \frac{1}{3}\sum_{v \in V(C)}d_C(v)=\frac{2}{3}|E(C)|$. Euler's formula \cite{10.5555/3134208} now yields $|V(C)|+|F(C)|+g=|E(C)|+2\geq \frac{3}{2}|V(C)|+2$, so $|V(C)|\leq 2 |F(C)|+2g-4$. The statement now follows from Proposition \ref{twroot}.
\end{proof}

\subsubsection{Multicuts and Multicut Duals}\label{mcmcd}
In this section, we review the concept of so-called multicut duals, which will play a crucial role throughout \cref{sec:algres}. Let $N$ be a surface. Recall that an instance $(G,H)$ of $\mc$ is \emph{$N$-embedded} if $G$ is given in the form of a graph cellularly embedded in $N$.
Given a graph $C$ that is embedded in $N$ and in general position with $G$, we denote by $e_G(C)$ those edges of $G$ that are crossed by at least one edge of $C$. The \emph{weight $w(C)$} of $C$ is defined to be $w(e_G(C))$ where $w$ is the weight function associated with $G$. A \emph{multicut dual} for $(G,H)$ is a graph $C$ embedded on $N$ that is in general position with $G$ and has the property that, for every $t_1,t_2 \in V(H)$ with $t_1t_2 \in E(H)$, the terminals $t_1$ and $t_2$ are contained in different faces of $C$.

Now we show the connection linking multicuts and multicut duals. The corresponding proof ideas are from the proof of \cite[Proposition 3.1]{ECDV}.

\begin{lem}\label{dualcut}
Let $(G,H)$ be an $N$-embedded instance of $\mc$ and $C$ be a multicut dual for $(G,H)$. Then $e_G(C)$ is a multicut for $(G,H)$.
\end{lem}
\begin{proof}
Let $t_1,t_2 \in V(H)$ with $t_1t_2 \in E(H)$. Let $Q$ be a $t_1t_2$-path in $G$. Observe that, as $C$ is a multicut dual, we have that $t_1$ and $t_2$ are contained in distinct faces of $C$. Hence $V(Q)$ contains two vertices $v_1,v_2$ which are linked by an edge $e \in E(Q)$ such that $v_1$ and $v_2$ are contained in distinct faces of $C$. It follows that $e$ is crossed by an edge of $C$. Hence $t_1$ and $t_2$ are contained in distinct components of $G\setminus e_G(C)$.
\end{proof}

\begin{lem}\label{dualcut2}
Let $(G,H)$ be an $N$-embedded instance of $\mc$ and $S$ a multicut for $(G,H)$. Then there exists a multicut dual $C$ for $(G,H)$ with $e_G(C)=S$.
\end{lem}
\begin{proof}
Let $G'$ be a triangulation in $N$ that is obtained from $G$ by adding a set $X$ of edges and let $S'=S \cup X$. We will show that there exists a multicut dual $C$ for $(G',H)$ with $e_{G'}(C)=S'$.
Let $C$ be the embedded graph that has a vertex $x_f$ inside every face $f$ of $G'$ and, such that for every $e \in S'$ separating two faces $f,f'$ of $G'$, $C$ has an edge between $x_f$ and $x_{f'}$ that crosses $e$ in one point and is fully contained in the union of $f$ and $f'$ otherwise. Observe that $e_{G'}(C)=S'$. 

Clearly, we have $e_{G'}(C)=S'$. We now show that $C$ is a multicut dual for $(G',H)$. Let $t_1,t_2 \in V(H)$ with $t_1t_2 \in E(H)$ and suppose for the sake of a contradiction that $t_1$ and $t_2$ are in the same face of $C$. Then there exists a polygon trail $P$ from $t_1$ to $t_2$ that does not cross any edge of $C$. As the surface obtained from $N$ by removing $C$ is open, we may suppose that $P$ is in general position with $G'$ except for its endpoints. Let $x_2,\ldots,x_{k-1}$ be the points in which $P$ crosses edges of $G'$ when going from $t_1$ to $t_2$. We further set $x_1=t_1$ and $x_k=t_2$. As $C$ crosses every edge of $G'$ at most once by construction, we may suppose by possibly shortcutting $P$ that $P$ crosses every edge of $G'$ at most twice, hence $k$ is finite. Next, for every $i \in [k]$, we define a set $V_i \subseteq V(G)$ containing one or two vertices in the following way: we set $V_1=\{t_1\}, V_k=\{t_2\}$ and for every $i \in \{2,\ldots,k-1\}$, we set $V_i$ to be the set of vertices of $V(G)$ which are reachable from $x_i$ by following an edge segment of $G'$ that is not crossed by an edge of $C$. We will show that for any $i \in [k-1]$, we have that $V_i$ and $V_{i+1}$ are contained in a single component of $G'-S'$.

 Now fix some $i \in [k-1]$.
By construction, there exists a face $f$ of $G'$ such that  $x_i$ and $x_{i+1}$ are on the boundary of $f$ and the segment of $P$ going from $x_i$ to $x_{i+1}$ is fully contained in $f$. As $G$ is a triangulation, $f$ is bounded by a triangle $Q$. Let $f'$ be the face of the overlay of $Q$ and $C$ that is contained in $f$ and is incident to $x_i$. As $P$ contains a segment from $x_i$ to $x_{i+1}$ fully contained in $f$, we obtain that $f'$ is also incident to $x_{i+1}$. Now consider the set $V'\subseteq V(G)$ of vertices of $G$ incident to the face $f'$. By construction, we have $V_i \cup V_{i+1}\subseteq V'$. Further, by the construction of $C'$ and the fact that $Q$ is a triangle, we obtain that $G[V']-S$ is connected, hence in particular  $V_i$ and $V_{i+1}$ are contained in a single component of $G'-S'$. As $i$ was chosen arbitrarily, we obtain in particular that $t_1$ and $t_2$ are contained in the same component of $G'-S'$, a contradiction to $S'$ being a multicut for $(G',H)$.

We hence obtain that $C$ is a multicut dual for $(G',H)$. It follows that $C$ is also a multicut dual for $(G,H)$. Further, we have $e_G(C)=e_{G'}(C)\cap E(G)=S'\cap E(G)=S$.

\end{proof}

We will use the following argument repeatedly. If $C$ is an inclusionwise minimal multicut dual and $e$ is an edge that separates two faces $f_1$ and $f_2$ of $C$, then there have to be two terminals $t_1$ and $t_2$ inside $f_1$ and $f_2$, respectively, that are adjacent in $H$: otherwise, $C-e$ would be a valid multicut, contradicting the mimimality of $C$. In particular, every face of $C$ should contain at least one terminal. 
\begin{obs}\label{faceterminal}
	Let $(G,H)$ be an $N$-embedded instance of $\mc$ and let $C$ be an inclusionwise minimal minimum-weight multicut dual for $(G,H)$. Then every face of $C$ contains a terminal of $V(H)$.
\end{obs}



\subsection{Algorithm for Bounded Distance to Extended Bicliques}\label{sec:mainalgo}

In this section we prove our main algorithmic result which we restate here for convenience.
\mainalgo*

The main technical contribution for the proof of \cref{thm:mainalgo} is captured by \cref{mainprop}, which we also restate.

\mainprop*

In \cref{algogen} we established that an $\bigO(\beta)$ bound on the treewidth of every subcubic, inclusionwise minimal minimum-weight multicut dual means that a minimum weight multicut can be computed in time $f(t,g)n^{O(\beta)}$.
Hence, \cref{mainprop} together with \cref{algogen} directly imply \cref{thm:mainalgo}.  
In order to prove \cref{mainprop}, we need some more preliminaries.
More specifically, we wish to give an upper bound on the treewidth of graphs of bounded Euler genus that have a bounded size $p$-dominating set for some fixed integer $p$. We use the following result of Eppstein \cite{Eppstein1999DiameterAT}, which considers the diameter of the graph instead of the size of a smallest $p$-dominating set. 

\begin{prop}[{\cite[Theorem 2]{Eppstein1999DiameterAT}}]\label{eppstein}
Let $C$ be a graph of Euler genus $g$ whose diameter is $D$. Then $\tw(C)= \bigO((g+1)D)$.
\end{prop}

The following analogous result for the size of a smallest $p$-dominating set now follows easily.
Given a positive integer $p$, a \emph{$p$-dominating} set of $G$ is a set $A \subseteq V(G)$ such that for every $v \in V(G)$, there exists a $va$-path with at most $p$ edges for some $a \in A$.

\begin{prop}\label{domplanar}
Let $C$ be a graph of Euler genus $g$ that admits a $p$-dominating set of size $k$ for some positive integers $k$ and $p$. Then $\tw(C)= \bigO(pk(g+1))$.
\end{prop}
\begin{proof}
Let $P=v_1\ldots v_D$ be a diameter of $C$ and let  $q= \lceil \frac{D}{2p+1}\rceil$. Further, let $A$ be a $p$-dominating set of $C$. For $i \in [q]$, let $Z_i$ be the set of vertices that can be reached from $v_{(i-1)(2p+1)+1}$ by a path with at most $p$ edges in $C$. As $A$ is a $p$-dominating set of $C$, we obtain that $A \cap Z_i$ contains a vertex $w_i\in A$ for all $i \in [q]$. Let $P_i$ be a shortest $v_iw_i$-path.

If $w_i=w_j$ for some distinct $i,j \in q$, we obtain that the concatenation of $v_1Pv_i, P_i,P_j$ and $v_jPv_D$ is a shorter $v_1v_D$-walk than $P$, a contradiction. It follows that $w_i \neq w_j$ for all distinct $i,j \in [q]$. This yields  $|A|\geq q \geq \frac{1}{2p+1}D$ and so the statement follows by Proposition \ref{eppstein}. 
\end{proof}

We are now ready to prove \cref{mainprop}.

\begin{proof}[Proof of \cref{mainprop}] Recall that, by assumption, there is an extended biclique partition $(B_1,B_2,I,X)$ of $H$ with $\abs{X}\le \mu$. 
 Let $C$ be a subcubic, inclusionwise minimal minimum-weight multicut dual for $(G,H)$ and let $C'$ be a component of $C$. If all vertices of $C'$ have degree at most $2$ in $C'$, we obtain that $C'$ is a cycle or a loop, so in particular $\tw(C')\leq 2=\bigO(1)$ holds. We may hence suppose that $C'$ contains a vertex of degree at least 3. Further, let $F_X$ be the set of faces of $C'$ that contain a vertex from $X$. 
 
\begin{claim}\label{xoradjx} 
Every face of $C'$ is either in $F_X$ or adjacent to a face in $F_X$.
\end{claim}
\begin{claimproof}
Let $f'_1$ be a face of $C'$. We will show that $f'_1$ is in $F_X$ or otherwise is adjacent to a face from $F_X$.  Hence the claim holds.

By assumption, as $C'$ is subcubic, there exists a vertex $v \in V(C')$ incident to $f_1'$ that satisfies $d_{C'}(v)=3$. Observe that, by the minimality of $C$, we have that $v$ is incident to exactly 3 faces of $C'$. Let $f_2'$ and $f_3'$ be the other two faces incident to $v$. If one of $f'_1$, $f'_2$, $f'_3$ contains a terminal of $X$, the claim holds for $f'_1$. For the sake of a contradiction, suppose otherwise. 

Observe that there are unique faces $f_1,f_2$, and $f_3$ of $C$ such that $f_i$ is contained in $f'_i$ for $i \in[3]$ and $v$ is incident to $f_1,f_2$, and $f_3$ in $C$. As $C$ is a multicut dual for $(G,H)$ and by the assumption that none of $f_1',f_2'$ and $f_3'$ contains a terminal of $X$, we obtain that for every $i \in [3]$, there exists a $k \in [2]$ such that all terminals contained in $f_i$ are contained in $B_k \cup I$, where we use the fact that no face can contain terminals from both $B_1$ and $B_2$ as these sets are completely connected in $H$. By symmetry, we may suppose that there exist distinct $i,j \in [3]$ such that all terminals contained in $f_i$ or $f_j$ or contained in $B_1 \cup I$. Observe that there exists an edge $e$ of $C'$ that separates $f_i$ from $f_j$. However, none of the terminals in $B_1 \cup I$ need to be separated and therefore none of the terminals in $f_i$ and $f_j$ need to be separated. Consequently, $C-e$ is also a multicut dual for $(G,H)$, a contradiction to the minimality of $C$.
\end{claimproof}

Let us point out that \cref{xoradjx} implies that $F_X$ is nonempty and therefore we can now use that $\mu\ge 1$.
For every face $f$ of $C'$ we now introduce a vertex $z_f$ that is connected by an edge to every  vertex of $C'$ that is incident to $f$. Let $C'_0$ be the resulting (triangulated) graph.
Observe that $C'_0$ can be embedded on $N$. Now let $A\subseteq V(C'_0)$ contain the vertex $z_f$ for all faces $f$ of $C'$ which contain a terminal of $X$. It follows directly from \cref{xoradjx} that $A$ is a 3-dominating set of $C'_0$: if some face $f'$ of $C'$ is adjacent to a face $f$ of $C'$ containing a vertex from $X$, then $z_{f'}$ is at distance $2$ from $z_f$, and consequently every vertex incident to $f'$ is at distance at most $3$ from $z_f\in A$. 
By Proposition \ref{domplanar} and $|A|\leq |X|\leq \mu$, we obtain that $\tw(C'_0) =O(\mu g+\mu)$.  As $C'$ is a subgraph of $C'_0$, we obtain that $\tw(C') =O(\mu g+\mu)$. As $C'$ was chosen to be an arbitrary component of $C$, this yields $\tw(C) =O(\mu g+\mu)$. 
\end{proof}

For planar graphs, we can obtain a stronger result in the same way. It follows from \cite[Theorem 3.2]{10.1145/1077464.1077468} and the fact that the branchwidth of any graph is at most by a constant factor larger than its treewidth that a planar graph $G$ that contains a 3-dominating set of size $\mu$ for some integer $\mu$ satisfies $\tw(G)=O(\sqrt{\mu})$. We hence obtain the following result by the exact same proof as that of \cref{thm:mainalgo}.
\begin{thm}
\label{thm:mainalgoplanar}
		Let $\calH$ be a class of graphs whose distance to extended bicliques is at most $\mu$.
		Then $\mcut{\calH}$ on planar graphs can be solved in time $f(t)n^{O(\sqrt{\mu})}$.
\end{thm}

\subsection{Face cover number}
Given a graph $G$ which is cellularly embedded in $N$, $X \subseteq V(G)$ and $F^*\subseteq F(G)$, we say that $F^*$ {\it covers} $X$ if every element of $X$ is incident to some face of $F^*$.
Given an $N$-embedded instance $(G,H)$ of \mc, we denote by $k$ the minimum number of faces of $G$ that cover $V(H)$, and we say that $k$ is the {\it face cover number} of $(G,H)$.

As another application of Theorem \ref{algogenext}, we can prove the following result.

\facealg*

\begin{proof}
As pointed out by Bienstock and Monma~\cite{doi:10.1137/0217004}, in $f(t,g)n^{O(1)}$, we can compute a set $F^*$ of $k$ faces of $G$ that covers $V(H)$.

Now let $C$ be a minimum weight, inclusion-wise minimal multicut dual for $(G,H)$ and let $C'=C-F^*$. Let $f_0$ be a face of $C'$. If $f_0$ does not contain a terminal of $V(H)$, then for some arbitrary $e \in E(f_0)$, we have that $C-e$ is a muticut dual for $(G,H)$, a contradiction to the minimality of $C$. Hence there exists some $v \in V(H)$ contained in $f_0$. Let $f_1 \in F^*$ be a face incident to $v$. Observe that, by the definition of $C'$, no edge of $f_0$ intersects an edge of $f_1$. As $v$ is contained in $f_0$, it follows that $f_1$ is contained in $f_0$. As $f_0$ was chosen arbitrarily and $|F^*|=k$, we obtain that $C'$ contains at most $k$ faces. By Proposition \ref{twroot2}, we obtain that $\tw(C')=O(\sqrt{g^2+gk+k})$. The statement now follows directly from Theorem \ref{algogenext} and the fact that $|F^*|\leq t$.
\end{proof}



\section{Combinatorial Results}\label{sec:combinatorics}\label{class}
This section is dedicated to proving the combinatorial results we need in this paper. More precisely, we show Lemma \ref{triisi} and Theorem \ref{combmain}.
While the proof of Lemma \ref{triisi} is rather simple, the proof of Theorem \ref{combmain} is somewhat more involved and requires techniques of Ramsey theory.

Recall that a graph $H'$ is a projection of another graph $H$ if it can be obtained from $H$ by a sequence of  vertex deletions and identifications of independent vertices. We further need the following definition. A \emph{cograph} is a graph that does not contain $P_4$ as an induced subgraph. 
We start with two preliminary results which will be used in the proofs of both Lemma \ref{triisi} and Theorem \ref{combmain}.

The first result shows that cographs have no long induced odd cycle.
\begin{prop}\label{biptri}
Let $H$ be a cograph. Then either $H$ is bipartite or $H$ contains a triangle as a subgraph.
\end{prop}
\begin{proof}
Let $H$ be a cograph and et $C=v_1v_2\ldots v_nv_1$ be a shortest odd cycle in $H$. It is well-known that $C$ is chordless as $H$ does not contain a shorter odd cycle. If $n \geq 5$, then $H[\{v_1,\ldots,v_4\}]$ is isomorphic to $P_4$, a contradiction to $H$ being a cograph.
\end{proof}

The following result shows that bipartite cographs have a very limited structure.
\begin{lem}\label{cobi}
	Let $H$ be a bipartite cograph. Then every connected component of $H$ is either a complete bipartite graph or an isolated vertex.
\end{lem}
\begin{proof}
Let $C$ be a connected component of $H$, $(A,B)$ a bipartition of $C$, $a \in A$, and $b \in B$. Let $P=x_1y_1\ldots x_ty_t$ be a shortest $ab$-path in $C$ with $x_1=a$ and $y_t=b$. As $C$ is bipartite, we have $\{x_1,\ldots,x_t\}\subseteq A$ and $\{y_1,\ldots,y_t\}\subseteq B$. If $t \geq 2$, then, as $P$ is a shortest $ab$-path in $C$, we have that $E(C)$ does not contain an edge linking $x_1$ and $y_2$. As $(A,B)$ is a bipartition of $C$ and $x_1y_1x_2y_2$ is a subpath of $P$, it follows that $C[\{x_1,y_1,x_2,y_2\}]$ is isomorphic to $P_4$, a contradiction to $C$ being a cograph. We hence obtain that $t=1$. As $a\in A$ and $b \in B$ were chosen arbitrarily, this yields that $C$ is a complete bipartite graph.
\end{proof}
We can now prove Lemma \ref{triisi}, which we restate for convenience.
\combsimple*
\begin{proof}
Suppose otherwise. If $H$ is not a cograph, then by definition, $H$ contains an induced $P_4=v_1v_2v_3v_4$ as a subgraph. Then the graph obtained from $H$ by identifying $\{v_1,v_4\}$ and deleting all vertices in $V(H)-\{v_1,\ldots,v_4\}$ is a triangle, a contradiction. We may hence suppose that $H$ is a cograph.
By Lemma \ref{cobi}, we obtain that every nontrivial connected component of $H$ is a complete bipartite graph.

First suppose that $H$ contains at least 3 nontrivial connected components $C_1,C_2,C_3$ and let $u_iv_i$ be an edge in $C_i$ for $i=1,2,3$. Observe that the graph obtained from $H$ by identifying $\{u_1,v_2\},\{v_1,u_3\}$ and $\{u_2,v_3\}$ and deleting all vertices in $V(H)-\{u_1,u_2,u_3,v_1,v_2,v_3\}$ is a triangle, a contradiction.

Now suppose that $H$ contains exactly two nontrivial connected components $C_1,C_2$ and that at least one of them, say $C_1$, contains at least 3 vertices. Then $C_1$ contains a path $u_1u_2u_3$ and $C_2$ contains an edge $v_1v_2$. Observe that the graph obtained from $H$ by identifying $\{u_1,v_1\}$ and $\{u_3,v_2\}$ and deleting all vertices in $V(H)-\{u_1,u_2,u_3,v_1,v_2\}$ is a triangle, a contradiction.

We hence obtain that $H$ either contains at most one nontrivial connected component or exactly two nontrivial connected components each of which only contains  two vertices. Hence the statement follows.
\end{proof}

The rest of this section is dedicated to proving Theorem \ref{combmain}, which easily implies Theorem \ref{combmainclass} as pointed out earlier. We restate Theorem \ref{combmain} here for convenience.

\combmain*

\medskip

The following observation is immediate.

\begin{obs}\label{transitiv}
Let $H,H',H''$ be graphs such that $H'$ is a projection of $H$ and $H''$ is a projection of $H'$. Then $H''$ is a projection of $H$.
\end{obs}

We also use the following classical result due to Ramsey which can be found as Theorem 9.1.3 in \cite{10.5555/3134208}:

\begin{thm}\label{ramsey}
There is a function $r:\mathbb{Z}_{\geq 0}\rightarrow \mathbb{Z}_{\geq 0}$ such that for any positive integer $t$ and any 8-coloring of the edges of the complete graph $K_{r(t)}$, there is a set $X\subseteq V(K_n)$ such that $|X|\geq t$ and all the edges both of whose endvertices are in $X$ have the same color.
\end{thm}

The following result whose proof heavily relies on Theorem \ref{ramsey} will be applied multiple times throughout the proof of Theorem \ref{combmain}:

\begin{lem}\label{tri}
There is a function $f_0:\mathbb{Z}_{\geq 0}\rightarrow \mathbb{Z}_{\geq 0}$ such that, for any positive integer $t$ and any graph $H$ that contains at least $f_0(t)$ vertex-disjoint triangles, either $K_t$ or $K_{t,t,t}$ is a projection of $H$.
\end{lem}
\begin{proof}
Let $f_0(t)=r({t^3})$ and let $H$ be a graph containing at least $f_0(t)$ vertex-disjoint triangles $T_1,\ldots,T_{f_0(t)}$ for some fixed integer $t$. For every $i\in [f_0(t)]$, let $(a_i,b_i,c_i)$ be an arbitrary ordering of $V(T_i)$. We now consider the complete graph $K$ whose vertex set is $[f_0(t)]$ together with the edge coloring $\phi:E(K)\rightarrow 2^{[3]}$ where an edge $ij$ is mapped to a set $X \subseteq [3]$ such that $X$ contains 1 if $a_ia_j \in E(H)$, 2 if $b_ib_j \in E(H)$ and 3 if $c_ic_j \in E(H)$. By Theorem \ref{ramsey}, there is a set $S$ of $t^3$ vertices in $V(K)$ such that all edges in $K[S]$ are mapped to the same set $X$ by $\phi$. By symmetry, we may suppose that $S=[t^3]$.

 First suppose that $X \neq \emptyset$, say $1 \in X$. In this case, we obtain that $H[\{a_1,\ldots,a_t\}]$ forms a clique, hence $K_t$ is a projection of $H$.
We may hence suppose from now on that $X=\emptyset$. We now define several sets of indices. For $i\in [t]$, we let $I_i^1=\{(i-1)t^2+1,\ldots,it^2\}$, $I_i^2=\{(i-1)t+1,\ldots,it\}\cup \{t^2+(i-1)t+1,\ldots,t^2+it\}\cup \ldots \cup \{(t-1)t^2+(i-1)t+1,\ldots,(t-1)t^2+it\}$ and $I_i^3=\{i,t+i,\ldots,t^3-t+i\}$. Observe that for $i\in[3]$, we have that $(I_1^i,\ldots,I_t^{i})$ is a partition of $[t^3]$ and that for $i,j,k \in \{1,\ldots,t\}$, we have that $I_i^1\cap I_j^2\cap I_k^3 \neq \emptyset$. We now define a projection $H'$ of $H$ in the following way: For every $i \in [t]$, we contract all the vertices in $\{a_j:j \in I_i^1\}$ into a new vertex $a_i'$,  we contract all the vertices in $\{b_j:j \in I_i^2\}$ into a new vertex $b_i'$ and we contract all the vertices in $\{c_j:j \in I_i^3\}$ into a new vertex $c_i'$. Observe that as $X=\emptyset$, we obtain that $H'$ is indeed a projection of $H$. 

Further, as $\{a_1,\ldots a_{t^3}\}$ is a stable set in $H$, we obtain that $\{a_1',\ldots,a_t'\}$ is a stable set in $H'$. Similarly, $\{b_1',\ldots,b_t'\}$ and $\{c_1',\ldots,c_t'\}$ are stable sets in $H'$. Now consider some $i,j \in \{1,\ldots,t\}$. As $I_i^1\cap I_j^2\neq \emptyset$, there is some $\mu \in I_i^1\cap I_j^2$. As $a_\mu b_\mu\in E(H)$, we obtain that $a'_ib'_j \in E(H')$. Similarly, we obtain $a'_ic'_j$ and $b'_ic'_j$ are in $E(H')$. Hence $H'$ is isomorphic to $K_{t,t,t}$. This finishes the proof.
\end{proof}

As in the proof of Lemma \ref{triisi}, cographs will play a crucial role throughout the rest of the proof. The following result shows that a graph that is far away from being a cograph satisfies the condition of Lemma \ref{tri}. For the remaining cases, we will benefit from the structure of cographs.

\begin{lem}\label{delco}
Let $H$ be a graph and $t$ a positive integer. Then either $H$ has a projection that contains $t$ vertex-disjoint triangles or a cograph can be obtained from $H$ by deleting at most $4t$ vertices.
\end{lem}
\begin{proof}
Let $P^1,\ldots,P^q$ be a maximal collection of vertex-disjoint induced  copies of $P_4$ in $H$. If $q\leq t$, observe that $H-\bigcup_{i=1}^qV(P^{i})$ is a cograph due to the maximality of $P^1,\ldots,P^q$. Further, we have $|\bigcup_{i=1}^qV(P^{i})|=4q \leq 4t$, so we are done.

If $q \geq t$, observe that for $i \in [t]$, we have that $P_i$ is a nontrivial pattern. By Lemma \ref{triisi}, we obtain that for $i \in [t]$, the triangle is a projection of $P_i$. It follows that $H$ has a projection that  contains $t$ vertex-disjoint triangles. 
\end{proof}
The next result is an immediate consequence of the definition of cographs.
\begin{prop}\label{coco}
Every induced subgraph of a cograph is a cograph.
\end{prop}
We now conclude the following result which allows us to only consider bipartite cographs.
\begin{lem}\label{delcobi}
Let $H$ be a cograph and $t$ a positive integer. Then either $H$  contains $t$ vertex-disjoint triangles or a bipartite cograph can be obtained from $H$ by deleting at most $3t$ vertices.
\end{lem}
\begin{proof}
 Let $T_1,\ldots,T_q$ be a maximal collection of vertex-disjoint triangles in $H$. If $q \geq t$, there is nothing to prove. If $q\leq t$, observe that $H-\bigcup_{i=1}^qV(T_i)$ does not contain a triangle due to the maximality of $T_1,\ldots,T_q$. By Proposition \ref{biptri}, we obtain that $H-\bigcup_{i=1}^qV(T_i)$ is bipartite. Next, by Propostion \ref{coco}, we obtain that $H-\bigcup_{i=1}^qV(T_i)$ is a cograph. Finally, we have $|\bigcup_{i=1}^qV(T_i)|=3q \leq 3t$, so we are done.
\end{proof}

We say that a component of a graph is \emph{nonsingular} if it contains at least two vertices. The following result allows us to restrict ourselves to graphs with few nonsingular components.

\begin{lem}\label{fewco}
Let $H$ be a graph that contains at least $3t$ nonsingular components for some positive integer $t$. Then $H$ has a projection that contains $t$ vertex-disjoint triangles.
\end{lem}
\begin{proof}
Let $C_1,\ldots,C_{3t}$be a collection of nonsingular components of $H$. For $i\in [3t]$, let $u_i,v_i \in V(C_i)$ such that $u_iv_i \in E(H)$. For $i \in [t]$, observe that $H[\{u_{3i-2},u_{3i-1},u_{3i},v_{3i-2},v_{3i-1},v_{3i}\}]$ is a nontrivial pattern. It follows by Lemma \ref{triisi} that the triangle is a projection of $H[\{u_{3i-2},u_{3i-1},u_{3i},v_{3i-2},v_{3i-1},v_{3i}\}]$. Hence  $H$ has a projection that contains $t$ vertex-disjoint triangles. 
\end{proof}
We need one more result that allows to restrict the structure of the graphs in consideration.
\begin{lem}\label{2bip}
The graph that consists of two disjoint copies of $K_{2t,2t}$ has a projection that contains $t$ vertex-disjoint triangles.
\end{lem}
\begin{proof}
Let $H$ be the graph that consists of two disjoint copies of $K_{2t,2t}$ with bipartitions $(X_1,X_2)$ and $(X_3,X_4)$, respectively. For $i=1,\ldots,4$, let $x_i^1,\ldots,x_i^{2t}$ be an arbitrary ordering of the elements of $X_i$. For $i \in [t]$, observe that $H[\{x_{1}^{2i-1},x_{1}^{2i},x_{2}^{i},x_{3}^{i},x_{4}^{i}\}]$ is a nontrivial pattern. It follows by Lemma \ref{triisi} that the triangle is a projection of $H[\{x_{1}^{2i-1},x_{1}^{2i},x_{2}^{i},x_{3}^{i},x_{4}^{i}\}]$. Hence  $H$ has a projection that contains $t$ vertex-disjoint triangles. 
\end{proof}

We are now ready to proceed to the main proof of \cref{combmain}.
\begin{proof}[Proof of \cref{combmain}]
Let $f(t)=(6t+7)f_0(t)$. We now fix some graph $H$ and positive integer $t$. We may suppose that $H$ has none of $K_t$ and $K_{t,t,t}$ as a projection. Hence, by Lemma \ref{tri} and \cref{transitiv}, we obtain that no projection of $H$ contains $f_0(t)$ vertex-disjoint triangles. By Lemma \ref{delco}, we obtain that a cograph $H'$ can be obtained from $H$ by deleting a set $F_1$ of at most $4f_0(t)$ vertices. Next, we obtain by Lemma \ref{delcobi} that a bipartite cograph $H''$ can be obtained from $H'$ by deleting a set $F_2$ of at most $3f_0(t)$ vertices from $H'$. By Lemma \ref{cobi}, we obtain that every connected component of $H''$ is either an isolated vertex or a complete bipartite graph. Let $C_1,\ldots,C_q$ be the nonsingular connected components of $H''$ and let $(X_i,Y_i)$ be the bipartition of $C_i$ for $i=1,\ldots,q$. By symmetry, we may suppose that $|X_i|\geq |Y_i|$ for $i=1,\ldots,q$ and $|Y_1|=\max_{i \in \{1,\ldots,q\}}|Y_i|$. By Lemma \ref{fewco}, as no projection of $H$ contains $f_0(t)$ vertex-disjoint triangles and by \cref{transitiv}, we obtain that $q \leq 3f_0(t)$. Further, by Lemma \ref{2bip}, we obtain that $|Y_i|\leq 2t$ for $i=2,\ldots,q$. Now let $F=F_1\cup F_2 \cup Y_2 \cup \ldots \cup Y_q$. We obtain that $|F|=|F_1|+|F_2|+\sum_{i=2}^q|Y_i|\leq 4f_0(t)+3f_0(t)+3f_0(t)2t=f(t)$. Further, by construction, we have that $H-F$ is an extended biclique. This finishes the proof.
\end{proof}


\section{Lower Bounds}
In this section, we prove the lower bounds appearing in our main results. First, for the W[1]-hardness proofs, we need an argument that formalizes that a pattern $H$ is always at least as hard as its projection $H'$ (\cref{sec:projection}). Then we revisit the grid gadgets of Marx~\cite{Marx12} and observe that they have even stronger properties than used in the original paper, which will turn out useful for our purposes (\cref{sec:gridgadgets}). The lower bounds for \tcut{3}\ and \gcut\ appear in Sections~\ref{sec:3TcutByGenus} and \ref{sec:3groupMC}, respectively.

As usual in case of hardness results, we prove the lower bounds for the decision version of the problem. That is, we assume that an instance of \mc\ is a triple $(G,H,\lambda)$, and the task is to decide if there is a solution of weight at most $\lambda$.
\subsection{Reductions Using Projections of Patterns}

Recall that $H'$ is a projection of $H$ if $H'$ can be obtained from $H$ by repeatedly deleting vertices and identifying nonadjacent vertices. It is not difficult to see that in this case a $\mc$ instance with demand graph $H'$ can be reduced to an instance with demand graph $H$: intuitively, terminals deleted from $H$ can be added as isolated vertices and terminals arising from identification can be replaced with the original terminals connected with many edges. In this section, we formally state and prove this reduction. 

First we  handle the deletion of a single vertex.
\begin{prop}\label{del}
Let $(G',H',\lambda)$ be an instance of $\mc$ and $H$ be a graph from which $H'$ can be obtained by deleting a vertex. Then in time $f(H)(|V(G')|+|E(G')|)$, we can compute an instance $(G,H,\lambda)$ of $\mc$ that is equivalent to $(G',H',\lambda)$ such that $|V(G)|+|E(G)|\leq f(H)(|V(G')|+|E(G')|)$. Furthermore, $G$ has the same orientable genus as $G'$ and has the same maximum edge weight as $G'$.\end{prop}
\begin{proof}
  First, for every $v \in V(H)$, we check whether $H-v$ is isomorphic to $H'$ using a brute force approach, which can be done in time $f(H)$. By assumption, we find a vertex $v_0 \in V(H)$ and an isomorphism $\tau:V(H)-v_0 \rightarrow V(H')$. We now create $G$ from $G'$ by relabeling $v$ by $\tau^{-1}(v)$ for all $v \in V(H')$ and adding $v_0$ as an isolated vertex. It is easy to see that $(G,H,\lambda)$ and $(G',H',\lambda)$ are equivalent instances of $\mc$. Moreover, we have $|V(G)|+|E(G)|=|V(G')|+|E(G')|+1$. Finally, as $G$ can be obtained from a graph isomorphic to $G'$ by adding an isolated vertex, we obtain that $G$ has the same orientable genus as $G'$.
  \end{proof}

Next we consider the operation of identifying two nonadjacent vertices.
\begin{prop}\label{ident}
  Let $(G',H',\lambda)$ be an instance of $\mc$ and $H$ be a graph from which $H'$ can be obtained by identifying two nonadjacent vertices. Then in time $f(H)(|V(G')|+|E(G')|)$, we can compute an instance $(G,H,\lambda)$ of $\mc$ that is equivalent to $(G',H',\lambda)$ such that $|V(G)|+|E(G)|\leq f(H)(|V(G')|+|E(G')|)$.
  Furthermore, $G$ has the same orientable genus as $G'$ and has the same maximum edge weight as $G'$.
  \end{prop}
\begin{proof}
  First, for every pair of nonadjacent vertices $u,v \in V(H)$, we check whether the graph obtained from $H$ by identifying $u$ and $v$ is isomorphic to $H'$ using a brute force approach, which can be done in $f(H)$. By assumption, we find two nonadjacent vertices $u_0,v_0 \in V(H)$ and an isomorphism $\tau:V(H_0) \rightarrow V(H')$ where $H_0$ is the graph obtained from $H$ by identifying $u_0$ and $v_0$ into a vertex $w_0$.

  Let $W$ be the maximum weight of an edge in $G'$. If $|E(G')|W \le \lambda$, then $(G',H',\lambda)$ is a yes-instance: $E(G')$ forms a multicut of weight at most $\lambda$. Thus we may assume that $\lambda<|E(G')|W$. We now create $G$ from $G'$ by

  \begin{itemize}
  \item relabeling $v$ by $\tau^{-1}(v)$ for all $v \in V(H')-w_0$,
  \item relabelling $w_0$ to $u_0$,
  \item introducing a vertex $v_0$, and
  \item adding $|E(G')|$ paths of length two between $u_0$ and $v_0$, with edges of weight $W$.
  \end{itemize}
  
Let us observe that for any multicut $S$ for $(G,H)$ of weight at most $\lambda<|E(G')|W$, vertices $u_0$ and $v_0$ are in the same connected component of $G\setminus S$. Now it is easy to see that $(G,H,\lambda)$ and $(G',H',\lambda)$ are equivalent instances of $\mc$. Moreover, we have $|V(G)|+|E(G)|\leq |V(G')|+1+2|E(G')|$. Finally, as $G$ can be obtained from a graph isomorphic to $G'$ by attaching a planar graph to a single vertex, we obtain that $G$ has the same orientable genus as $G'$.
\end{proof}
If $H'$ is a projection of $H$, then the reduction can be obtained by multiple applications of \cref{del} and \ref{ident}.
\begin{lem}\label{hfix}
  Let $(G',H',\lambda)$ be an instance of $\mc$ and $H$ be a graph such that $H'$ is a projection of $H$. Then in time $f(H)(|V(G')|+|E(G')|)$, we can compute an instance $(G,H,\lambda)$ of $\mc$ that is equivalent to $(G',H',\lambda)$ and such that $|V(G)|+|E(G)|\leq f(H)(|V(G')|+|E(G')|)$. Furthermore, $G$ has the same orientable genus as $G'$ and has the same maximum edge weight as $G'$.
  \end{lem}
\begin{proof}
  We first compute a sequence $(H_0,\ldots,H_q)$ of graphs such that $H_0=H,H_q=H'$ and for every $i \in [q-1]$, we have that $H_{i+1}$ can be obtained from $H_i$ by deleting a vertex or identifying two nonadjacent vertices. It is clear that such a sequence can be found in time $f(H)$ by brute force.
  
  We set $G_q=G'$. Iteratively for $i =q-1, q-2, \dots, 1$, we transform instance $(G_{i+1},H_{i+1},\lambda)$ into an equivalent instance  $(G_{i},H_{i},\lambda)$ such that $|V(G_{i})|+|E(G_{i})|\leq f(H_{i})(|V(G_{i+1})|+|E(G_{i+1})|)$, and furthermore $G_{i}$ has the same orientable genus as $G_{i+1}$. By Propositions~\ref{del} and \ref{ident}, this can be done in time $f(H_{i})(|V(G_{i+1})|+|E(G_{i+1})|)$. The statement then follows for $(G_0,H_0,\lambda)$.
\end{proof}
We use \cref{hfix} to argue that considering the projection closure of $\calH$ does not make the problem harder from the viewpoint of fixed-parameter tractability parameterized by $t$.

\lemredclosure*
\begin{proof}
Let $(G',H')$ be an instance of \mcut{\calH'}. 
In a first step, we compute a graph $H \in \mathcal{H}$ that contains $H'$ as a projection. In order to do this, we start enumerating every simple graph of at least $|V(H')|$ vertices in a nondecreasing order of the number of vertices. For each graph $H$, we check if $H\in\calH$ (which is possible as $\mathcal{H}$ is computable), and, if this is the case, we test whether $H'$ is a projection of $H$ (which can be done in time depending on the size of $H$).

Let $q(H')$ be the smallest number of vertices of a graph $H\in\calH$ that contains $H'$ as a projection. Then the enumeration considers graphs only of at most $q(H')$ vertices, thus we find $H$ in time $f(q(H'),\calH)$. We can now use Lemma~\ref{hfix} to produce an equivalent instance $(G,H)$ of \mcH. The total running time of the reduction can be bounded as $g(H',\mathcal{H})(|V(G')|+|E(G')|)$. Observe that the number $t$ of terminals in the constructed instance $(G,H)$ is $|V(H)|=q(H')$, thus it can be bounded by a function of the number $t'=|V(H')|$ of terminals in $(G',H')$. Thus, for a fixed class $\calH$, this is indeed a correct parameterized reduction. The claim about genus and the maximum edge weights follow from the statement of Lemma~\ref{hfix}.
\end{proof}

\label{sec:projection}

\subsection{Revisiting Grid Gadgets}\label{sec:gridgadgets}
The goal of this section is to prove~\cref{lem:gadget_main_new}. This result is about the properties of so-called grid gadgets that were introduced in~\cite{Marx12}. These gadgets were also used in \cite{Cohen-AddadVMM21}, where they were referred to as cross gadgets. We show a strengthening of the previous analysis of these gadgets.
In order to state the result, we first need to define grid gadgets and some corresponding terminology.

Let $\Delta$ be an integer, and let $S\subseteq [\Delta]^2$.
We revisit the construction of the \emph{grid gadget} $G_S$ as it was introduced in \cite[Section 3]{Marx12}. 
We do not repeat every detail of the construction and instead focus on those properties that are relevant in order to show \cref{cor:gadget_stronger}.
However, let us emphasize that we do not alter the original construction; we only refine the analysis.

\begin{figure}[t]
	\centering
	\includegraphics[scale=0.35]{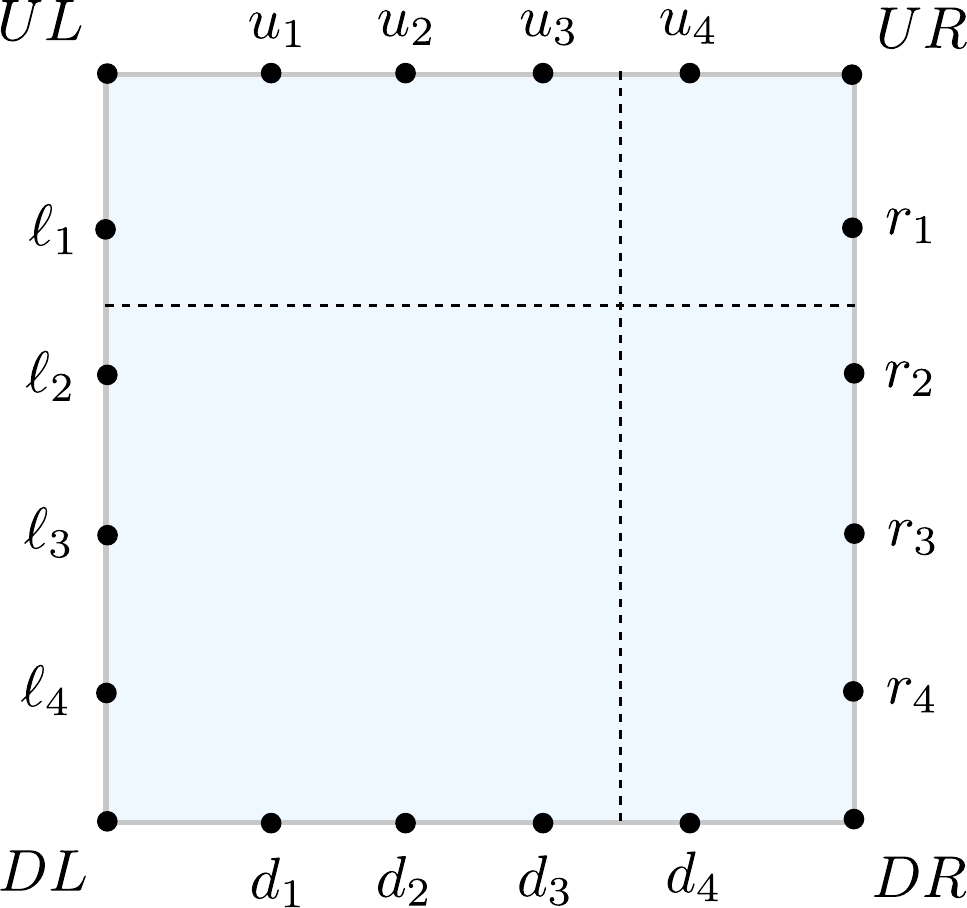}
	\caption{A grid gadget $G_S$ for $\Delta=3$. The dashed line denotes a multiway cut that represents the pair $(1,3)$.}
	\label{fig:grid gadget}
\end{figure}

\paragraph*{\boldmath The gadget $G_S$}
For $N=\Delta^2+2\Delta + 1$, $G_S$ has $(N+1)^2$ vertices $g[i,j]$ for $i,j\in \{0,\ldots, N\}$ that form a grid, i.e., $g[i,j]$ and $g[i',j']$ are adjacent if and only if $\abs{i-i'}+\abs{j-j'}=1$.
It is useful to give special names to some vertices on the boundary of the grid. To this end, for $r,s\in [\Delta+1]$, let $\alpha(s)=N-2-\Delta+s$ and let $\beta(r,s)=r+\Delta\cdot s$.
\begin{itemize}
	\item The corners are $UL=g[0,0]$, $UR=g[0,N]$, $DL=g[N,0]$, and $DR=g[N,N]$.
	\item For $s\in [\Delta+1]$, $\ell_s=g[\alpha(s),0]$ and $r_s=g[\alpha(s),N]$ are some special vertices in the leftmost and in the rightmost column, respectively.
	\item For $s\in [\Delta+1]$, $u_s=g[0,\beta(1,s)]$ and $d_s=g[N,\beta(1,s)]$ are some special vertices in the uppermost and in the downmost row, respectively.
\end{itemize}
We refer to the vertices $UL, u_1, \ldots, u_{\Delta+1}, UR$ as the distinguished vertices of the $U$ side of $G_S$. Analogously, we define the distinguished vertices of the $D$ side, $R$ side, and $L$ side.

Let $W=100N^2$. The following edge weights are used.
For a vertical edge $g[i,j], g[i+1,j]$ ($i\in \{0, \ldots, N-1\}$, $j\in \{0, \ldots, N\}$), we set
\[
\begin{cases}
	\infty & \text{if }i=j-1 \quad\text{(the diagonal)},\\
	\infty & \text{if }i\notin [\alpha(1),\alpha(\Delta)] \text{  and }j=0 \quad\text{(unbreakable part of leftmost column)},\\
	W^3+W^2 & \text{if }i\in [\alpha(1),\alpha(\Delta)] \text{ and }j=0 \quad\text{(breakable part of leftmost column)},\\
	\infty & \text{if }i\notin [\alpha(1),\alpha(\Delta)] \text{ and }j=N \quad\text{(unbreakable part of rightmost column)},\\
	W^3+W^2 & \text{if }i\in [\alpha(1),\alpha(\Delta)] \text{ and }j=N \quad\text{(breakable part of rightmost column)},\\
	W^2 & \text{otherwise.}
\end{cases}
\] 
For a horizontal edge $g[i,j], g[i,j+1]$ ($i\in \{0, \ldots, N\}$, $j\in \{0, \ldots, N-1\}$), we set
\[
\begin{cases}
	\infty & \text{if }i=0 \text{ and }j\notin [\beta(1,1), \beta(\Delta,\Delta)]\quad\text{(unbreakable part of uppermost row)},\\
	W^3+W^2+jW & \text{if }i=0 \text{ and }j\in [\beta(1,1), \beta(\Delta,\Delta)]\quad\text{(breakable part of uppermost row)},\\
	W^2+jW & \text{if }0<i<j \quad\text{(topright triangle, weight increasing to the right)},\\
	W^3+W^2-i^2W-(N-i)^2W & \text{if }1\le i=j\le N-1 \quad\text{(the diagonal)},\\
	W^2+(N-j)W & \text{if }j<i<N \quad\text{(bottomleft triangle, weight increasing to the left)},\\
	W^3+W^2+(N-j)W & \text{if }i=N \text{ and }j\in [\beta(1,1), \beta(\Delta,\Delta)]\quad\text{(breakable part of bottommost row)},\\
	\infty & \text{if }i=N \text{ and }j\notin [\beta(1,1), \beta(\Delta,\Delta)]\quad\text{(unbreakable part of bottommost row)}	
\end{cases}
\] 

In addition to the grid edges, the gadget $G_S$ contains, for each $s\in [\Delta]$, an \emph{upper ear edge} $u_s u_{s+1}$ with weight $W^3$ as well as a \emph{lower ear edge} $d_sd_{s+1}$, also with weight $W^3$.
For each $0\le i,j\le N-1$, there are some additional edges between the vertices inside a \emph{cell} 
\[
C[i,j]=\{g[i,j], g[i,j+1], g[i+1, j], g[i+1,j+1]\}.
\]
We do not specify these inner-cell edges since they are irrelevant to our arguments.
However, let us note that these inner cell edges are used to encode the set $S$. So here is where the set $S$ comes into play, and different $S$ lead to different gadgets. Instead of going through the details of the construction we can stand on the shoulders of earlier results from~\cite{Marx12}.
 
The path $UL=g[0,0]$, $g[0,1]$, $g[1,1]$, $g[1,2],\ldots, g[N-1, N-1]$, $g[N-1, N]$, $g[N,N]=DR$ is the \emph{diagonal path} of $G_S$. 
The following observation will be key in our arguments.
\begin{obs}\label{obs:key}
	Every path between $DL$ and $UR$ contains a vertex on the diagonal path.
\end{obs}

For each $i,j\in \{0, \ldots, N\}$, the edges of the (horizontal) path $g[i, 0], g[i,1], \ldots, g[i,N]$ form the $i$th \emph{row} $R_i$ of $G_S$; and the edges of the (vertical) path $g[0, j],g[1,j], \ldots, g[N,j]$ form the $j$th \emph{column} $C_j$ of $G_S$.

\paragraph*{\boldmath Properties of $G_S$}

A \emph{$4$-corner cut} of $G_S$ is a subset $M$ of its edges such that each connected component in $G_S-M$ contains at most one of the four corners $UL$, $UR$, $DL$, and $DR$. A $4$-corner cut \emph{represents} a pair $(x,y)\in [\Delta]^2$ if $G_S-M$ has exactly four connected components, each of which contains precisely one of the sets 
\begin{itemize}
	\item $\{UL, u_1, \ldots, u_y, \ell_1, \ldots, \ell_x\}$,
	\item $\{UR, u_{y+1}, \ldots, u_{\Delta+1}, r_1, \ldots, r_x\}$,
	\item $\{DL, d_1, \ldots, d_y, \ell_{x+1}, \ldots, \ell_{\Delta+1}\}$, and
	\item $\{DR, d_{y+1}, \ldots, d_{\Delta+1}, r_{x+1}, \ldots, r_{\Delta+1}\}$.
\end{itemize}

\noindent Let $W^*=7W^3+(2N+2)W^2+4(2N-3)+10$.
The following observation will be important later.
\begin{obs}\label{obs:Wstar}
	$W^*$ depends on $\Delta$ but it does not depend on $S$.
\end{obs}

We now state the previously observed properties of $G_S$.
\begin{lem}[{\cite[Lemma 2]{Marx12}}] \label{lem:gadget_orig}
	Given a subset $S\subseteq [\Delta]^2$, the grid gadget $G_S$ can be constructed in time polynomial in $\Delta$, and it has the following properties:
	\begin{enumerate}
		\item For every $(x,y)\in S$, the gadget $G_S$ has a $4$-corner cut of weight $W^*$ representing $(x,y)$.
		\item If a $4$-corner cut of $G_S$ has weight $W^*$, then it represents some $(x,y)\in S$.
		\item Every $4$-corner cut of $G_S$ has weight at least $W^*$.
	\end{enumerate}
\end{lem}

We show that an analogous result also holds if we relax the requirements of the cuts that we consider. A \emph{good cut} of $G_S$ is a subset $M$ of its edges such that the connected components of $G_S-M$ that contain one of $UL$ or $DR$ do not contain any other corners of $G_S$. Observe that a good cut does not necessarily separate $DL$ and $UR$, but separates all other pairs of corners. Therefore, every $4$-corner cut is a good cut but not necessarily the other way around. We will show that, for minimum weight cuts, the reverse direction also holds by construction of $G_S$.
We first show an auxiliary result.
\begin{lem}\label{lem:goodcutforGS}
	Given a subset $S\subseteq [\Delta]^2$, every good cut of $G_S$ of weight at most $W^*$ is also a $4$-corner cut of $G_S$.
\end{lem}
\begin{proof}
	Suppose that $M$ is a good cut of $G_S$ of weight at most $W^*$. 
	Let $\KUL$ be the connected component in $G_S-M$ that contains $UL$. Similarly, we define $\KUR$, $\KDL$, and $\KDR$. Since $M$ is a good cut $\KUL$ and $\KDR$ are distinct from the others, while $\KDL=\KUR$ is allowed.
	We observe the following strengthening of \cite[Claim 7]{Marx12}.
	\begin{clm}\label{clm:oneedge}
		$M$ contains exactly one upper ear edge, exactly one lower ear edge, exactly one edge of the diagonal path, and exactly one edge from each of $C_0$, $C_N$, $R_0$, and $R_N$.
	\end{clm}
	\begin{claimproof}
		The proof is almost verbatim from the original~\cite[Claim 7]{Marx12}. The original proof used the fact that the diagonal path as well as the boundaries of the grid each connect two of the corners; and therefore need to be cut. For this argument it suffices to use the fact that both $UL$ and $DR$ need to be separated from the other corners.
		
		For convenience, here are the details:
		Since $u_1$ is in $\KUL$, $u_{\Delta+1}$ is in $\KUR$, and $\KUL\neq \KUR$ the cut $M$ contains at least one upper ear edge.
		Similarly, it contains at least one lower ear edge. 
		The fact that the cut $M$ contains at least one edge from the diagonal path and at least one edge from each of $C_0$, $C_N$, $R_0$, and $R_N$ follows from the fact that each of these sets connects one of $UL$ or $DR$ to some other corner; and a good cut disconnects $UL$ and $DR$ from the other corners. These seven edges are distinct since the mentioned edge sets intersect only in edges with infinite weight. Each of the upper/lower ear edges, and each of the edges from the diagonal path, as well as  $C_0$, $C_N$, $R_0$, and $R_N$ have weight at least $W^3$. Since $W^*<8W^3$ it follows that at most one edge from each of these sets is in $M$.
	\end{claimproof}
	
	From \cref{clm:oneedge} we directly obtain the following.
	\begin{clm}\label{clm:diagonal}
		Each vertex of the diagonal path is either in $\KUL$ or in $\KDR$.
	\end{clm}
	
	To show that the good cut $M$ is a $4$-corner cut it suffices to show that $\KDL\neq \KUR$.
	From \cref{obs:key} we know that each path between $DL$ and $UR$ contains a vertex of the diagonal path. The statement then follows from \cref{clm:diagonal} together with the fact that $\KUL$ and $\KDR$ are distinct from both $\KDL$ and $\KUR$ by definition of a good cut.
\end{proof}

From \cref{lem:gadget_orig,lem:goodcutforGS} we directly obtain \cref{lem:gadget_main_new} as strengthening of \cref{lem:gadget_orig}. We restate it here for convenience.
\goodcutmain*

\subsection{$\tcut{3}$ Parameterized by Genus} \label{sec:3TcutByGenus}
In this section we study $\tcut{3}$, i.e., the setting where the demand pattern is a triangle.
Our goal is to show \cref{thm:3TCut}. As a consequence, we extend known hardness results for four terminals, see \cite[Proposition 4.1]{Cohen-AddadVMM21}.

\threeterm*

\Cref{lem:threeterm} follows from the original proof of \cite[Proposition 4.1]{Cohen-AddadVMM21} with only a few simple modifications. 
Intuitively, that original proof uses the grid gadgets from~\cref{sec:gridgadgets} and, from all introduced grid gadgets, identifies the corners with the same label to form four terminals. This ensures that the overall $4$-terminal cut induces a $4$-corner cut in each grid gadget used in the construction. However, in \cref{cor:gadget_stronger}, we observe that a \emph{good cut} suffices to ensure the properties of a grid gadget that make the reduction work. A good cut does not necessarily disconnect $UR$ and $DL$. Thus, by identifying $UR$ and $DL$ with the same terminal, a corresponding $3$-terminal cut then induces a good cut on each of the grid gadgets. Identifying two vertices (the two terminals) increases the orientable genus by at most $1$.
For convenience of the reader, we give the proof in detail. The proof uses a reduction from the following problem.

\vbox{
	\begin{description}\setlength{\itemsep}{0pt}
		\setlength{\parskip}{0pt}
		\setlength{\parsep}{0pt}   			
		\item[\bf Name:] $4$-\textsc{Regular Graph Tiling} 
		\item[\bf Input:]   A tuple $(k,\Delta, \Gamma, \{S_v\})$, where
			\begin{itemize}
				\item $k$ and $\Delta$ are positive integers,
				\item $\Gamma$ is a $4$-regular graph (parallel edges allowed) on $k$ vertices in which the edges are labeled $U$, $D$, $L$, $R$ in a way that each vertex is adjacent to exactly one of each label,
				\item for each vertex $v$ of $\Gamma$, we have $S_v\subseteq [\Delta]\times [\Delta]$.
			\end{itemize}
		\item[\bf Output:]  For each vertex $v$ of $\Gamma$, a pair $s_v\in S_v$ such that, if $s_v=(i,j)$,
			\begin{itemize}
				\item the first coordinates of $s_{L(v)}$ and $s_{R(v)}$ are both $i$, and
				\item the second coordinates of $s_{U(v)}$ and $s_{D(v)}$ are both $j$,
			\end{itemize}
			where $U(v)$, $D(v)$, $L(v)$, and $R(v)$ denote the vertex of $\Gamma$ that is connected to $v$ via the edge labeled by $U$, $D$, $L$, and $R$, respectively.
	\end{description}
}

\begin{proof}[{Proof of \cref{lem:threeterm}}]
	We define a class $\mathcal{I}$ of instances of $\tcut{4}$ that contains, for every $4$-\textsc{Regular Graph Tiling} instance $(k, \Delta, \Gamma, \{S_v\})$ with bipartite graph $\Gamma$, the instance $(G,T,\lambda)$ of $\tcut{4}$ as it was defined in \cite[Proof of Prop. 4.1]{Cohen-AddadVMM21}:
	\begin{enumerate}[(1)]
		\item For each vertex $v$ of $\Gamma$, we create a grid gadget $G_S(v)$ such that the set $S\subseteq [\Delta]^2$ is chosen to be equal to $S_v$.
		\item For each edge $e=uv$ of $\Gamma$ labeled $U$, we identify pairwise the distinguished vertices of the $U$ side of the grid gadget $G_S(v)$ with the distinguished vertices of the $U$ side of the grid gadget $G_S(u)$ that have the same label.
		 Similarly, if an edge is labeled $D$, $R$, and $L$ then the distinguished vertices on the $D$, $R$, and $L$ sides, respectively, are identified. This creates multi-edges as we identify only the vertices but not the edges.
		\item The four corner vertices $UL$, $UR$, $DL$, and $DR$ of all the grid gadgets are identified in four vertices $UL$, $UR$, $DL$, $DR$ that act as four terminals.
		\item Let $\lambda\coloneqq k\cdot W^*$, where $W^*$ is the value from \cref{cor:gadget_stronger} (which is the same value as $D_1$ in \cite[Lemma 2.13]{Cohen-AddadVMM21}).
	\end{enumerate}
	Let $n$ denote the number of vertices of $G$.
	Technically, the grid gadgets are weighted graphs, but all weights are polynomial in $\Delta$ and integer. Therefore, an edge with weight $w$ for some positive integer $w$ can  be replaced by $w$ parallel edges (or paths to avoid multi-edges), see the explanation in~\cite[Section 1]{Marx12}.

	\begin{clm}[{\cite[Proof of Prop. 4.1]{Cohen-AddadVMM21}}]\label{clm:CA0}
		On the class of instances $\mathcal{I}$, the unweighted $\tcut{4}$ problem is \W{1}-hard parameterized by the orientable genus $\go$ of $G$. 
		Moreover, there is a universal constant $\alpha>0$ such that it cannot be solved in time $\bigO(n^{\alpha_{\textup{MC}}\cdot (\go+1)/\log{(\go+2)}})$, unless the $\ETH$ fails.
	\end{clm}
	
	We will also need the following previously established fact.
	\begin{clm}[{\cite[Proof of Prop.~4.1]{Cohen-AddadVMM21}}]\label{clm:CA1}
		If we choose, for each vertex $v$ of $\Gamma$, a $4$-corner cut of $G_S(v)$ of weight $W^*$, then the union of all of these cuts is a $4$-terminal cut of $(G,T)$.
	\end{clm}

	In order to go from four terminals to three, we modify an instance $(G,T)$ by identifying the terminals $UR$ and $DL$ into a single terminal. This forms the graph $G^*$. The three terminals in $T^*$ are $UL$, $UR=DL$, and $DR$. 
		
	\begin{clm}\label{clm:3TCutequivalence}
		The instance $(G^*,T^*)$ admits a $3$-terminal cut of weight at most $\lambda$ if and only if the corresponding original instance $(G,T)$ admits a $4$-terminal cut of weight at most $\lambda$.
	\end{clm}
	\begin{claimproof}
		For the one direction, suppose $(G,T)$ admits a $4$-terminal cut of weight at most $\lambda$. The same set of edges is a $3$-terminal cut of $(G^*,T^*)$.
		
		For the other direction, suppose $(G^*,T^*)$ admits a $3$-terminal cut $M$ of weight at most $\lambda=k\cdot W^*$. For each vertex $v$ of $\Gamma$, the restriction of $M$ to the edges of the grid gadget $G_S(v)$ has to be a \emph{good cut} as otherwise the three terminals would not be disconnected in $G^*$. Furthermore, there are $k$ such grid gadgets with pairwise disjoint edges. Consequently, each of the $k$ good cuts has weight at most $W^*$. Then, by \cref{cor:gadget_stronger}, each such good cut has weight exactly $W^*$. Consequently, also according to \cref{cor:gadget_stronger}, for each grid gadget $G_S(v)$ the corresponding good cut is a $4$-corner cut of weight $W^*$. 
		Then, according to \cref{clm:CA1}, $M$ as the union of these cuts is a $4$-terminal cut of $G$.
	\end{claimproof}
	
	\begin{clm}\label{clm:3TCutgenus}
		The orientable genus of $G^*$ is at most 1 larger than the orientable genus of $G$.
	\end{clm}
	\begin{claimproof}
		Recall that $G^*$ is obtained from $G$ by identifying one pair of vertices. This can be viewed as connecting the pair of vertices by an edge (which increases the orientable genus by at most $1$), and then contracting this edge (which does not increase the orientable genus).
	\end{claimproof}

	Thus, from \cref{clm:3TCutequivalence,clm:3TCutgenus} together with the hardness results from \cref{clm:CA0}, we obtain that $\tcut{3}$ is also \W{1}-hard parameterized by the orientable genus $\go$ of $G^*$.
	Moreover, according to \cref{clm:CA0} , there is no $\bigO(n^{\alpha_{\textup{MC}}\cdot (\go+1)/\log{(\go+2)}})$ algorithm for the considered instances of $\tcut{4}$, assuming $\ETH$. So, with \cref{clm:3TCutgenus}, we also obtain that there is no $\bigO(n^{\alpha_{\textup{MC}}\cdot (\go/\log{(\go+2)}})$ algorithm for $\tcut{3}$. Moreover, for $\alpha<\alpha_{\textup{MC}}/2$, there is no $\bigO(n^{\alpha\cdot (\go+1)/\log{(\go+2)}})$ algorithm, as required.
\end{proof}

\subsection{\gcut}\label{sec:3groupMC}
In this section, we consider the problem $\gcut$. Recall that this coincides with the problem $\mcH$ if $\calH$ is the class of all complete tripartite graphs.
This problem serves as a base case for several of our hardness results.

The goal of this section is to prove \cref{negmain}, which we restate here for convenience.
\threegroupCut*

For the proof of Theorem \ref{negmain}, we strongly rely on the following lower bound by Cohen-Addad et al.~\cite{Cohen-AddadVMM21}, which is a modification of a result of Marx~\cite{DBLP:journals/toc/Marx10}. From \cref{sec:intro}, recall our notation for instances of a CSP: $V$ refers to the set of variables, $D$ is the domain of values, and $K$ is the set of constraints.
\begin{thm}[\cite{Cohen-AddadVMM21}]\label{csphard}
Assuming the $\ETH$, there exists a universal constant $\alpha$ such that for any fixed primal graph $G$ with $\tw(G)\geq 2$, there is no algorithm deciding the $\csp$ instances $(V,D,K)$ whose primal graph is $G$ in time $O(|D|^{\alpha\tw(G)/\log(\tw(G))})$.
\end{thm}

\begin{prop}[\cite{Cohen-AddadVMM21}]\label{expander}
There is a universal constant $c_{\tw}\leq 1$ such that for every choice of $\go \geq 0$ and $s \geq 48(\go+1)$, there exists a 4-regular multigraph $P$ with orientable genus $\go$ with $|V(P)|\leq \frac{1}{12}s$ and $\tw(P)\geq \max\{2,c_{\tw}\sqrt{\go s+s}\}$.
\end{prop}
As a consequence, we obtain Theorem \ref{thm:CSP-genus-lower}, which we restate here. It will be the main tool for the proof of Theorem \ref{negmain}.

\thmCSPgenuslower*

\begin{proof}
Let $\alpha_{\textup{CSP}}$ be the universal constant given by \cref{csphard}.
Let $\alpha=\frac{1}{96}c_{\tw}\alpha_{\textup{CSP}}$ and consider some nonnegative integers $s$ and $\go$ with $s \geq \go+1$. By Proposition \ref{expander}, there exists a 4-regular multigraph $P$ embeddable on an orientable surface of genus $\lfloor \frac{1}{48}\go \rfloor$ with $|V(P)|\leq \frac{1}{12}s\leq s$ and $\tw(P)\geq \max\{2,c_{\tw}\sqrt{\lfloor \frac{1}{48}\go \rfloor s+s}\}\geq \max\{2,\frac{1}{96}c_{\tw}\sqrt{\go s+s}\}$. Suppose that there exists an algorithm $\mathbb{A}$ that solves all instances $(V,D,K)$ of $\csp$ whose primal graph is $P$ and that runs in $O(|D|^{\alpha \sqrt{\go s+s}/\log{(\go +s)}})$. If $\go=0$, we clearly have $\sqrt{\go s+s}=\sqrt{s}\leq s =\go+s$. Otherwise, we have $\sqrt{\go s+s}\leq \sqrt{2\go s}\leq 2\sqrt{\go s}\leq \go +s$, as $(\sqrt{\go }-\sqrt{s})^2\geq 0$. In either case, we have $\sqrt{\go s+s}\leq \go +s$. Together with $c_{\tw}\leq 1$, this yields:

\begin{align*}
\alpha \sqrt{\go s+s}/\log{(\go +s)}&\leq \alpha \sqrt{\go s+s}/\log{(\sqrt{\go s+s})}\\
&\leq \alpha_{\textup{CSP}}\frac{1}{96}c_{\tw}\sqrt{\go s+s}/\log{(\frac{1}{96}c_{\tw}\sqrt{\go s+s})}\\
&\leq \alpha_{\textup{CSP}}\tw(P)/\log(\tw(P)).
\end{align*}

We hence obtain by Theorem \ref{csphard} that $\ETH$ fails.
\end{proof}
The main technical part of the reduction is contained in the following lemma.
\begin{lem}\label{lemprinc}
There exists a universal constant $d \geq 1$ with the following properties. Let $(V,D,K)$ be an instance of $\csp$ whose primal graph is a four-regular graph $P$ cellularly embedded in a surface $N$. Then, in $O((|V||D|)^d)$, we can compute an $N$-embedded equivalent instance $(G,T_1, T_2,T_3,\lambda)$ of $\gcut$ where $|V(G)|+|E(G)|=O((|V||D|)^d)$ and $\sum_{i=1}^3 \abs{T_i} \leq 12 |V(P)|$. 
\end{lem}
\begin{proof}
Without loss of generality, we may suppose that $D=[\Delta]$ for some positive integer $\Delta$. 

In order to proceed with our reduction, we first give a construction of an auxiliary graph $Q$ and prove some of its properties. For every $v \in V$, we let $Q$ contain a 4-cycle $x_v^1\ldots x_v^4x_v^1$ and refer to this 4-cycle as the 4-cycle \emph{associated} with $v$. Next, for every $e \in E(P)$,  we let $Q$ contain a 4-cycle $y_e^1\ldots y_e^4y_e^1$ and refer to this 4-cycle as the 4-cycle \emph{associated} with $e$. Finally, for every $v \in V$ and $e \in E(P)$ that is incident to $v$, we choose (in a way that we describe momentarily) an $i\in [4]$ and a $j \in [2]$ and add the two edges $x_v^{i}y_e^{j}$ and $x_v^{i+1}y_e^{j+1}$ where $x_v^5$ refers to $x_v^1$. We refer to the 4-cycle $x_v^{i+1}x_v^{i}y_e^{j}y_e^{j+1}x_v^{i+1}$ as the 4-cycle \emph{associated} to the pair $(v,e)$. We do this in a way that for every $v \in V$, the index $i \in [4]$ is chosen according to the index of $e$ in an ordering $(e_1,\ldots,e_4)$ of the edges incident to $v$ in $P$ that can be observed when going around $v$. Further, we choose the index $j$ so that for every $e \in E(P)$, each $j \in [2]$ is chosen exactly once. This finishes the description of $Q$. We use $\mathcal{C}$ for the set of all 4-cycles of $Q$ that are associated with some $v \in V, e \in E(P)$ or pair $(v,e)$. Observe that $|\mathcal{C}|=7|V(P)|$ as $P$ is 4-regular. It follows by construction that an embedding of $Q$ in $N$ for which the interior of every $C \in \mathcal{C}$ forms a face can be obtained from an embedding of $P$. In the following, we also use $Q$ to refer to this embedded graph. An illustration is given in Figure \ref{fig:four-regular gadget}.

\begin{figure}[t]
	\centering
	\includegraphics[scale=0.35]{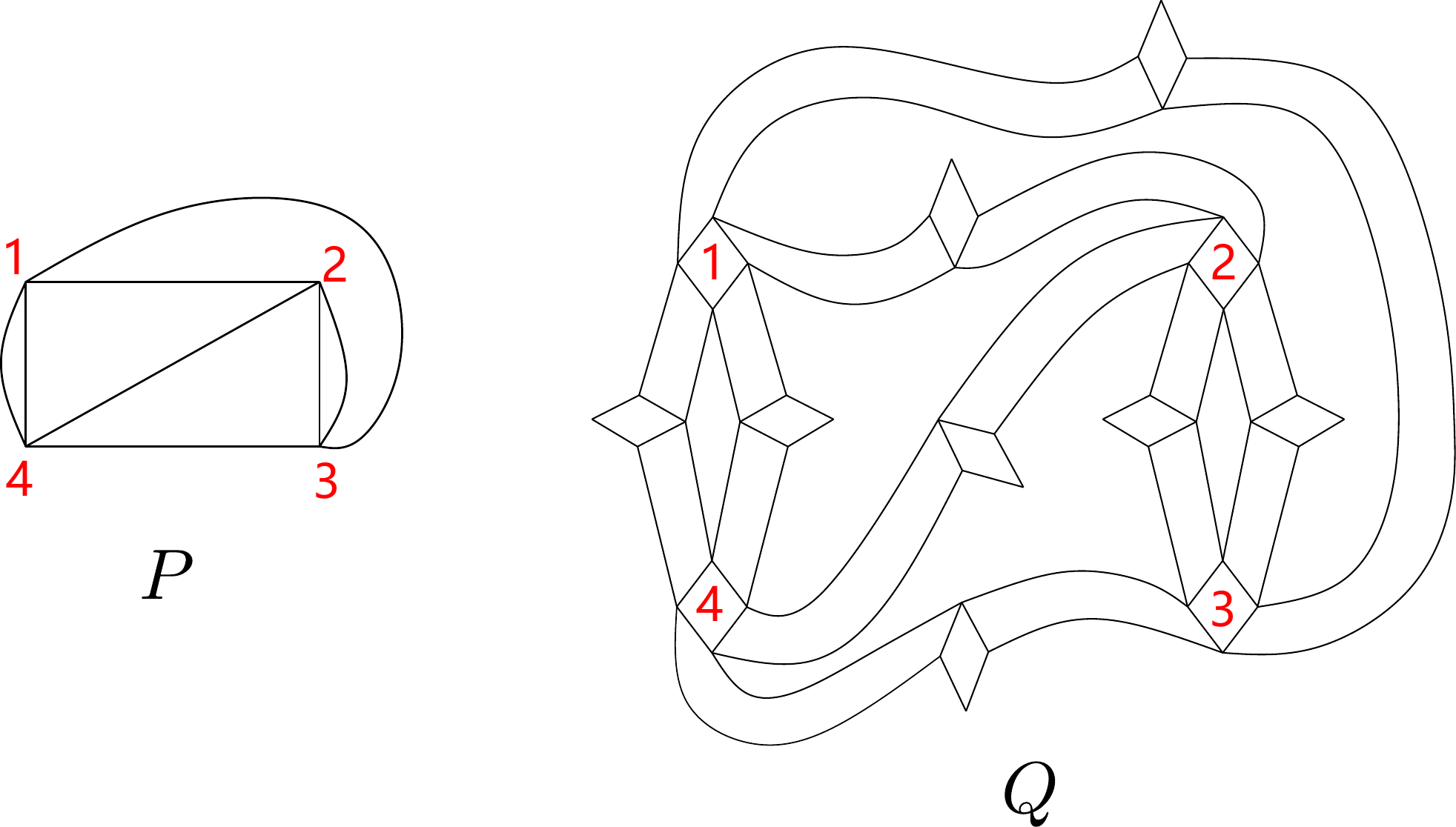}
	\caption{An example for a graph $Q$ which is created from the 4-regular graph $P$ as described in the proof of Lemma \ref{lemprinc}.}
	\label{fig:four-regular gadget}
\end{figure}

\begin{clm}
In polynomial time one can compute a proper $3$-coloring $\phi$ of the vertices of $Q$ such that all three colors appear on every cycle $C$ in $\mathcal{C}$, i.e., such that $\phi(V(C))=[3]$.
\end{clm}
\begin{claimproof}
For all $v \in V$, we set $\phi(x_v^1)=\phi(x_v^3)=1, \phi(x_v^2)=2,$ and $\phi(x_v^4)=3$. Now consider some $e \in E(P)$ and let $x_1,\ldots,x_4$ be the unique vertices in $V(Q)-\{y_e^1,\ldots,y_e^4\}$ such that $E(Q)$ contains the edges $x_1y_e^1,x_1x_2,x_2y_e^2,x_3y_e^2,x_3x_4$, and $x_4y_e^3$. We choose $\phi(y_e^2)$ to be some color in $[3]-\{\phi(x_2),\phi(x_3)\}$. Next, if $\phi(x_1)=\phi(y_e^2)$, we choose $\phi(y_e^1)$ to be the unique color in $[3]-\{\phi(x_1),\phi(x_2)\}$, otherwise we set $\phi(y_e^1)=\phi(x_2)$. Similarly,  if $\phi(x_4)=\phi(y_e^2)$, we choose $\phi(y_e^3)$ to be the unique color in $[3]-\{\phi(x_3),\phi(x_4)\}$, otherwise we set $\phi(y_e^3)=\phi(x_3)$. Finally, if $\phi(y_e^1)=\phi(y_e^3)$, we choose $\phi(y_e^4)$ to be the unique color in $[3]-\{\phi(y_e^1),\phi(y_e^2)\}$, otherwise we set $\phi(y_e^4)=\phi(y_e^2)$. It is easy to see that the coloring $\phi$ which is created this way has the required properties and can be computed in polynomial time.
\end{claimproof}

We are now ready to define the sought-for instance of $\gcut$. For every $C\in \mathcal{C}$, we let $G$ contain a grid gadget $G_{S_C}$ for a set $S_C$ which will be specified later. We now fix some $C\in \mathcal{C}$ and describe a correspondence between $V(C)$ and the corners of $G_{S_C}$. Observe that there is exactly one pair of vertices in $V(C)$ that have the same color with respect to $\phi$ and these vertices are nonadjacent. We let one of these vertices correspond to  $DL(G_{S_C})$ and the other one to $UR(G_{S_C})$. Observe that there are two possibilities to extend this correspondence to a correspondence from $V(C)$ to the corners of $G_{S_C}$. We choose this correspondence in a way such that the vertices $DL(G_{S_C}),UL(G_{S_C}),UR(G_{S_C}),DR(G_{S_C})$ appear in this order when going clock-wise around $C$ in the embedding of $C$ inherited from $Q$.

We now define $G$ from the collection of grid gadgets by first identifying the vertices in the grid gadgets corresponding to the same vertex of $Q$ and then identifying two sides of some grid gadgets whenever the vertices of $Q$ corresponding to the corners contained in these sides are linked by the same edge in $Q$. The identification of sides refers to an operation whose precise description can be found in \cite{Marx12}.
In order to finish the description of $G$, we still need to specify the sets $S_C$ for all $C\in \mathcal{C}$. 

Let $C\in \mathcal{C}$. If $C$ is associated with some $v \in V$, we choose $S$ to be the symmetric relation $\{(x,x):x \in D\}$. If $C$ is associated with $(v,e)$ for some incident $v \in V$ and $e \in E(P)$, we set $S_C=[\Delta] \times [\Delta]$. Now suppose that $C$ is associated with some $e=uv \in E(P)$, let $C'$ be the cycle in $\mathcal{C}$ associated with the pair $(u,e)$, and let $C''$ be the cycle in $\mathcal{C}$ associated with $u$. We have to deal with the fact that on the $U$ side and the $R$ side of the grid gadget the distinguished vertices are labeled clockwise, while on the $L$ side and the $U$ side they are labeled counterclockwise. We define the function $\tau_u: [\Delta] \rightarrow [\Delta]$ in the following way: if the side of $C''$ which is identified with a side of $C'$ is in the same part of $\{\{U,R\},\{D,L\}\}$ as the side of $C$ identified with a side of $C'$, we set $\tau_u(i)=\Delta+1-i$ for all $i \in [\Delta]$, otherwise we set $\tau_u(i)=i$ for all $i \in [\Delta]$. We similarly define $\tau_v$. Finally, if one of the sides $L$ or $R$ of $C$ is identified with a side of $C'$, we set $S_C$ to be the set that consists of $(\tau_u(i),\tau_v(j))$ for all $(i,j)$ contained in the relation associated with $e$. Otherwise, we set $S_C$ to be the set that consists of $(\tau_v(j),\tau_u(i))$ for all $(i,j)$ contained in the relation associated with $e$. This finishes the description of $G$. An illustration can be found in Figure \ref{fig:CSP gadget}.

\begin{figure}[t]
	\centering
	\includegraphics[scale=0.8]{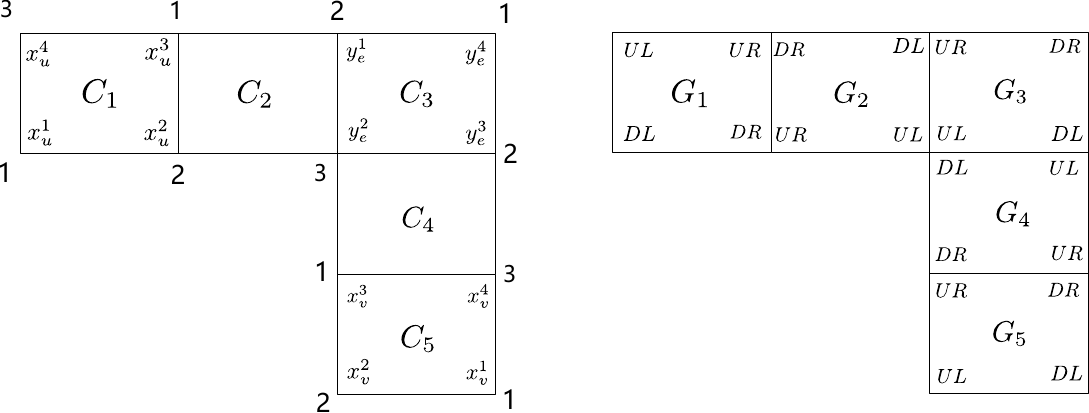}
	\caption{This drawing illustrates how a part of $G$ is constructed from a subgraph of $Q$ that corresponds to an edge $uv$ of $P$. On the left, an illustration of this subgraph can be found which contains the vertex names in $Q$ as well as a possible coloring $\phi$ of this subgraph of $Q$. Observe that $C_1,\ldots,C_5$ are the elements of $\mathcal{C}$ associated to $u,(u,e),e,(v,e)$, and $v$, respectively. On the right side a possible correponding subgraph of $G$ is depicted where $G_{S_{C_i}}$ is abbreviated to $G_i$ for $i \in [5]$ and the names of the vertices inside the cycle indicate which corner of the grid gadget is identified with the corresponding vertex. Observe that the sides $R$ of $G_1$ and $U$ of $G_3$ are identifed with sides of $G_2$, so $\tau_u(i)=\Delta+1-i$ for all $i \in [\Delta]$. Further, the sides $L$ of $G_3$ and $R$ of $G_5$ are identifed with sides of $G_5$, so $\tau_v(i)=i$ for all $i \in [\Delta]$. Finally, as the side $L$ of $G_3$ is identified with $G_4$, we set $S_C$ to be the set that consists of $(\tau_v(j),\tau_u(i))$ for all $(i,j)$ contained in the relation associated with $e$.}
	\label{fig:CSP gadget}
\end{figure}

 Then we set, for each $i\in [3]$, we set $T_i=\phi^{-1}(i)$. Further, we set $\lambda=7|V(P)|W^*$ for $W^*$ the value as given by the definition of grid gadgets, see \cref{sec:gridgadgets}. We point out that, according to \cref{obs:Wstar}, the value $W^*$ depends only on $\Delta$, and not on $S$.

\begin{clm}
$(G,T_1, T_2,T_3,\lambda)$ is a yes-instance of $\gcut$ if and only if $(V,D,K)$ is a yes-instance of $\csp$.
\end{clm}
\begin{claimproof}
First suppose that $(V,D,K)$ is a yes-instance of $\csp$, so there is a satisfying mapping $\gamma:V \rightarrow D$. We define a corresponding group cut $B$ of $(G,T_1, T_2, T_3)$. Consider some $C \in \mathcal{C}$. If $C$ is associated with $v$ or to a pair $(v,e)$ for some $v \in V$, we let $B_C$ be a good cut of $G_{S_C}$ which represents $(\gamma(v),\gamma(v))$. Now consider some $C \in \mathcal{C}$ which is associated with some $e=uv \in E(P)$. By symmetry, we may suppose that one of the sides $L$ or $R$ of $G_{S_C}$ has been identified with a side of a grid gadget one of whose sides has been identified with a side of $G_{S_{C'}}$ where $C'$ is the element of $\mathcal{C}$ associated with $u$. We then let $B_C$ be a good cut of $G_{S_C}$ that represents $(\tau_u(\gamma(u)),\tau_v(\gamma(v)))$. We now let $B=\bigcup_{C \in \mathcal{C}}B_C$. We claim that $B$ is the sought-for group cut. Indeed, it follows by the definition of good cuts that for every $C \in \mathcal{C}$, we have that the corners of $G_{S_C}$ are contained in distinct components of $G_{S_C}\setminus B_C$. We can now prove that every $v \in V(H)$ is contained in a distinct component of $G\setminus B$ using similar arguments as in \cite{Marx12}, see also \cite{Cohen-AddadVMM21}. Observe that these arguments use the fact that $\gamma$ is satisfying. Further observe that $|B|=\sum_{C \in \mathcal{C}}|B_C|=\sum_{C \in \mathcal{C}}W^*=|\mathcal{C}|D_1=7|V(P)|W^*=\lambda$. Hence $(G,T_1, T_2, T_3,\lambda)$ is a yes-instance of $\gcut$.

For the other direction, suppose that $(G,T_1, T_2, T_3,\lambda)$ is a yes-instance of $\gcut$, so there exists a group cut $B$ with $|B|\leq \lambda$. For every $C \in \mathcal{C}$, let $B_C=B\cap E(G_{S_C})$. As $B$ is a group cut and by construction, it follows that $B_C$ is a good cut of $G_{S_C}$ for all $C \in \mathcal{C}$. By \cref{cor:gadget_stronger}, we obtain that $|B_C|= W^*$ for all $c \in \mathcal{C}$. Due to the exact definition of the identification operation, we have that $B_C \cap B_{C'}=\emptyset$ for all distinct $C,C' \in \mathcal{C}$. This yields $|B|=\sum_{C \in \mathcal{C}}|B_C|\geq |\mathcal{C}|W^*=7|V(P)|W^*=\lambda\geq |B|$. Hence equality holds throughout and we have $|B_C|= W^*$ for all $C \in \mathcal{C}$. It follows from \cref{cor:gadget_stronger} that for every $C \in \mathcal{C}$, we have that $B_C$ is a good cut representing some $(i,j) \in S_C$. In particular, for every $v \in V$, there exists some $i_v \in [\Delta]$ such that $B_C$ is a good cut representing $(i_v,i_v)$ where $C$ is the element of $\mathcal{C}$ associated with $v$. We now define the function $\gamma:V \rightarrow [\Delta]$ by $\gamma(v)=i_v$ for all $v \in V$ and claim that $\gamma$ is a satisfying assignment for $(V,D,K)$. In order to prove this, consider some $e=uv \in E(P)$. By symmetry, we may suppose that one of the sides $L$ or $R$ of $G_{S_C}$ has been identified with a side of a grid gadget one of whose sides has been identified with a side of  $G_{S_{C'}}$ where $C'$ is the element of $\mathcal{C}$ associated with $u$. It follows by the exact definition of the identification of sides and by construction that $B_C$ is a good cut of $G_{S_C}$ representing $(\tau_u(\gamma(u)),\tau_v(\gamma(v)))$. In particular, we obtain that $(\gamma(u),\gamma(v))$ is contained in the relation associated with $e$. It follows that $\gamma$ is a satisfying assignment of $(V,D,K)$, so $(V,D,K)$ is a yes-instance of $\csp$.
\end{claimproof}
\end{proof}

\begin{prop}\label{negnew}

	Assuming the $\ETH$, there exists a universal constant $\alpha$ such that for any fixed choice of integers $\go \geq 0$ and $t \geq \go+1$, there is no algorithm that decides all unweighted $\gcut$ instances with orientable genus at most $\go$ and at most $t$ terminals in time $O(n^{\alpha \sqrt{\go t+t}/\log(\go+t)})$.

\end{prop}
\begin{proof}
	Let $\alpha'$ be the universal constant as given by \cref{thm:CSP-genus-lower}.
Let $\bar{t}$ be an integer such that  $(\alpha'/d)\sqrt{t}/\log(2t)\geq 1$ holds for all $t\geq \bar{t}$ and $\bar{t}\geq 24$.

Now fix two positive integers $t$ and $\go$ such that $t \geq \max\{\bar{t},(\go+1)\}$. Suppose for the sake of a contradiction that there exists an algorithm $\mathbb{A}$ that decides all instances of $\gcut$ with orientable genus at most $\go$ and at most $t$ terminals in time $O(n^{(\frac{1}{24}\alpha'/d)\sqrt{\go t+t}/\log(\go +t)})$. By \cref{thm:CSP-genus-lower}, there exists a 4-regular multigraph $P$ embeddable on a surface of genus $\go $ with $|V(P)|\leq \lfloor\frac{1}{12}t\rfloor$ and such that, unless ETH fails, there is no algorithm deciding all binary instances $(V,D,K)$ of $\csp$ whose primal graph is $P$ that runs in $O(|D|^{\alpha' \sqrt{\go \lfloor\frac{1}{12}t\rfloor+\lfloor\frac{1}{12}t\rfloor}/\log{(\go +\lfloor\frac{1}{12}t\rfloor)}})$.

Let $(V,D,K)$ be an instance of binary $\csp$ whose primal graph is $P$. As $p$ and $t$ are fixed, and $|V(P)|\leq \frac{1}{12}t$, we may compute an embedding of $P$ in a surface of genus at most $\go $ in constant time. By Lemma \ref{lemprinc}, in $O(|D|^d)$, we can compute an $N$-embedded instance $(G,H,\lambda)$ of $\mc$ where $|V(G)|+|E(G)|\in O((|V||D|)^d)$ and $H$ is a complete tripartite graph with $|V(H)| \leq 12 |V(P)|\leq t$. Note that the instance $(G,H,\lambda)$ of $\mc$ encodes an instance of $\gcut$ as $H$ is complete tripartite. Therefore, we can use algorithm $\mathbb{A}$ to solve it. As $t$ and $\go $ are constant and the choice of $P$ only depends on $t$ and $\go $, using algorithm $\mathbb{A}$, we can hence decide in time $O((|D|^d)^{(\alpha'/d)\sqrt{\go t+t}/\log(\go +t)})$ whether $(G,H,\lambda)$ is a yes-instance of $\mc$. As $t \geq 24$, we have $\lfloor\frac{1}{12}t\rfloor \geq \frac{1}{24}t$. As $t \geq \max\{\bar{t},\go \}$, this dominates $O(|D|^d)$ and hence the total running time of this procedure is $O(|D|^{\alpha'\sqrt{\go t+t}/\log(\go +t)})$ which is $O(|D|^{\frac{1}{24}\alpha'\sqrt{\go\lfloor\frac{1}{12}t\rfloor+\lfloor\frac{1}{12}t\rfloor}/\log(\go +\lfloor\frac{1}{12}t\rfloor)})$.  Further, as $(G,H,\lambda)$ and $(V,D,K)$ are equivalent, we can decide in the same running time whether $(V,D,K)$ is a yes-instance of $\csp$. We obtain that the $\ETH$ fails.

We now choose some $\alpha>0$ such that $\alpha<\frac{1}{24}\alpha'/d$ and $\alpha \sqrt{\go t+t}/\log(\go +t)<1$ holds for every pair of integers $(t,\go )$ that satisfy $\bar{t}>t\geq (\go +1)$. Observe that $\alpha$ exists as there are only finitely many pairs $(t,\go )$ with this property. Now suppose that there exists a pair of positive integers $(t,\go )$ with $t \geq (\go +1)$ for which there exists an algorithm that solves all $\gcut$ instances with orientable genus at most $\go$ and at most $t$ terminals in time $O(n^{\alpha \sqrt{\go t+t}\log{\go +t}})$. If $\bar{t}>t$, we obtain that these instances can be solved in $n^{c}$ for some $c<1$. This contradicts the unconditional lower bound for $\gcut$ which stems from the fact that every algorithm needs to read the input before correctly deciding the instance. If $t \geq \bar{t}$, the above argument shows that the $\ETH$ fails. 
\end{proof}

We are now ready to conclude \cref{negmain}.
\begin{proof}[Proof of \cref{negmain}]
Let $\alpha_{\textup{3T}}$ be the universal constant from \cref{prop:3TCut}, and let $\alpha_{\textup{3GT}}$ be the universal constant from \cref{negnew}.
We choose $\alpha$ with $\alpha<\min\{\alpha_{\textup{3GT}},\frac{1}{3}\alpha_{\textup{3T}}\}$. Suppose that there are integers $\go \geq 0$ and $t \geq 3$ for which there exists an algorithm that decides all the $\gcut$ instances with orientable genus $\go $ and at most $t$ terminals, in time $\bigO(n^{\alpha\sqrt{\go^2+\go t+t}/\log{(\go +t)}})$.

First suppose that $t \geq (\go +1)$. Then, as $\alpha<\alpha_{\textup{3GT}}$, we obtain from Proposition \ref{negnew} that $\ETH$ fails. We may hence suppose that $t \leq (\go +1)$. As $t \geq 3$, we have $\go \neq 0$ and hence $t \leq 2 \go $. By the assumption, in particular, all $\gcut$ instances with orientable genus $\go $ and a single terminal per group can be solved in time $\bigO(n^{\alpha\sqrt{\go ^2+\go t+t}/\log{(\go +t)}})$. As $\alpha<\frac{1}{3}\alpha_{\textup{3T}}$, we have 
\begin{align*}
	\alpha\sqrt{\go ^2+\go t+t}/\log{(\go +t)}&\leq \alpha\sqrt{\go ^2+2\go ^2+2\go }/\log{(\go +2)}\\
	&\leq \alpha\cdot 3\go /\log{(\go +2)}<\alpha_{\textup{3T}} \go /\log{(\go +2)}.
\end{align*}

We hence obtain that $\ETH$ fails by \cref{prop:3TCut}. 
\end{proof}

      \subsubsection{A  note on planar graphs}

      For planar graphs (i.e., for $\go=0$), \cref{thm:mc-tight} and \cref{thm:3groupCut} rule out algorithms with running time $O(n^{\alpha\sqrt{t}/\log t})$. This does not exactly match the $n^{O(\sqrt{t})}$ time algorithms: there is a logarithmic gap in the exponent. For planar graphs, the lower bounds can be improved to avoid this logarithmic gap and obtain tight bounds. Unfortunately, some of results that are required for this improvement are not stated in the literature in the form we need. For example, Marx~\cite{Marx12} showed that, assuming $\ETH$, there is no $f(t)n^{o(\sqrt{t})}$ algorithm for \mwc\ on planar graphs, but no statement for fixed $t$ (similar to e.g., \cref{thm:mc-tight}) appears in the literature. We briefly discuss how such statements can be obtained and then show how to improve \cref{thm:mc-tight} and \cref{thm:3groupCut} in case of planar graphs.

Chen et al.~\cite{DBLP:journals/jcss/ChenHKX06} showed that, assuming $\ETH$, there is no $f(k)n^{o(k)}$ algorithm for $k$-$\clique$. 
Marx~\cite{Marx12} uses this result to first show that $k\times k$ $\gt$ has no $f(k)n^{o(k)}$ algorithm and then, with a further reduction, shows that \mwc\ has no $f(t)n^{o(\sqrt{t})}$ algorithm. We can start this chain of reductions with the following stronger lower bound for $k$-$\clique$:

\begin{thm}[\cite{DBLP:journals/corr/abs-2311-08988}]\label{thm:cliquestronger}
  Assuming $\ETH$, there is a universal constant $\alpha>0$ such that no algorithm solves $k$-$\clique$ in time $O(n^{\alpha k})$.
\end{thm}  
The reduction presented by Marx~\cite{Marx12} immediately implies lower bounds of this form for $k\times k$ $\gt$ and \mwc.

\begin{thm}\label{thm:gridtight}
  Assuming $\ETH$, there is a universal constant $\alpha>0$ such that for any fixed choice of $k\ge 1$, there is no $O(n^{\alpha \sqrt{k}})$ algorithm for $k\times k$ $\gt$.
  \end{thm}

\begin{thm}\label{thm:mwctight}
	Assuming $\ETH$, there is a universal constant $\alpha>0$ such that for any fixed choice of $t\ge 3$, there is no $O(n^{\alpha \sqrt{t}})$ algorithm for unweighted $\mwc$, even when restricted to planar instances with at most $t$ terminals.
\end{thm}

The lower bound for $k\times k$ $\gt$ implies a lower bound for $\csp$ instances where the primal graph is a $k\times k$ grid.
\begin{thm}\label{gridhard}
There is a universal constant $\alpha$ such that for every nonnegative integer $t$, unless $\ETH$ fails, there is no algorithm deciding the $\csp$ instances $(V,D,K)$ of $\csp$ whose primal graph is the grid on $t^2$ vertices and that runs in $O(|D|^{\alpha t})$. 
\end{thm}

With a slight modification, we can make the primal graph in \cref{gridhard} 4-regular. 
\begin{thm}\label{csphard3}
There is a universal constant $\alpha$ such that for every nonnegative integer $s$ with $s \geq 108$, there exists a 4-regular planar graph $P$ with $|V(P)|\leq \frac{1}{12}s$ and the property that, unless $\ETH$ fails, there is no algorithm deciding the instances $(V,D,K)$ of $\csp$ whose primal graph is $P$ and that runs in $O(|D|^{\alpha \sqrt{s}})$. 
\end{thm}
\begin{proof}
Let $\alpha_{\textup{grid}}$ be the constant from \cref{gridhard}.
Let $\alpha =\frac{1}{24}\alpha_{\textup{grid}}$. Let $s \geq 108$ be an integer and let $s'$ be the largest integer such $s\geq 12(s')^2$. As $s \geq 108$, we have $s' \geq 3$ and hence $s \leq 12(s'+1)^2=12(s')^2+24s'+12\leq 24 (s')^2$. Let $P'$ be the grid on $(s')^2$ vertices and let $P$ be a graph obtained from $P'$ by, while $P'$ contains a vertex of degree at most 3, repeatedly adding an edge linking a vertex of minimum degree with a closest vertex of degree at most 3 when following the outer face of $P'$ in the canonical embedding of $P'$. Observe that $P$ is planar, 4-regular and $|V(P)|\leq \frac{1}{12}s$ holds. Suppose that there exists an algorithm that decides all instances whose primal graph is $P$ and that runs in $O(|D|^{\alpha s})$.

Let $(V,D,K)$ be a binary instance of $\csp$ whose primal graph is $P'$. We now define a new binary instance $(V,D,K')$ of $\csp$ by duplicating all constraints of $K$ which are associated with an edge of $P'$ which is doubled when obtaining $P$. Observe that the primal graph of $(V,D,K')$ is $P$ and that $(V,D,K)$ and $(V,D,K')$ are equivalent. By the assumption, we can solve $(V,D,K')$ in $O(|D|^{\alpha \sqrt{s}})$. We have $\alpha \sqrt{s}\leq \alpha \sqrt{24 (s')^2}\leq 24 \alpha s'=\alpha_{\textup{grid}}s'$. Hence by Proposition \ref{gridhard}, we obtain that the $\ETH$ fails.
\end{proof}

Now we can prove a tight lower bound for \gcut\ on planar graphs: the proof is the same as in \cref{thm:3groupCut}, but instead of using the $\csp$ lower bound of \cref{thm:CSP-genus-lower}, we use \cref{csphard3}.

\begin{thm}\label{planehard}
	Assuming $\ETH$, there exists a universal constant $\alpha>0$ such that for any $t\ge 3$, there is no algorithm that decides all unweighted $\gcut$ instances $(G,H,\lambda)$ for which $G$ is planar and $\abs{V(H)}\le t$, in time $\bigO(n^{\alpha\sqrt{t}})$.
      \end{thm}

As we have lower bounds for both \mwc\ (\cref{thm:mwctight}) and \gcut\ (\cref{planehard}), we can conclude the following slight strengthening of \ref{thm:maingenus}(3b) for planar graphs, removing a logarithmic factor from the exponent.   
      
\planarstronger*
\begin{proof}
	Let $\alpha_1$ and $\alpha_2$ be the universal constants from \cref{thm:mwctight,planehard}, respectively.
	Since $\calH$ is both projection-closed and has unbounded distance to extended bicliques, by \cref{combmainclass}, it contains all cliques or all complete tripartite graphs. Consequently, for $\alpha\le \min\{\alpha_1,\alpha_2\}$, the result follows from \cref{thm:mwctight} or from \cref{planehard}, respectively.
\end{proof}

We would further like to point out that similar techniques allow to prove \cref{thm:3groupCutw1} which we restate here for convenience.

\threegroupCutw*

Indeed, it follows from \cite{Marx12} that $\gt$ is $W[1]$-hard when parameterized by $|V|$. We hence obtain that $CSP$ is $W[1]$-hard when parameterized by $|V|$ for instances whose primal graph is a squared grid. Following the proof of \cref{csphard3}, we obtain that $CSP$ is $W[1]$-hard when parameterized by $|V|$ for instances whose primal graph is 4-regular and planar. The statement now follows from \cref{lemprinc}.


}

\bibliographystyle{plainurl}
\bibliography{refs}

\longversion{
\appendix

\section{Bounded Treewidth Multicut Duals Enable Algorithms}\label{algtw}
This appendix is dedicated to proving \cref{algogen} and its generalization \cref{algogenext}, which we restate below for reference.
Throughout this whole section, we fix a surface $N$ of Euler genus $g$. Recall that an instance $(G,H)$ of $\mc$ is {\em $N$-embedded} if $G$ is given in the form of a graph cellularly embedded in $N$. 
An \emph{extended} $N$-embedded instance $(G,H,F^*)$ of \mc\ is an $N$-embedded instance $(G,H)$ of \mc\ together with a set $F^*$ of faces of $G$ that is given with the input. We always set $\kappa=\abs{F^*}$. Let $C-F^*$ be the graph obtained from $C$ by removing every vertex that is in a face of $F^*$ and removing every edge that intersects a face of $F^*$.

\algogenext*

The appendix is organized as follows.
In~\cref{sec:algopre}, we give a geometric setting in which our algorithm works. In \cref{struccdv}, we review a result from \cite{ECDV} that proves the existence of multicut duals with some extra properties and extend this result by showing that there exist multicut duals with even more extra properties. In \cref{subr}, we describe a more discrete equivalent of multicut duals which will be useful for algorithmic purposes. In \cref{subr2}, we give a collection of subroutines for our algorithm. Finally, in \cref{subr3}, we give the proof of \cref{algogenext}.

\subsection{Setups}\label{sec:algopre}

We now give a collection of definitions and observations that allow to establish a geometrical setup in which our algorithm can be executed.

Let $G_1$ and $G_2$ be two embedded graphs that may have common vertices but are otherwise in general position. The \emph{overlay} $G_1 \cup G_2$ of $G_1$ and $G_2$ is the embedded multigraph with
\begin{itemize}
	\item a vertex for every vertex in $V(G_1)\cup V(G_2)$,
	\item a vertex for every intersection point of an edge in $E(G_1)$ with an edge in $E(G_2)$,
	\item an edge between each two vertices that are connected by a segment of an edge of either $G_1$ or $G_2$ that does not contain any other vertices of the overlay.
\end{itemize}

\begin{defn}[Setup]\label{def:setup}
Let $(G,H)$ be an $N$-embedded instance of $\mc$. Consider an embedded graph $K$ with $V(H)\subseteq V(K)$ that is otherwise in general position with $G$ and has the property that each of its faces forms an open disk.
We refer to the faces of $G \cup K$ as \emph{parcels} of $(G,K)$, and we denote this set by $P(G,K)$.
We say that $(G,K)$ is a \emph{setup} for $(G,H)$ if the following conditions are satisfied:
\begin{itemize}
	\item The image of every vertex of $H$ is the same in $G$ and in $K$.
	\item Every edge $e=uv$ of $K$  minimizes the total weight of the multiset of edges of $G$ that are crossed by $e$ among all polygon trails linking $u$ and $v$,
	\item data structures encoding edge-parcel incidence and parcel adjacency in $G \cup K$ can be computed in time polynomial in $n$.
\end{itemize}
\end{defn}

   Furthermore, given an orientation $\vec{C}$ of a graph $C$ embedded in $N$ in general position with $G \cup K$, for every $a=uv \in A(\vec{C})$, the \emph{crossing sequence} of $a$ is the sequence of edges of $K$ that are crossed when following $a$ from $u$ to $v$.
 Roughly speaking, the importance of cut graphs and and setups comes from their relationship to multicut duals for $(G,H)$. We will show in \cref{subr} that, in order to determine whether an embedded graph $C$ is a multicut dual for $(G,H)$, it suffices to know the parcel every of its vertices is contained in rather than their absolute position and the crossing sequence of every edge rather than its complete geometric embedding. This allows to solve a discrete problem rather than a geometric, continuous one. It is crucial in this context that the crossing sequences are defined with respect to $K$ and the size of $K$ does not depend on $G$. This allows to bound the length of the crossing sequences. We next give some more properties of setups.
\begin{prop}[{\cite[Proposition 4.1]{ECDV}}]\label{findcutgraph}
There is a constant  $\alpha \geq 0$ such that, for each $N$-embedded instance $(G,H)$ of $\mc$, a setup $(G,K)$ with $|V(K)|+|E(K)|\leq \alpha(t+g)$ can be computed in time $O((t+g) n \log n)$.
\end{prop}

From now on, we say that a setup $(G,K)$ is \emph{compact} if it respects the size bound $|V(K)|+|E(K)|\leq \alpha(t+g)$ from \cref{findcutgraph}.

\begin{prop}\label{psize}
	For every $N$-embedded instance $(G,H)$ of $\mc$ and compact setup $(G,K)$ for $(G,H)$, we have $|P(G,K)|=O((n+g)(t+g))$.
\end{prop}
\begin{proof}

 As mentioned in \cite{mt}, the fact that $G$ is embedded in $N$ implies $|E(G)|=O(n+g)$.
 Observe that, as the embedding of every edge of $K$ minimizes the total weight of edges of $G$ it crosses, each pair of embedded edges from $G$ and $K$ intersects at most once.  By construction, this yields $|V(G \cup K)|\leq |V(G)|+|V(K)|+|E(G)||E(K)|=O((n+g)(t+g))$. Observe that every edge in $E(G \cup K)$ is either in $E(G) \cup E(K)$ or is incident to some vertex $v$ in $V(G \cup K)\setminus(V(G)\cup V(K))$. Further, by construction, we have $d_{G \cup K}(v)=4$ for such a vertex. This yields $|E(G \cup K)|\leq |E(G)|+|E(K)|+4|V(G \cup K)|=O((n+g)(t+g))$.
 Euler's formula~\cite{mt} yields the desired $|P(G,K)|=|E(G \cup K)|-|V(G \cup K)|+2-2g=O((n+g)(t+g))$.
\end{proof}
Given an extended $N$-embedded instance $(G,H,F^*)$ of \mc\ and a setup $(G,K)$ for $(G,H)$, we say that a parcel in $P(G,K)$ is \emph{special} if it is contained in a face of $F^*$ and we use $P^*(G,K,F^*)$ for the set of special parcels. Similarly to Proposition \ref{psize}, we prove the following result bounding the number of special parcels.

\begin{prop}\label{pspsize}
	For every extended $N$-embedded instance $(G,H,F^*)$ of $\mc$ and every compact setup $(G,K)$ for $(G,H)$, we have $|P^*(G,K,F^*)|=f(t,g,\kappa)$.
\end{prop}
\begin{proof}

 Let $f \in F^*$. If every edge in $E(K)$ is disjoint from the interior of $f$, then, clearly, we obtain that the only parcel of $P^*(G,K,F^*)$ contained in $f$ is $f$ itself. Otherwise, observe that every parcel of $P^*(G,K,F^*)$ contained in $f$ is incident to an edge of $E(K)$. Further, as the embedding of every edge of $K$ is chosen so to minimize the weight of edges of $G$ it crosses, such an edge intersects the boundary of $f$ (which is a face of $G$) at most twice. Hence, each edge of $K$ is incident to at most two parcels of $P^*(G,K,F^*)$ contained in $f$. Therefore, the number of parcels of $P^*(G,K,F^*)$ contained in $f$ is at most $2|E(K)|$. Summarizing, $|P^*(G,K,F^*)|=2|E(K)||F^*|=f(t,g,\kappa)$.
\end{proof}

\subsection{Multicut Duals with Additional Properties}\label{struccdv}

In order to prove the main result of this section, we rely on some theory developed by Colin de Verdière in \cite{ECDV}. The following result can be obtained by combining \cite[Lemma 5.1]{ECDV} and \cite[Lemma 5.2]{ECDV}.
\begin{lem}\label{cdv2}
 Let $(G,H)$ be an $N$-embedded instance of $\mc$ and let $(G,K)$ be a compact setup for $(G,H)$. Then there exists an embedded graph $C$ that satisfies the following conditions:
\begin{enumerate}
\item $C$ is in general position with $G \cup K$,
\item $C$ is a minimum weight multicut dual for $(G,H)$,
\item for every orientation $\vec{C}$ of $C$, the crossing sequence of every $a \in A(\vec{C})$ with respect to $K$ is of length $f(g,t)$,
\item $|V(C)|+|E(C)|=f(g,t)$.
\end{enumerate}
\end{lem}

In order to use Lemma \ref{cdv2}, we need to strengthen it in two ways. First, we need to make sure that $C$ satisfies a certain condition that restricts the intersection of any edge of $C$ with any parcel of $P(G,K)$. Formally, let $C$ be a graph that is embedded in $N$ and in general position with $G\cup K$. We point out that we now actually make use of the definition of general position and use the fact that all edges with intersection points are actually crossing. We say that $C$ is \emph{well-behaved} if for every parcel $p \in P(G,K)$, one of the following holds
\begin{itemize}
	\item exactly one edge of $C$ crosses the boundary of $p$, and this edge crosses the boundary exactly twice,
	\item every edge of $C$ crossing the boundary of $p$ crosses this boundary exactly once.
\end{itemize} 
Note that in the last case, using the fact the $C$ is in general position with $G\cup K$, it follows that every such edge is incident to a vertex of $C$ inside the parcel $p$.

\begin{prop}\label{zfvtgzh}
Let $C$ be an $N$-embedded graph such that for every component of $C$, either every vertex of the component is of degree at least 3 or the component consists of a single vertex with a loop. Then $|V(C)|\leq 2|F(C)|+4g$.
\end{prop}
\begin{proof}

Suppose for the sake of a contradiction that $C$ is a counterexample to that statement. 

First suppose that $C$ contains a loop $e$ incident to a vertex $v$. We first consider the case that $e$ is incident to two distinct faces in $C$. Let $C'=C-v$. Observe that $C'$ has a face that corresponds to the two faces of $C$ incident to $e$. Clearly, $C'$ is embedded in $N$ and thus, as $C'$ is smaller than $C$, we obtain $|V(C)|=|V(C')|+1\leq 2|F(C')|+4g+1<2|F(C)|+4g$, a contradiction to $C$ being a counterexample.

Now suppose that $e$ is incident to only one face of $C$, so in particular $e$ is not surface-separating. Let a surface $N'$ be obtained from $e$ by cutting along $e$ and capping the two holes. It follows from \cite[Lemma B.4]{10.5555/3134208} that the genus of $N'$ is at most $g-1$. Further let $C'$ be the embedding of $C-v$ in $N'$ inherited from $C$. Observe that there is natural bijection between the faces  of $C'$ and the faces of $C$. As $C'$ is smaller than $C$ and $C'$ is embedded in $N'$, we obtain that $|V(C)|=|V(C')|+1\leq 2|F(C')|+4(g-1)+1<2|F(C)|+4g$, a contradiction to $C$ being a counterexample.

Now suppose that every vertex of $C$ is of degree at least 3. We obtain a graph $C'$ from $C$ by adding a new set of edges that is minimal with the property that $C'$ is connected. Observe that, by the minimality of this edge set, there is a natural bijection between the faces of $C$ and the faces of $C'$. Now Euler's formula and $d_{C'}(v)\geq d_C(v)\geq 3$ for all $v \in V(C')$ yield 

\begin{align*}
|V(C)|+|F(C)|&=|V(C')|+|F(C')|\\
&=|E(C')|+2-2g\\
&\geq \frac{3}{2}|V(C')|+2-2g\\
&=\frac{3}{2}|V(C)|+2-2g,
\end{align*}

yielding $|V(C)|\leq 2|F(C)|+4g$, again a contradiction.
\end{proof}

We are now ready to give the following modification of \cref{cdv2} that shows that we can ensure that $C$ is both well-behaved and subcubic by increasing the size of $C$ only insignificantly. Note that the properties stated in \cref{item:cdv2neu1,item:cdv2neu2,item:cdv2neu3,item:cdv2neu4} are identical in \cref{cdv2,cdv2neu}.

\begin{lem}\label{cdv2neu}\label{cdv2ganzneu}
 Let $(G,H)$ be an $N$-embedded instance of $\mc$ and let $(G,K)$ be a compact setup for $(G,H)$. Then there exists a subcubic, well-behaved embedded graph $C$ that satisfies 
 \begin{enumerate}
 	\item $C$ is in general position with $G \cup K$, \label{item:cdv2neu1}
 	\item $C$ is a minimum weight multicut dual for $(G,H)$, \label{item:cdv2neu2}
 	\item for every orientation $\vec{C}$ of $C$, the crossing sequence of every $a \in A(\vec{C})$ with respect to $K$ is of length $f(g,t)$, \label{item:cdv2neu3}
 	\item $|V(C)|+|E(C)|=f(g,t)$. \label{item:cdv2neu4}
 \end{enumerate}
\end{lem}
\begin{proof}
	By \cref{cdv2}, there is a multicut dual $C$ for $(G,H)$ that satisfies \cref{item:cdv2neu1,item:cdv2neu2,item:cdv2neu3,item:cdv2neu4}.
	It remains to show that we can modify this graph so that it is also well-behaved and subcubic without losing any of the desired properties.
	
	\subparagraph*{\boldmath Part 1: Making $C$ well-behaved}
	We now modify $C$ to obtain a slightly larger multicut dual $C'$ that is well-behaved.
For every $p \in P(G,K)$, we introduce a new vertex $x_p$ inside $p$. Then, every segment of an edge of $C$ linking two intersection points $q_1,q_2$ with the boundary of $p$ is replaced by segments linking $q_1$ and $x_p$ and $q_2$ and $x_p$, respectively, in a way that these new segments are fully contained in $p$. See\cref{fig:cdv2neuwellbehaved}. 
\begin{figure}[t]
	\centering
	\includegraphics[scale=0.45]{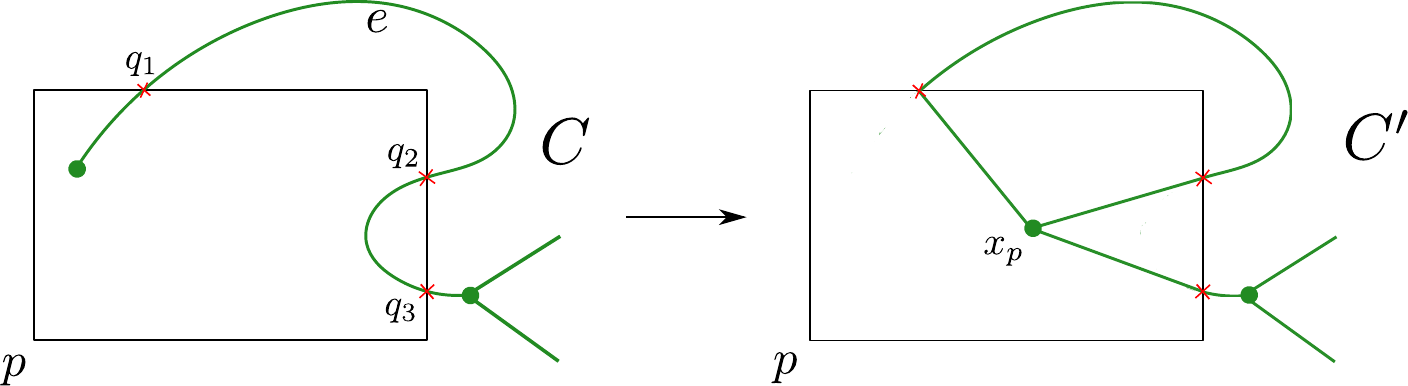}
	\caption{Illustration of the modification of a multicut dual $C$ to make it well-behaved used in the proof of \cref{cdv2ganzneu}.}
	\label{fig:cdv2neuwellbehaved}
\end{figure}

We also replace every edge segment linking a point $q$ on the boundary of $p$ to an endvertex of the edge inside $p$ by a segment from $q$ to $x_p$ fully contained in $p$, thus replacing the previous endpoint of the edge by $x_p$. We choose these edge segments in a way that the obtained graph is embedded in $N$ and we denote this embedded graph by $C''$. It follows by construction and the fact that $C$ is a multicut dual for $(G,H)$ that $C''$ is a multicut dual for $(G,H)$. Hence there exists a subgraph $C'''$ of $C''$ that is an inclusion-wise minimal multicut dual for $(G,H)$.

Observe that, by the minimality of $C'''$, we have that $d_{C'''}(v)\geq 2$ for every $v \in V(C''')$. We now obtain $C'$ from $C'''$ in the following way: in every component of $C'''$ that contains a vertex of degree at least $3$, we omit all vertices of degree $2$. Observe that the order in which these vertices are omitted is irrelevant. Further, in every component of $C'''$ all of whose vertices are of degree 2, we choose one vertex of the component and omit all other vertices. Again, the order of the omissions is irrelevant. 
We will show that $C'$ has the desired properties.

It follows by construction that $C'$ is well-behaved. As $C$ satisfies \cref{item:cdv2neu1} and by construction, so does $C'$. Moreover, as every edge of $G$ that is crossed by at least one edge of $C'$ is also crossed by some edge of $C$ and as $C$ satisfies \cref{item:cdv2neu2}, the cut $C'$ satisfies \cref{item:cdv2neu2} as well. Now consider an orientation $\vec{C'}$ of $C'$, $a' \in A(\vec{C'})$ and let $e'$ be the edge of $C'$ corresponding to $a'$. Then there exists a collection of edges $S \in E(C)$ and an orientation $\vec{C}$ of $C$ such that the crossing sequence of $a'$ is a concatenation of subsequences of the crossing sequences of the arcs of $A(\vec{C})$ associated with $S$. Hence, as $C$ satisfies \cref{item:cdv2neu3} and \cref{item:cdv2neu4}, we obtain that $C'$ satisfies \cref{item:cdv2neu3}. By \cref{faceterminal}, we obtain that every face of $C'$ contains a terminal of $V(H)$, which implies that $C'$ has at most $t$ faces. Hence, by construction, \cref{zfvtgzh}, and Euler's formula, we obtain that \cref{item:cdv2neu4} holds.

\subparagraph*{\boldmath Part 2: Making $C'$ subcubic}

It is straightforward that \cref{item:cdv2neu1,item:cdv2neu2,item:cdv2neu3,item:cdv2neu4} are maintained when subdividing every loop once. Moreover, this operation maintains the property of being well-behaved. We may hence suppose that $C'$ does not contain loops. We now define a new graph $C''$ embedded in $N$. Namely, for every edge $e=uv \in E(C')$, we let $V(C'')$ contain two vertices $x_u^{e}$ and $x_v^{e}$ and we let $E(C'')$ contain $x_u^{e}x_v^{e}$. For each $v \in V(C')$, let $e_1,\ldots,e_{d_{C'}(v)}$ be the incident edges of $C'$ in the order they appear when going clockwise around $v$, starting with an arbitrary edge. For $i \in [d_{C'}(v)-1]$, the edge $x_v^{e_i}x_v^{e_{i+1}}$ is added to $E(C'')$, see \cref{fig:cdv2neusubcubic} 
\begin{figure}[t]
	\centering
	\includegraphics[scale=0.45]{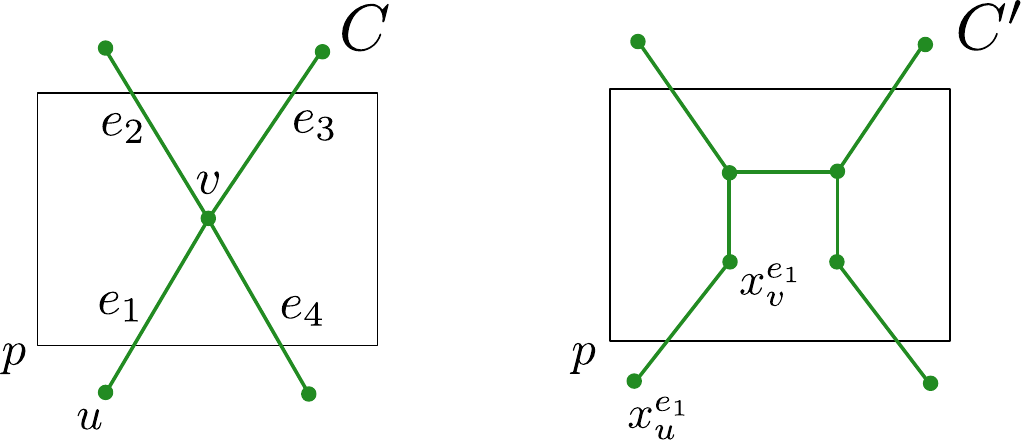}
	\caption{Illustration of the modification of a well-behaved multicut dual $C$ to make it subcubic used in the proof of \cref{cdv2ganzneu}.}
	\label{fig:cdv2neusubcubic}
\end{figure}

By construction, $C''$ can be embedded in $N$. Moreover, this embedding can be chosen in a way that $C''$ is in general position with $G \cup K$ and, for every $v \in V(C')$ and every $e \in E(C')$ incident to $v$, we have that $x_v^{e}$ in $C''$ is embedded in the same parcel of $P(G,K)$ as $v$. Moreover, for any $v \in V(C')$ and any two edges $e,e' \in E(C')$ with $x_v^{e}x_v^{e'}\in E(C'')$, we choose the embedding of $C''$ such that the edge $x_v^{e}x_v^{e'}$ is fully contained in the same parcel of $P(G,K)$ as $v$. This completes the definition of $C''$.

By construction, we have that $C''$ is subcubic. Next, observe that every $e \in E(C'')$ either corresponds to an edge of $C'$ or is fully contained in a parcel of $P(G,K)$. Hence, $C''$ is also well-behaved.
By construction, $C''$ satisfies \cref{item:cdv2neu1}. Note that there is a direct correspondence between the faces of $C'$ and the faces of $C'$ yielding that two vertices of $G$ are in the same face of $C''$ if and only if they are in the same face of $C'$. Hence, as $e_G(C'')=e_G(C')$ and $C'$ is a minimum weight multicut dual for $(G,H)$, we obtain that $C''$ satisfies \cref{item:cdv2neu2}.
Now let $\vec{C''}$ be an orientation of $C''$ and let $\vec{C'}$ be an orientation of $C'$. Then for every $a' \in A(\vec{C''})$, there either exists a corresponding arc $a \in A(\vec{C'})$ or $a'$ is fully contained in a parcel of $P(G,K)$. In the latter case, the crossing sequence of $a'$ has length $0$. Hence, $C''$ also satisfies \cref{item:cdv2neu3}.

Finally, observe that $|V(C'')|=2|E(C')|$ and, as $C''$ is subcubic, we have $|E(C'')|\leq \frac{3}{2}|V(C'')|=3|E(C')|$. It follows that $C''$ satisfies \cref{item:cdv2neu4}. 
\end{proof}

\subsection{A More Discrete Setting}\label{subr}
In this section, we describe a discrete equivalent of an embedding of a graph $C$ in $N$. It takes into account its position relative to a setup $(G,K)$ rather than its absolute position. We require this setting for algorithmic purposes later on.

Let $(G,K)$ be a setup for an $N$-embedded instance $(G,H)$ of $\mc$. For $p_1,p_2 \in P(G,K)$, a $p_1p_2$-trail $R$ of length $k$ is an alternating sequence $p^1e^1p^2e^2\ldots p^{k+1}$ where $p^1,\ldots,p^{k+1}\in P(G,K)$ and $e^1,\ldots,e^{k}\in E(G)\cup E(K)$ such that $p^1=p_1,p^{k+1}=p_2$ and for every $i \in [k]$, we have that $p^i$ and $p^{i+1}$ are separated by a nontrivial segment of $e^{i}$.

Let $E(R)$ be the \emph{multi}set of edges $e^1,\ldots, e^{k-1}$ on the trail $R$. Next, let $e_G(R)=E(R)\cap E(G)$ and let $e_K(R)=E(R)\cap E(K)$. The \emph{crossing sequence} of $R$ refers to the sequence of edges of $e_K(R)$ in the order that they appear in $R$. We set $w(R)=\sum_{e \in e_G(R)}w(e)$ where $w$ is the weight function associated with $G$.
For a set of trails $\calR$, we set $w(\mathcal{R})=\sum_{R \in \mathcal{R}}w(R)$ and $e_G(\mathcal{R})=\bigcup_{R \in \mathcal{R}}e_G(R)$.

Given a setup $(G,K)$ and a multigraph $C$ (which may also have loops), consider
\begin{enumerate}
	\item an orientation $\vec{C}$ of $C$,
	\item a mapping $\phi:V(C)\rightarrow P(G,K)$, and 
	\item a collection $\mathcal{R}$ that consists of a $\phi(u)\phi(v)$-trail $R_a$ for each $a=uv \in A(\vec{C})$.
\end{enumerate}
We say that $(\vec{C}, \phi, \calR)$ is a \emph{generalized discretized embedding} of $C$. If we use it in the context of embedded graphs, then we mean the generalized discretized embedding of the corresponding abstract graph.

\begin{lem}\label{dualdiscrete}
Let $(G,K)$ be a setup for an $N$-embedded $\mc$ instance $(G,H)$. Let $C$ be a  multicut dual for $(G,H)$ in general position with $G \cup K$.
Then there exists a generalized discretized embedding $(\vec{C},\phi,\mathcal{R}=(R_a:a \in A(\vec{C})))$ of $C$ such that $e_G((\vec{C},\phi,\mathcal{R}))=e_G(C)$ and for every $a \in A(\vec{C})$, the crossing sequence of $a$ is the same as the crossing sequence of $R_a$. Moreover, for every $v \in V(C)$, it holds that $\phi(v)$ is contained in the same face of $G$ as $v$.
\end{lem}
\begin{proof}
First let $\vec{C}$ be an arbitrary orientation of $C$. Next observe that, as $C$ is in general position with $G$ and $K$, for every $v \in V(C)$, there exists a parcel $p_v \in P(G,K)$ such that $v$ is contained in $p_v$. We may hence define $\phi(v)=p_v$ for all $v \in V(C)$. Now consider some $a=uv \in A(\vec{C})$. Let $p^1=p_u,p^2,\ldots,p^q=p_v$  be the parcels in $P(G,K)$ which are traversed by $a$ in this order when going from $u$ to $v$. For $i \in [q-1]$, let $e^{i}$ be the edge of $G\cup K$ that is crossed by $a$ when going from $p^{i}$ to $p^{i+1}$. Let $R_a=(p^1,e^1,p^2, e^2, \ldots,p^q)$. Clearly, $R_a$ is a $\phi(u)\phi(v)$-trail with the same crossing sequence as $a$. It follows that $(\vec{C},\phi,\mathcal{R})$ has the desired properties.
\end{proof}

The following result links embedded graphs and generalized discretized embeddings.  Its proof is almost identical to the one of~\cite[Lemma 7.3]{ECDV}, where the statement is proved in a slightly less general form. It heavily relies on homology techniques described in \cite{cen}, in which basic notions of homology are introduced as well. However, our proof can be read without any prior knowledge of homology.

\begin{lem}\label{srdztfgzh}
	Let $(G,K)$ be a setup for an $N$-embedded $\mc$ instance $(G,H)$. Let $C$ be a  multicut dual for $(G,H)$ in general position with $G \cup K$ and let $(\vec{C},\phi,\calR=(R_a:a \in A(\vec{C})))$ be a generalized discretized embedding of $C$. If for each $a\in A(\vec{C})$, the crossing sequences of $a$ and $R_a$ are the same, then $e_G(\calR)$ is a multicut for $(G,H)$.
\end{lem}
\begin{proof}
	Let $t_1,t_2\in V(H)$ with $t_1t_2 \in E(H)$.
	As $C$ is a multicut dual, $t_1$ and $t_2$ are contained in distinct faces of $C$. Let $E_1$ be the set of those edges of $C$ that are
	incident to the face containing $t_1$ and also to some other face. Let $A_1\subseteq A(\vec{C})$ be the arcs associated to the edges in $E_1$. Let $C_1$ be the subgraph of $C$ induced by $E_1$.
	Observe that $C_1$ is an even graph.
	Let $\vec{C_1}$  be the orientation of $C_1$ inherited from $\vec{C}$.
	
		We now define a graph $C_2$ that is associated to $\bigl(\vec{C_1},\phi(V(C_1)),(R_a:a \in A(\vec{C_1}))\bigr)$. 
		The graph $C_2$ has a vertex $x_f$ for every face $f$ of $G$ and it has an edge $x_f x_{f'}$ for every trail $R \in \calR$ that links a parcel contained in the face $f\in F(G)$ with a parcel contained in $f'\in F(G)$. In order to embed such an edge, for every $e \in e_G(R)$ that separates two parcels contained in two distinct faces $f_0$, $f_0'$ of $G$, we add a polygon trail  from $x_{f_0}$ to $x_{f_0'}$ that crosses $e$ in one point and is fully contained in the interior of $f_0$ and $f_0'$, otherwise.
This finishes the description of $C_2$. Appropriately choosing the polygon trails, we may suppose that $C_2$ is embedded in $N$.
	
	As $t_1$ and $t_2$ are in different faces of $C_1$, it follows from \cite[Lemma 3.1]{cen} that $C_1$ is homologous to a small circle around $t_1$ in the surface $N'$ which is obtained from $N$ by deleting $t_1$ and $t_2$. Next, it follows from \cite[Lemma 3.4]{cen} that $C_1$ and $C_2$ are homologous in $N$ and hence also in $N'$. Hence $C_2$ is homologous to a small circle around $t_1$ in $N'$. It follows from \cite[Lemma 3.1]{cen} that $t_1$ and $t_2$ are in distinct faces of $C_2$. Hence $t_1$ and $t_2$ are in distinct components of $G\setminus e_G(\mathcal{R})$.  
\end{proof}

\subsection{Algorithmic Subroutines}\label{subr2}
In this section, we show that, given a setup $(G,K)$, a minimum weight generalized discretized embedding of a given graph with prescribed crossing sequences can be found efficiently if the treewidth of a certain subgraph of the graph is bounded. 
We first show how to efficiently compute a minimum weight trail with given endpoints and  crossing sequence.

\begin{lem}\label{trail}
Let $(G,K)$ be a compact setup for an $N$-embedded $\mc$ instance $(G,H)$. Further, let $p_1,p_2 \in P(G,K)$ and let $\sigma=e_1,\ldots,e_\gamma$ be a sequence of $\gamma$ edges of $K$. Then there is an algorithm that decides in time $f(t,g,\gamma)n^3$ whether a $p_1p_2$-trail with crossing sequence $\sigma$ exists and, if so, outputs one such trail of minimum weight.
\end{lem}
\begin{proof}
Let $w$ be the weight function associated to $G$. We proceed by a dynamic programming approach. For every $p \in P(G,K)$, every $i \in \{0,\ldots,\gamma\}$ and every $j \in \{0,\ldots,\gamma |P(G,K)|\}$, we compute a value $z[p,i,j]$ that is the minimum weight of a $p_1p$-trail whose crossing sequence is $(e_1,\ldots,e_i)$ and whose length is at most $j$ where $z[p,i,j]=\infty$ if no such trail exists. If $i=0$, we mean that the crossing sequence is empty. Observe that the length of a trail that has minimum weight among all those trails linking two specific parcels and having an empty crossing sequence is at most $|P(G,K)|-1$. We hence obtain that $z[p_2,\gamma,\gamma |P(G,K)|]$ is indeed the minimum weight of a $p_1p_2$-trail whose crossing sequence is $(e_1,\ldots,e_\gamma)$.

We now show how to compute $z[p,i,j]$ for $(p,i,j)$ in the considered domain. First, for ease of notation, we assign $z[p,-1,j]=\infty$ for all $p$ and $j$ in the considered domain. Next, we set $z[p,i,0]=0$ if $(p,i)=(p_1,0)$, and $z[p,i,0]=\infty$ otherwise. It is easy to see that this assignment is correct.

We assume from now on that $j \geq 1$ and $z$ has already been computed for all arguments whose last entry is at most $j-1$. 
Let $P_G$ and $P_K$ be the set of parcels $p'\in P(G,K)-p$ that are separated from $p$ by a segment of an edge $e_{p'}$ of $G$ or $K$, respectively. 

We set
\begin{equation*}
z[p,i,j]=\min\Bigl\{z[p,i,j-1],\min_{p' \in P_G}(z[p',i,j-1]+w(e_{p'})),\min_{p' \in P_K}z[p',i-1,j-1]\Bigr\}.
\end{equation*}

In order to see that the assignment is correct, observe that the minimum is attained for the first argument if the minimum weight trail of the desired form is of length at most $j-1$. Otherwise, the minimum is attained for an argument of the second form if the last edge of the trail is in $E(G)$, and for an argument of the third form if the last edge of the trail is in $E(K)$.

The desired value is now obtained for $z[p_2,\gamma,\gamma |P(G,K)|]$. Further, a minimum weight trail can be obtained by saving a trail of weight $z[p,i,j]$ for all $(p,i,j)$ with $z[p,i,j]\neq \infty$. By construction, the total number of values computed is $|P(G,K)|^2\gamma^2$. Further, for the computation of every value, a total of at most $|P(G,K)|$ precomputed values is compared. As $|P(G,K)|=O((n+g)(t+g))$ by \cref{psize}, the total running time is bounded by $f(t,g,\gamma) n^3$.
\end{proof}

\begin{lem}\label{algor}
	Let $(G,K)$ be a setup for an $N$-embedded $\mc$ instance $(G,H)$. Let $C$ be a graph with $|V(C)|+|E(C)|\leq \pi$ for some positive integer $\pi$, $\vec{C}$ an orientation of $C$, $X \subseteq V(C), \psi:X \rightarrow P(G,K)$ a mapping,  and $\Gamma=\{\sigma_a:a \in A(\vec{C})\}$ a collection of possibly empty sequences of edges of $K$ such that, for all $a \in A(\vec{C})$, the length of $\sigma_a$ is at most $\gamma$ for some integer $\gamma$. Also, let a tree decomposition of $C-X$ of width $\beta$ for some positive integer $\beta$ be given.

	Consider the set $\mathcal{C}$ of generalized discretized embeddings $(\vec{C},\phi,\calR)$ of $C$ with $\phi|_X=\psi$ and for which, for all $a \in A(\vec{C})$, the crossing sequence of $R_a$ is $\sigma_a$.
	Then there is an algorithm that computes in time $f(t,\gamma,\pi,g)n^{O(\beta)}$ a minimum weight member of $\mathcal{C}$, or otherwise correctly reports that $\mathcal{C}$ is empty.
\end{lem}
\begin{proof}
We formulate the task of finding the respective generalized discretized embedding as a valued binary constraint satisfaction problem. An instance of \emph{binary VCSP} is a triple $(V ,D, Z)$,
in which

\begin{itemize}
\item $V$ is a set of variables,
\item $D$ is a domain of values,
\item $Z$ is a set of constraints, each of which is a triple $\langle
u,v,c\rangle$ with $u,v\in V$ and $c: D^2\to \mathbb{Z}_{\ge 0}$.
\end{itemize}

Given an assignment $\phi:V\to D$, the \emph{value} of the assignment is $\sum_{\langle u,v,c\rangle\in K}c(\phi(u),\phi(v))$. The goal is to find an assignment with minimum value.

Finding the required generalized discretized embedding can be formulated as a VCSP as follows. The set $V$ of variables contains $V(C)-X$ and two more variables $v_0,v_1$ and the domain $D$ corresponds to the parcels $P(G,K)$. For every edge $(u,v)\in A(\vec C)$ with $\{u,v\}\cap X=\emptyset$, there is a corresponding constraint $\langle u,v,c\rangle$ in $Z$, where, for all $p_1,p_2 \in P(G,K)$, we define $c(p_1,p_2)$ to be the minimum  weight of a $p_1p_2$-trail whose crossing sequence is $\sigma_{(u,v)}$. Further, for every $v\in V(\vec C)-X$, there is a constraint $\langle v,v_0,c\rangle$ in $Z$, where, for all $p_1,p_2 \in P(G,K)$, we define $c(p_1,p_2)$ to be the sum of the minimum  weight of a $p_1\psi(x)$-trail whose crossing sequence is $\sigma_{(v,x)}$ over all $x \in X$ with $vx\in A(\vec{C})$ plus the minimum weight of a $\psi(x)p_1$-trail whose crossing sequence is $\sigma_{(x,v)}$ over all $x \in X$ with $xv\in A(\vec{C})$. Finally there is a constraint $\langle v_0,v_1,c\rangle$ in $Z$ which, for all $p_1p_2 \in P(G,K)$, is defined as $c(p_1p_2)$ being the sum of the minimum  weight of a $\psi(x)\psi(x')$-trail whose crossing sequence is $\sigma_{(x,x')}$ over all $xx'\in A(\vec{C})$ with $x,x'\in X$. Observe that for every constraint, by\cref{trail}, $c$ can be computed in $f(t,g,\gamma)n^3$. Further, as $(G,K)$ is compact, we have $|Z|=O(n^2)$ and so $(V,D,Z)$ can be computed in $f(t,g,\gamma)n^5$. By the definition of generalized discretized embeddings, the optimum value of this VCSP instance is exactly the required minimum weight.

Observe that the primal graph of $(V,D,Z)$ can be obtained from $G-X$ by adding two vertices and so its treewidth is at most $\beta+2$. Further, it is folklore that a VCSP instance whose primal graph has treewidth $O(\beta)$ can be solved in time $n^{O(\beta)}$: for example, Carbonnel et al.~\cite{DBLP:journals/siamcomp/CarbonnelRZ22} attribute this to Bertel\`e and Brioschi~\cite{bb}. 
\end{proof}

\subsection{Benefit from Bounded Treewidth}\label{subr3}
We are now ready to put the pieces together.
\begin{proof}[Proof of Theorem \ref{algogenext}]
Consider a compact setup $(G,K)$ for $(G,H)$. We first describe a certain multicut dual $C^*$ for $(G,H)$ and use \cref{cdv2ganzneu} to determine two bounds $q_1$ and $q_2$ related to the instance. These values will be used to describe the algorithm. 

\subparagraph*{\boldmath Describing $C^*$ and determining $q_1$, $q_2$}
Let $C^*_0$ be a subcubic, well-behaved minimum-weight multicut dual for $(G,H)$
as given by \cref{cdv2ganzneu}. 
As $C^*_0$ is well-behaved, recall that, for every parcel $p$ of $P(G,K)$ exactly one edge of $C^*_0$ crosses the boundary of $p$, and this edge crosses the boundary exactly twice, or otherwise every edge of $C^*_0$ crossing the boundary of $p$ crosses this boundary exactly once.

Now let $P_0$ be the set of those parcels in $P(G,K)$ that are contained in a face of $F^*$ and for which there exists exactly one edge of $C^*_0$ that crosses the boundary of $p$ exactly twice. We modify $C_0^*$ by subdividing, for every $p \in P_0$, the unique edge of $C_0^*$ that crosses the boundary of $p$, where the newly created vertex is inside $p$. We use $C^*$ to denote the resulting graph.

By \cref{item:cdv2neu3} in \cref{cdv2neu}, there is a bound $q_1$ depending on $g$ and $t$ on the length of the crossing sequences associated to the arcs of an orientation of $C^*_0$. The orientations of  $C^*$ inherit this bound.
Furthermore, using the bound on $|V(C^*_0)|+|E(C^*_0)|$ from \cref{cdv2neu} together with \cref{pspsize}, it follows that $|V(C^*)|$ is bounded by some value $q_2$ that depends only on the values of $t$, $g$, and $\kappa$. Observe that $q_1$ and $q_2$ can be efficiently computed from $(G,H)$ and do not require the availability of $C^*$. We may hence suppose for our algorithm that $q_1$ and $q_2$ are available.

\subparagraph*{Algorithm}
We are now ready to describe our algorithm. Let the algorithm maintain a set $S$, which is a multicut of $(G,H)$. This set is initialized by setting $S=E(G)$. Let $\mathcal{C}$ be the set of graphs on a fixed vertex set of size $q_2$. 
\begin{itemize}
\item Compute the set $\mathcal{C}_1$ of tuples $(C,X)$ of graphs $C$ in $\mathcal{C}$ and sets $X \subseteq V(C)$ that satisfy $\tw(C-X)\leq \beta$. 
\item Consider some $(C,X) \in \mathcal{C}_1$ and compute a tree decomposition of width $\beta$ of $C-X$. 
\item Let $\vec{C}$ be an orientation of $C$ and let $\Gamma=(\sigma_a:a \in A(\vec{C}))$ be a collection of sequences of $E(H)$ each of which has length at most $q_1$. Consider a mapping $\psi:X \rightarrow P^*(G,K,F^*)$. 
Check whether there exists a weighted generalized discretized embedding $(\vec{C},\phi,(R_a:a \in A(\vec{C})))$ of $C$ that satisfies $\phi|_X=\psi$ and has the property that, for every $a\in A(\vec{C})$, the crossing sequence of $R_a$ is precisely $\sigma_a$. 
\item If the previous check succeeds, compute one such generalized discretized embedding $(\vec{C},\phi,\mathcal{R})$ of minimum weight.
Then check whether $e_G(\vec{C},\phi,\mathcal{R})$ is a multicut for $(G,H)$. If this is the case and the weight of this multicut is smaller than the previously stored multicut $S$, replace $S$ by $e_G((\vec{C},\phi,\mathcal{R}))$. 
\item Repeat the previous three steps for every $(C,X)\in \mathcal{C}_1$, every orientation $\vec{C}$ of $C$, every collection of sequences $\Gamma=(\sigma_a:a \in A(\vec{C}))$ of $E(H)$ with length at most $q_1$, and every mapping $\psi:X \rightarrow P^*(G,K,F^*)$.
\item Output the assignment of $S$ after the last iteration. 
\end{itemize}
This finishes the description of the algorithm.

\subparagraph*{Correctness}
For the graph $C^*$ described previously, let $V_{F^*}^{C^*}$ denote the set of vertices of $C^*$ whose position is contained in a face of $F^*$.
As every subgraph of a subcubic, well-behaved graph is again subcubic and well-behaved, we may suppose that $C_0^*$ is an inclusion-wise minimal multicut dual. By assumption and construction, we obtain $\tw(C^*-V_{F^*}^{C^*})=\tw(C_0^*-F^*)\leq \beta$. By definition, we obtain that $(C^*,V_{F^*}^{C^*})\in \mathcal{C}_1$. Next, let $\psi^*:V_{F^*}^{C^*} \rightarrow P^*(G,K,F^*)$ be the mapping where every $v \in V_{F^*}^{C^*}$ is mapped to the parcel it is contained in in $C^*$, let $\vec{C^*}$ be an orientation of $C^*$ and let $(\sigma^*_a:a \in A(\vec{C^*}))$ be the crossing sequences of the arcs of $\vec{C^*}$. By construction, our algorithm at some point computes a minimum weight generalized discretized embedding of $(\vec{C}^*,\phi^*,(R_a^*:a \in A(\vec{C^*})))$ such that $\phi^*|_{V_{F^*}^{C^*}}=\psi^*$ and, for every $a \in A(\vec{C^*})$, the crossing sequence of $R_a^*$ is $\sigma^*_a$. Let $S^*=e_G\bigl((\vec{C}^*,\phi^*,(R_a^*:a \in A(\vec{C^*})))\bigr)$. By \cref{srdztfgzh}, we obtain that $S^*$ is a multicut for $(G,H)$. Hence, by the minimality of $C^*$, we obtain by Lemma \ref{dualdiscrete}  that the weight of $S^*$ is the same as the weight of a minimum weight multicut dual for $(G,H)$. It follows from Lemmas \ref{dualcut} and \ref{dualcut2} that $S^*$ is a minimum weight multicut for $(G,H)$. By construction, the algorithm outputs a multicut for $(G,H)$ whose weight is not larger than the one of $S^*$. Hence the algorithm is correct.

\subparagraph*{Runtime}
It remains to analyze the runtime of the algorithm. The set $\mathcal{C}$ is of size $f(q_2)$ and can be computed in $f(q_2)$ by a brute force approach. For every $C \in \mathcal{C}$ and every $X \subseteq V(C)$, as $|V(C-X)|\leq |V(C)|\leq q_2$, using a brute force approach, we can decide in $f(q_2)$ whether $\tw(C-X)\leq \beta$, and compute a corresponding tree decomposition if this is the case. Further, given some $C \in \mathcal{C}$, as $|V(C)|\leq q_2$, there are only $f(q_2)$ choices for $X$. Hence, the size of $\mathcal{C}_1$ is bounded by $f(q_2)$ and $\mathcal{C}_1$ can be computed in time $f(q_2)$.

For some fixed $(C,X)\in \mathcal{C}_1$, observe that the number of orientations of $C$ is bounded by $2^{|E(C)|}=f(q_2)$, the number of mappings $\psi:X \rightarrow P^*(G,K,F^*)$ is bounded by $|P^*(G,K,F^*)|^{|X|}=f(t,g,\kappa,q_2)$ by \cref{pspsize}; and the number of possible choices for $(\sigma_a:a \in A(\vec{C}))$ is bounded by $|E(H)|^{q_1|E(C)|}=f(q_1,q_2,t)$. 

Given $\vec{C},\psi$ and $(\sigma_a:a \in A(\vec{C}))$, it follows from \cref{algor} that a corresponding minimum weight generalized discretized embedding $(\vec{C},\phi,(R_a:a \in A(\vec{C})))$ can be computed in $f(t,q_1,q_2,g)n^{O(\beta)}$. Finally, we can check by a collection of breadth-first searches whether the computed edge set is indeed a multicut for $(G,H)$.

As $q_1=f(t,g)$ and $q_2=f(t,g,\kappa)$, we obtain that the total running time of the algorithm is $f(t,g,\kappa)n^{O(\beta)}$, as desired.
\end{proof}


}
\end{document}